\documentclass[sigconf, nonacm]{acmart}
\usepackage[utf8]{inputenc}
\usepackage[T1]{fontenc}
\usepackage{amsmath}
\usepackage{amsfonts}
\usepackage{nicefrac}
\usepackage{microtype}
\usepackage{xcolor}
\usepackage{graphicx}
\usepackage{graphicx}
\usepackage{tikz}
\usetikzlibrary{arrows.meta}
\usepackage{hyperref}
\usepackage{url}
\usepackage{booktabs}
\usepackage{multirow}
\usepackage{makecell}
\usepackage{array}
\usepackage{longtable}
\usepackage{subcaption}
\usepackage{enumitem}
\usepackage{comment}
\usepackage{fancyvrb}
\usepackage{fvextra}
\usepackage{algorithm}
\usepackage{algpseudocodex}
\newcommand{\name}{\textsc{\textsf{Veda}}\xspace}
\newcommand{\effname}{\textsc{\textsf{EffVeda}}\xspace}

\usepackage{pifont}

\definecolor{redx}{RGB}{180,0,0}
\definecolor{greenx}{RGB}{0,180,0}

\usepackage{threeparttable}
\newtheorem{assumption}{Assumption}

\newcounter{example}[section]
\renewcommand{\theexample}{\thesection.\arabic{example}}

\newenvironment{example}[1][]%
{%
\refstepcounter{example}%
\par\smallskip\noindent\textit{Example~\theexample.}%
\if\relax\detokenize{#1}\relax\else: #1\fi\space\ignorespaces }%

\usepackage[capitalize]{cleveref}

\usepackage{xspace}

\newcommand{\vldbdoi}{XX.XX/XXX.XX}
\newcommand{\vldbpages}{XXX-XXX}
\newcommand{\vldbvolume}{14}
\newcommand{\vldbissue}{1}
\newcommand{\vldbyear}{2020}
\newcommand{\vldbauthors}{\authors}
\newcommand{\vldbtitle}{\shorttitle}
\newcommand{\vldbavailabilityurl}{URL_TO_YOUR_ARTIFACTS}
\newcommand{\vldbpagestyle}{plain}

\newcommand{\shanshan}[1]{{\color{black}#1}}
\newcommand{\change}[1]{{\color{black}#1}}

\begin{document}
    \title{Don't Stir the Pot! Authorized Vector Data Retrieval \\via Access-Aware Indexing}

  \author{Shanshan Han}
  \affiliation{%
  \institution{University of California, Irvine} 
  }
  \email{shanshan.han@uci.edu}

  \author{Vishal Chakraborty}
  \affiliation{%
  \institution{University of California, Irvine} 
  }
  \email{vchakrab@uci.edu}

  \author{Sharad Mehrotra}
  \affiliation{%
  \institution{University of California, Irvine} 
  }
  \email{sharad@ics.uci.edu}




  \begin{abstract}
    Vector databases increasingly enforce role-based access control, where each
    top-$k$ approximate nearest neighbor query must return only vectors the
    querying role is authorized to access. Two extremes bracket the design space.
    A single global index built over all vectorsavoids duplication but wastes search effort on
    unauthorized vectors and degrades recall, while an oracle index, built with all authorized vectors to the query roles, searches only 
    authorized vectors but duplicates every shared vector between roles or queries. 
    We present \name and its efficient variant \effname, two indexing strategies
    built on an \emph{access-aware lattice} to address access control in vector databases. The methods first partitions the dataset into disjoint data blocks by role combination, then leverage the structure of the access-aware lattice to apply \emph{copy} and
    \emph{merge} operations to group co-accessed blocks under a user-specified
    storage budget. Large nodes in the lattice are then indexed with HNSW, while small
    nodes are retained for linear scan. To facilite query processing on the lattice, our methods construct a query plan that selects the
    minimal set of nodes that covers all authorized data for each role. At query time,
    \emph{coordinated search} first queries pure (authorized-only) nodes to
    populate a global top-$k$ heap, then leverages the resulting distance bound of the $k^\mathit{th}$ data in the heap to prune
    exploration on impure nodes. 
    Evaluations show that
    our methods deliver higher throughput at high recall while closely tracking the storage budget.
  \end{abstract}

  \maketitle

  \pagestyle{\vldbpagestyle}
  \begingroup\small
  \noindent
  \raggedright\textbf{PVLDB Reference Format:}\\
  \vldbauthors. \vldbtitle. PVLDB, \vldbvolume(\vldbissue): \vldbpages,
  \vldbyear.\\
  \href{https://doi.org/\vldbdoi}{doi:\vldbdoi} \endgroup \begingroup
  \renewcommand{\thefootnote}{}
  \footnote{\noindent
  This work is licensed under the Creative Commons BY-NC-ND 4.0 International
  License. Visit \url{https://creativecommons.org/licenses/by-nc-nd/4.0/} to view
  a copy of this license. For any use beyond those covered by this license, obtain
  permission by emailing \href{mailto:info@vldb.org}{info@vldb.org}. Copyright
  is held by the owner/author(s). Publication rights licensed to the VLDB
  Endowment. \\
  \raggedright Proceedings of the VLDB Endowment, Vol. \vldbvolume, No. \vldbissue\ %
  ISSN 2150-8097. \\
  \href{https://doi.org/\vldbdoi}{doi:\vldbdoi} \\
  }
  \addtocounter{footnote}{-1}
  \endgroup

  \ifdefempty{\vldbavailabilityurl}{}{ \vspace{.3cm} \begingroup\small\noindent\raggedright\textbf{PVLDB Artifact Availability:}\\ The source code, data, and/or other artifacts have been made available at \url{\vldbavailabilityurl}. \endgroup }

  \section{Introduction}\label{sec:intro}

Vector databases now back semantic search, recommendation, and
retrieval-augmented generation (RAG)~\cite{lewis2020retrieval}, and
enterprise deployments increasingly hold data of differing
sensitivity. A hospital's RAG corpus, for instance, mixes clinical
notes with billing records: physicians may \shanshan{be permitted to \change{read}} the former but not the
latter, while administrators have the reverse view. The vector store
must therefore enforce role-based access control
(RBAC)~\cite{bertino2011access}: a top-$k$ approximate
nearest-neighbor (ANN) query issued by role $r$ must return only
vectors that $r$ is authorized to read. Regulations such as the EU AI
Act~\cite{AI_Act} make this isolation a compliance requirement rather
than an optional feature. Recall remains critical at the same time\change{:} a
clinical query that misses an authorized record can change a treatment
decision.

\smallskip
\noindent\textbf{Design space.}
There are two extreme strategies. A \emph{global index} builds
one ANN graph, e.g., HNSW~\cite{hnsw}, over the entire dataset and
discards unauthorized results after search. Storage stays minimal,
but most of the search effort lands on vectors the querying role
cannot access, and recall drops unless the system over-samples to
compensate. An \emph{oracle index}, in contrast, builds a dedicated
index over exactly the vectors authorized to \shanshan{each
role combination that might be queried}~\cite{acorn}. Every query then runs on an index that contains
only vectors authorized to the query role. No post-filtering or
\shanshan{beam (i.e., \change{the search priority queue})} expansion is needed, but every vector shared across $m$ roles is
stored $m$ times. \shanshan{More generally, supporting \change{every} distinct role combination \change{that might be queried would
require a separate pure index per role combination. The resulting indices incur high memory
pressure and per-query index-switching overhead, so} the oracle \change{is} a reference point rather than a practical
design. }We measure this trade-off with two
metrics:

\noindent\emph{Storage amplification} (\textsf{SA}) is the ratio of
the number of total vectors to the dataset size. The global index attains
$\textsf{SA}=1$ and the oracle index sets the upper bound.

\noindent\emph{Query amplification} (\textsf{QA}) is the average query
cost normalized to oracle indexing. The oracle attains $\textsf{QA}=1$ and
the global index sets the upper bound.

An ideal index layout sits near the lower-left corner of the
\textsf{SA}--\textsf{QA} plane.

\smallskip
\noindent\textbf{Prior work.}
Existing approaches fall into two broad categories. The first treats \emph{the access
condition as a generic attribute predicate and folds it into ANN
search.} Filtered-DiskANN~\cite{filtered_diskann} adds label-aware
edges to the Vamana graph so that traversal stays close to matching
vectors; ACORN~\cite{acorn} elongates each HNSW node's neighbor list by
a factor $\gamma$ so that the predicate-matching subgraph remains
navigable after filtering. Both are predicate-agnostic and do not
exploit the structure of RBAC: a small, mostly fixed set of roles
whose authorized views overlap heavily under union semantics, \shanshan{i.e., a multi-role user's view
is the union of its roles' authorized vectors.} The
second category \emph{materializes multiple sub-indices.}
SIEVE~\cite{sieve} mines historical query workloads to build HNSW
indices for the most profitable filters, falling back to a global
index or brute force for the rest. HoneyBee~\cite{honeybee} casts
RBAC-aware indexing as constrained optimization and produces
partitions that trade storage for latency. Two limitations remain.
First, each role's data must fit inside a single partition, so
fine-grained overlap among roles cannot be exploited. Second,
partitions are searched independently and merged afterward. When a
partition is \emph{impure} for role $r$, i.e., it contains vectors
that $r$ may not have access to, the search on that partition must inflate its
beam to recover enough authorized candidates.

\smallskip
\noindent\textbf{Data partition.}
Data partitioning is the natural fit for access-constrained ANN
search. Three properties of RBAC make this so:
(i)~\emph{Structure:} Access policies follow organizational hierarchy
rather than arbitrary attributes. Physicians read clinical notes,
administrators read billing data, and a small public slice is shared.
The data therefore admits a stable, role-aligned set of partitions,
unlike category filters (e.g., color, price) whose predicate
combinations are unbounded; 
(ii)~\emph{Union semantics:} A user's
view is the union of its roles' authorized vectors, so queries are
answered by \emph{unioning} partition results rather than intersecting them; 
(iii)~\emph{Locality of change:} Policies evolve
as departments merge or permissions are revoked. A partitioned layout
absorbs such changes locally \shanshan{(Appendix~\ref{app:dynamic-workloads})}, whereas removing a role's vectors from a
global ANN graph can sever connectivity and force a full rebuild.

\change{Partitioning is not always optimal: when a query's authorized
region covers a large portion of the dataset, a global index with
post-filtering is cheaper. We treat the two as
complementary, i.e., partitioning for selective queries and global indexing with filtering for
broad ones, and confirm the crossover on multi-role query workloads in 
Exp~\ref{exp:multi-role} in \S\ref{sec:evaluations}.}

\smallskip
\noindent\textbf{This paper.}
We present \name and its efficient variant \effname, two partitioning
strategies built on an \emph{access-aware lattice}. \change{\name exhaustively
identifies candidate node pairs for each operation, whereas \effname applies a more selective bottom-up traversal to
avoid repeated re-evaluation. }Both methods first construct an access-aware lattice by splitting the dataset
into disjoint \emph{exclusive blocks}, \shanshan{where
each block contains exactly vectors authorized to the same role set. 
Then, under a user-specified storage amplification, the methods restructure the lattice with two operations on blocks. A~\emph{copy} operation duplicates a block into
another block whose roles are a subset of its roles. It spends storage budget
but preserves purity, i.e., every vector remains authorized to the
roles on the destination block.} A~\emph{merge} operation fuses two blocks; it
spends no storage budget but may introduce vectors outside some role's view.

Our approaches optimize the structure of the lattice by applying operations greedily with the highest
\emph{benefit ratio}, the
query-cost reduction per unit of added
storage until the storage budget is exhausted or no operation has positive benefit.
\shanshan{Once the lattice is fixed, large nodes (sizes higher than an indexing threshold $\Lambda$) are materialized as HNSW indices, and smaller nodes are decomposed into exclusive blocks and are kept as \emph{leftovers} for efficient linear scan, since linear scan
dominates HNSW at small sizes (\S\ref{sec:challenges}). \change{Leftover vectors are
stored as exclusive blocks and each block is scanned as a unit when queried.}}
For each role
we derive a minimal set of nodes whose union
covers that role's authorized data and call it a \emph{query plan}. The query plan is executed with
\emph{coordinated search}, which shares a global top-$k$ heap across
the nodes in the plan. \shanshan{Searching on indices (and leftovers) that contain only authorized vectors is prioritized over indices that may also contain unauthorized vectors. The former fills the heap first and sets a global distance bound.
The latter is then probed with a standard HNSW search, and the beam of the HNSW index is expanded only if its local $k$-th candidate is closer than the global bound. Inflation is thus paid only when it can improve the result, and the mechanism is orthogonal to the partitioning strategy.}

\begin{figure}[h]
  \centering
  \includegraphics[width=0.8\linewidth]{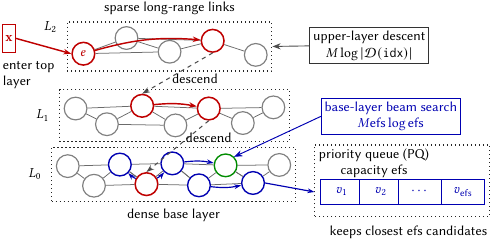}
  \caption{HNSW search through graph layers.}
  \label{fig:hnsw-overview}
\end{figure}

\smallskip
\noindent\textbf{Contributions.}
This paper makes the following contributions:

\noindent
$\bullet$ We formalize access-controlled ANN indexing as an
optimization problem that minimizes the query cost subject to an
\textsf{SA} budget (\S\ref{sec:problem-spec});

\noindent
$\bullet$ We present \name, a data paritition algorithm that exhaustively enumerates candidates for each operation,
(\S\ref{sec:veda}), and \effname, a more selective bottom-up variant that reaches
comparable search performance at a much lower construction cost (\S\ref{sec:eff_veda}).

\noindent
$\bullet$ We introduce \emph{coordinated search}, a query execution strategy
that schedules searches across selected indices and uses a global top-$k$ heap
to prune unnecessary inflated exploration on indices containing unauthorized
vectors, while preserving recall (\S\ref{sec:query-execution}).

\noindent
$\bullet$ \shanshan{Experimentally, \effname achieves searching efficiency comparable to \name while substantially reducing data-partitioning time, e.g., from $77.412$--$1338.673$s to $1.983$--$3.794$s on SIFT-1M.}

\noindent$\bullet$ \shanshan{
\name and \effname substantially outperform existing data partitioning
methods, e.g., under uniform single-role workloads across three datasets and
different \textsf{SA} budgets, \name is $2.24$--$5.60$x faster than SIEVE
and $8.69$--$42.61$x faster than HoneyBee, while \effname is
$2.22$--$4.19$x and $8.46$--$32.59$x faster, respectively.}

\section{Role-Based Access Via Indices}\label{sec:preliminary}

\begin{figure*}[t]
  \centering
\includegraphics[width=0.8\linewidth]{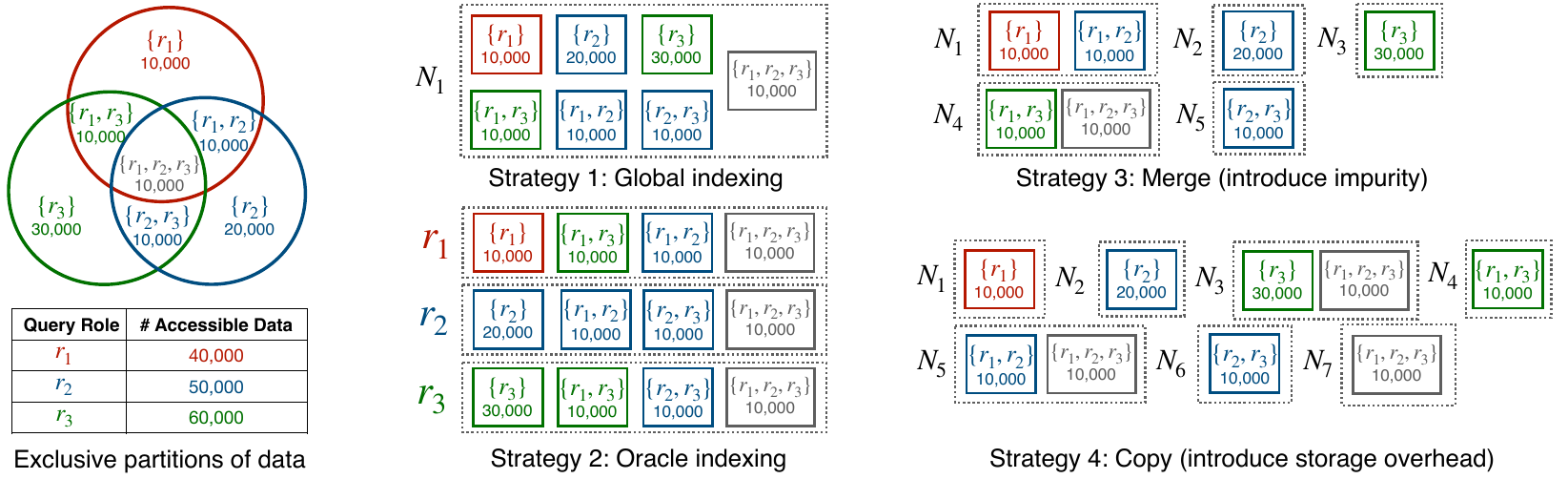}
  \caption{Indexing strategies with a running example with three roles. Dotted boxes denote data groups.}
  \label{fig:indexing_example}
\end{figure*}

This section lays the groundwork for access-aware indexing.
\S\ref{sec:hnsw} reviews HNSW~\cite{hnsw} (details in
Appendix~\S\ref{sec:prelim-hnsw}). \S\ref{sec:ac-indexing} defines index
\emph{impurity} under access constraints, formalizes the query-cost
model, and walks through four indexing strategies on an example to motivate our approach.
Finally, \S\ref{sec:challenges} distills three challenges that motivate the design choices for \name
and \effname.

\subsection{HNSW Overview}\label{sec:hnsw}
HNSW organizes vectors in a multi-layer proximity graph: upper layers hold
sparse long-range links and lower layers hold dense short-range
links~\cite{hnsw}; see Figure~\ref{fig:hnsw-overview}. 
A top-$k$ query enters at a random node on the top layer
and greedily descends toward closer neighbors until it reaches the bottom
layer. Search on the bottom layer expands each candidate node's neighbors
through a priority queue $\textsf{PQ}$ of capacity $\mathsf{efs}$ and returns
the $k$ closest entries.
In practice, $\mathsf{efs}$ is a query-time beam width chosen to be at least
$k$ and tuned for the desired recall--latency trade-off. We set $\mathsf{efs}$ to $\alpha k$, where $\alpha$ is a small
constant (typically $5$--$10$ in practice~\cite{faiss,acorn,sieve,aumuller2020ann}).
For instance, a top-$10$ query on $\mathbf{x}$ retains
$\mathsf{efs}\in[50,200]$ candidates in the priority queue $\textsf{PQ}$. The implicit assumption is that the true top-$k$ neighbors of $\mathbf{x}$ appear among these $\mathsf{efs}$ candidates with high probability. We state it formally.

\begin{assumption}[HNSW Search Accuracy]
\label{assump:hnsw}
For a HNSW index {\tt idx}, let $\mathcal{D}({\tt idx})$ be the set of vectors stored in the index and $\textsf{PQ}$ be the
priority queue of capacity $\mathsf{efs}$ maintained during search on
${\tt idx}$.
Given a top-$k$ query $q=(\mathbf{x}, r)$ for the nearest
neighbors of the vector $\mathbf{x}$ that are accessible to a role $r$, we assume that, 
with high probability,
$\textsf{dist}(\mathbf{x}, v) < \textsf{dist}(\mathbf{x}, v')$
for all $v \in \textsf{PQ}$ and $v' \in \mathcal{D}({\tt idx}) \setminus \textsf{PQ}$.
\end{assumption}

\noindent\textbf{Asymptotic Query Cost.}
Let ${\tt idx}$ be an HNSW index and let $M$ be the number of neighbors per node. A top-$k$ query on the index proceeds in three steps: descent through the upper $\log|{\tt idx}|$ layers, base-layer beam search using a priority queue of size $\mathsf{efs}$, and
final top-$k$ extraction. The total standard comparison-based bound is
$\mathcal{O}\!\left(M\log |{\tt idx}| + M\mathsf{efs}\log \mathsf{efs}
  + \mathsf{efs}\log k\right)$.
Typically, $k$ is small, e.g., top-$1$ to
top-$5$ passages suffice for LLM inference~\cite{reichman2024retrieval,
liu2023lost,xu2023retrieval}. 
To use this bound as a comparable cost proxy, 
we define $\theta=(a,b,c)$, where $a$ weights upper-layer traversal, $b$ weights
base-layer search, and $c$ captures fixed per-query overhead. The resulting
cost is
    $C^{\textsf{cmp}}_{\theta}({\tt idx},\mathsf{efs})
  = a\log |{\tt idx}| + b\,\mathsf{efs}\log\mathsf{efs} + c.$
However,
$C^{\textsf{cmp}}_{\theta}({\tt idx},\mathsf{efs})$ counts \emph{comparisons}, not wall-clock latency. The
$\log\mathsf{efs}$ factor arises solely from heap maintenance, whereas each of
the $\approx\mathsf{efs}$ expansion steps is dominated by $M$ distance
evaluations of $\mathcal{O}(d)$ FLOPs and cache-missing neighbor fetches
(${\sim}2{,}048$ multiply-adds vs.\ $\log_2\mathsf{efs}{\approx}7$ heap
comparisons at $d{=}128$, $M{=}16$). The marginal cost of an extra unit of
$\mathsf{efs}$ is therefore effectively constant, and we replace
$b\,\mathsf{efs}\log\mathsf{efs}$ with $b\,\mathsf{efs}$ in the cost model.

\begin{definition}[Cost of HNSW Search]
  Let ${\tt idx}$ be an HNSW index with size $|{\tt idx}|$ and beam width
$\mathsf{efs} = \alpha k$. Let $\theta=(a,b,c)$ be the calibrated coefficients
for index-size cost, beam-search cost, and fixed per-query overhead. We define the cost model as follows.
\begin{equation}\label{eqn:derivedcostmodel}
C_{\theta}({\tt idx},\mathsf{efs}) = a\log |{\tt idx}| + b\,\mathsf{efs} + c
\end{equation}
\end{definition}



Appendix~\ref{sec:cost-function-selection} validates this choice empirically
(e.g., $b\,\mathsf{efs}$ fits measured latency with $R^2{=}0.99$ vs.\ $0.98$ for
$\mathsf{efs}\log\mathsf{efs}$ at $d{=}128$, $M{=}16$) and details how
$\theta=(a,b,c)$ is estimated. 

\subsection{Access-Control-Aware Indexing}\label{sec:ac-indexing}

To enforce access control with vector databases, we partition the dataset into groups and
build one HNSW index per group. A group may hold vectors from several roles. An index with data accessible to role $r$ might be \emph{impure} with respect to $r$, i.e., the index stores vectors that $r$ is not authorized to access.

\begin{definition}[Pure and Impure Indices]
An index ${\tt idx}$ is \emph{pure} with respect to role $r$ if every vector
it stores is authorized to $r$; otherwise ${\tt idx}$ is \emph{impure}. 
\end{definition}

\begin{definition}[Inflation Factor]\label{def:impurity-factor}
Given a vector dataset $\mathcal{D}$ and an index ${\tt idx},$ let $\mathcal{D}(r)$ be all vectors accessible to role $r$ and let
$\mathcal{D}({\tt idx})$ be all vectors in index ${\tt idx}$.  
We quantify
the impurity of {\tt idx} for role $r$ with an inflation factor $\lambda^r_{{\tt idx}}$ as the ratio of the number of vectors in $\mathcal{D}({\tt idx})$ and the number of vectors in the subset of data in ${\tt idx}$ that is accessible to $r,$ i.e., $\mathcal{D}({\tt idx}) \cap \mathcal{D}(r)$. Formally, we have
$\lambda^r_{{\tt idx}} =
  \left\lceil
  \frac{|\mathcal{D}({\tt idx})|}
       {|\mathcal{D}({\tt idx}) \cap \mathcal{D}(r)|}
  \right\rceil.$
\end{definition}

When the HNSW index {\tt idx} is impure for role $r$ and is selected for answering queries with role $r$, to compensate for the vectors not accessible to $r$, the priority
queue $\mathsf{PQ}$ in the index must grow linearly with the impurity $\lambda^r_{{\tt idx}}$, thereby examining more candidates to maintain recall. Therefore,
$k$ is inflated to $\lambda^r_{{\tt idx}}\, k$ and 
$\mathsf{efs}$ is
inflated from $\alpha k$ to $\alpha\, \lambda^r_{{\tt idx}}\, k$.  To avoid notational clutter, when the role $r$ and the index {\tt idx} are clear from context, we write $\lambda$ for
$\lambda^r_{{\tt idx}}$ for simplicity. In the following, we define the cost model based on $\lambda$ and our analysis in Equation~\ref{eqn:derivedcostmodel}.

\begin{definition}[Cost Model] 
\label{def:hnsw-cost} 
Let {\tt idx} be an HNSW index (denoted by \textsf{H}), with priority queue of size $\mathsf{efs}$, built on a dataset $\mathcal{D}$. For a role $r$ with $\lambda$ as the inflation factor in $\mathcal{D}$, the cost function of a top-$k$ query with role $r$ on ${\tt idx}$ is defined as follows:
\[
  \mathsf{Cost}_{\textsf{H}}({\tt idx}, r) =
  \begin{cases}
  C_{\theta}({\tt idx},\mathsf{efs}), & \text{if } {\tt idx} \text{ is pure w.r.t.\ } r\\
  C_{\theta}({\tt idx},\lceil \lambda\mathsf{efs} \rceil), & \text{if impure w.r.t. } r\\
  \end{cases}
\]
\end{definition}
When the index is pure with respect to the role, no filtering is needed; otherwise, the size of the priority queue is inflated to $\lceil \lambda\mathsf{efs} \rceil$. When inflation exceeds the index size $|{\tt idx}|$, the query degenerates into a full scan. For ease of explanation, by default, we assume that the index size is large enough to accommodate the inflation.

\begin{figure}[h!]
  \centering
  \begin{subfigure}[t]{0.23\textwidth}
    \includegraphics[width=\linewidth]{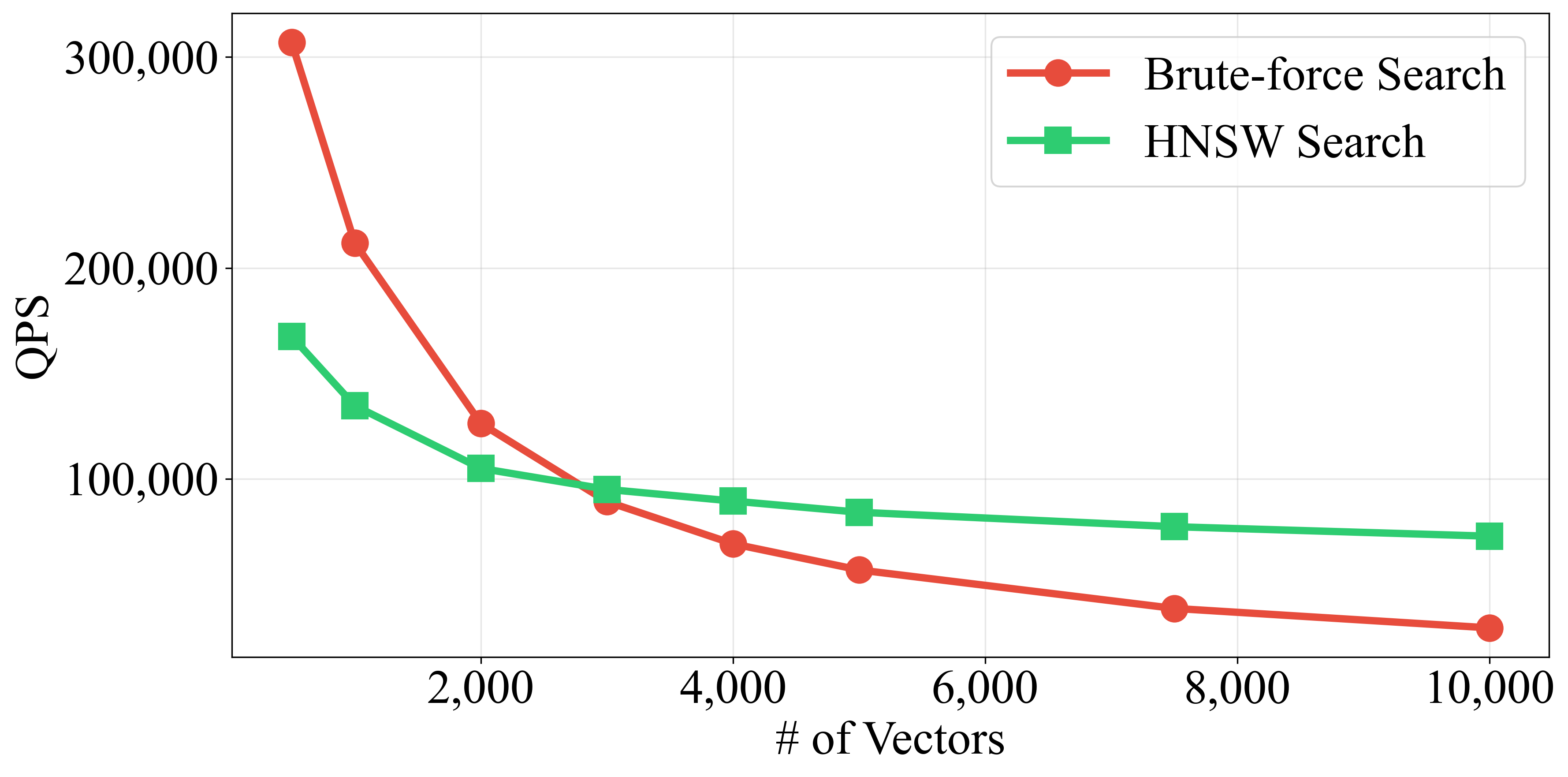}
    \caption{128D data.}\label{fig:128D}
  \end{subfigure}
  \hfill
  \begin{subfigure}[t]{0.23\textwidth}
    \includegraphics[width=\linewidth]{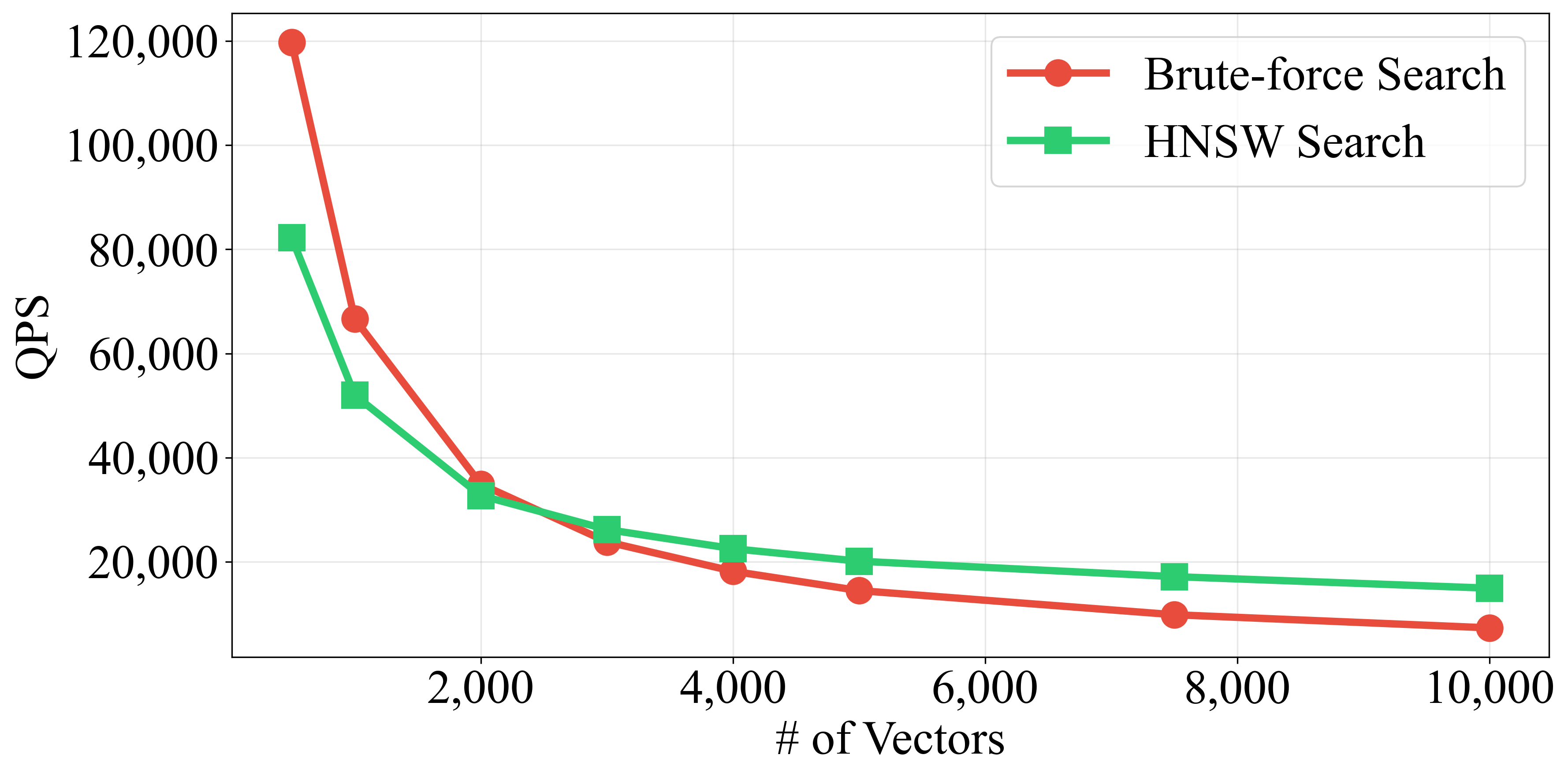}
    \caption{384D data.}\label{fig:384D}
  \end{subfigure}
  \caption{HNSW vs.\ brute-force search ($d{=}128, 384$).}
  \label{fig:vary-dimensions}
\end{figure}

We now compare four partitioning strategies on a toy dataset of $10{,}000$
vectors with three roles $\{r_1, r_2, r_3\}$ (Figure~\ref{fig:indexing_example}).

\noindent$\bullet$ \textbf{Baseline~1 (Global Index).} One index over the
entire dataset; every query post-filters unauthorized results ($\textsf{SA} = 1$). 

\noindent$\bullet$ \textbf{Baseline~2 (Oracle Index).} An ideal
baseline: construct an index over
exactly $\mathcal{D}(r)$ for each role $r$. This oracle partition
requires no post-filtering or search expansion, but it duplicates every
vector that is visible to multiple roles ($\textsf{SA} = 1.5$).

\noindent$\bullet$ \textbf{Strategy~3 (Merge).} Blocks are merged. No data is duplicated but impurity is introduced. In Figure~\ref{fig:indexing_example}, blocks of $\{r_1\}$ and $\{r_1, r_2\}$ are merged, so do blocks of
$\{r_1, r_3\}$ and $\{r_1, r_2, r_3\}$. $r_1$ and $r_3$ hit pure indices, but queries for $r_2$ have to traverse $10{,}000$ vectors exclusive to $r_1$ and $10{,}000$ vectors exclusive to $r_1$ and $r_3$ ($\textsf{SA}=1$).  

\noindent$\bullet$ \textbf{Strategy~4 (Copy).} Blocks are duplicated to improve query efficiency while cutting impurity. In
Figure~\ref{fig:indexing_example}, the block $\{r_1, r_2\}$ is copied into the group $\{r_1\}$. Queries with $r_2$ access the original node $\{r_1, r_2\}$, eliminating impurity when searching for role $r_2$.
Similarly, the block $\{r_1, r_2, r_3\}$ is copied
into $\{r_1, r_3\}$ ($\textsf{SA} = 1.2$).

\subsection{Challenges}\label{sec:challenges}

The example exposes an \textsf{SA}--\textsf{QA} trade-off and motivates an access-aware
indexing scheme that leverages copy and merge under a storage budget. 
Copying is preferable to merging when storage
permits, since it preserves the source block and never degrades any role; merging is the fallback when the
budget is exhausted. 
Thus, three
challenges follow.

\noindent\textbf{(C1)~Data Partition:} given an \textsf{SA} budget, partition
$\mathcal{D}$ to minimize overall query cost; choosing between HNSW and linear scan
per data group since linear scan dominates on small groups
(Figure~\ref{fig:vary-dimensions}).

\noindent\textbf{(C2)~Index Selection:} given the resulting indices, choose which
subset each role should query to minimize query cost.

\noindent\textbf{(C3)~Query Answering:} coordinate search across multiple (possibly
impure) indices so that redundant exploration is avoided. 

  \section{Data Partition and Index Selection}
\label{sec:problem-spec}

This section formalizes {Data Partition} and {Index Selection}
and illustrates both on the dataset of Figure~\ref{fig:indexing_example}.

\subsection{Problem Definition}
\label{sec:problem-def}

Given a vector dataset $\mathcal{D}$ in which each vector is accessible to one
or more roles in the set of roles $\mathcal{R}$, and a \textsf{SA} budget $\beta,$ we seek a partition of $\mathcal{D}$ into groups suitable for
indexing. The objective is to:
\textit{(i)}~minimize the average estimated query cost based on Definition~\ref{def:hnsw-cost}; \textit{(ii)}~limit the impurity of any group with respect to the
roles that access it; and \textit{(iii)}~respect the given \textsf{SA} budget $\beta$. We assume that $\mathcal{D}$ and $\mathcal{R}$ are static during construction, and discuss extensions to
dynamic changes in 
Appendix~\ref{app:dynamic-workloads}. For simplicity of explanation, we assume uniform single-role queries by default and extend to other workloads in \S\ref{sec:evaluations}. 

\smallskip\noindent\textbf{Roles and Access Tags.}
Let $\mathcal{R}=\{r_1,r_2,\ldots\}$ denote the set of roles on $\mathcal{D}$. Each vector
$v \in \mathcal{D}$ is associated with a non-empty set of roles (i.e., role combination)
$\tau_v \subseteq \mathcal{R}$. 
Any role $r\in \tau_v$ is authorized to access $v$.
For any non-empty set of roles $\tau \subseteq \mathcal{R}$, we define the
\emph{exclusive block} for $\tau$ as
$N^{\mathit{ex}}(\tau)=\{v \in \mathcal{D} : \tau_v=\tau\}.$
Thus, $N^{\mathit{ex}}(\tau)$ contains exactly vectors $\in \mathcal{D}$ that are exclusively accessible to the role combination $\tau$.
In Figure~\ref{fig:indexing_example}, the block
$N^{\mathit{ex}}(\{r_1,r_3\})$ contains the vectors that can be accessed by
both $r_1$ and $r_3$, and 
is disjoint from other blocks like $N^{\mathit{ex}}(\{r_1\})$ and
$N^{\mathit{ex}}(\{r_3\})$.
Let
$\mathcal{T}=\{\tau_v: \tau_v \subseteq \mathcal{R}, v \in \mathcal{D}\}$
be the set of role combinations that {are present} in dataset $\mathcal{D}$. 
Thus, $\mathcal{D}$ is partitioned
into disjoint exclusive blocks denoted as $\mathcal{N}_{\mathit{ex}}
= \{N^{\mathit{ex}}(\tau): \tau \in \mathcal{T}\}.$
Formally, $\mathcal{D}
=\bigcup_{\tau \in \mathcal{T}} N^{\mathit{ex}}(\tau),$ and $\forall \tau_1, \tau_2 \in \mathcal{T}, \text{ where } \tau_1 \ne \tau_2, \text{ we have }
N^{\mathit{ex}}(\tau_1) \cap N^{\mathit{ex}}(\tau_2)=\emptyset.$

For a role $r \in \mathcal{R}$, let $\mathcal{D}(r)$ denote the set of vectors
that $r$ is authorized to access. $\mathcal{D}(r)$ contains all exclusive blocks whose set of roles
includes $r.$ More specifically,
$\mathcal{D}(r) = \bigcup_{\tau \in \mathcal{T}: r \in \tau}
N^{\mathit{ex}}(\tau).$

\smallskip\noindent\textbf{Index Set.}
An index is built over a group of data that contains one or more exclusive blocks.
Different indices may contain the same exclusive block. Let $\mathcal{I}$ denote the set of all HNSW indices built over $\mathcal{D}$, and let $\mathcal{D}({\tt idx})$ be the set of vectors used to build index ${\tt idx}$. Each index can be represented as a union of exclusive blocks. Denote $\mathcal{T}^\prime$ as a set of role combinations corresponding to the exclusive blocks in ${\tt idx}$, 
we have $\mathcal{D}({\tt idx}) = \bigcup_{\tau \in \mathcal{T}^\prime} N^{\mathit{ex}}(\tau)$. 
Let $\mathcal{I}(r) \subseteq \mathcal{I}$ be the set of indices used to answer queries with role $r$. $\mathcal{I}(r)$ is \emph{correct} if 
the indices in $\mathcal{I}(r)$ cover all vectors that $r$ is authorized to access, i.e., $\mathcal{D}(r) \subseteq \bigcup_{{\tt idx} \in \mathcal{I}(r)} \mathcal{D}({\tt idx})$. 

\smallskip\noindent\textbf{Query Model.}
A query $q = (\mathbf{x}, r)$ is issued by users with role $r$ and retrieves
the top-$k$ nearest neighbors of $\mathbf{x}$ within $\mathcal{D}(r)$. The
query may touch any subset $\mathcal{I}(r) \subseteq \mathcal{I}$ whose union
covers $\mathcal{D}(r)$. The expected cost for role $r$ is computed as the sum over $\mathcal{I}(r)$, i.e., $\sum_{{\tt idx} \in \mathcal{I}(r)} \textsf{Cost}_{\textsf{H}}({\tt idx}, r).$ 

\smallskip\noindent\textbf{Objective.}
Let $Q = \{q_1, q_2, \ldots\}$ be a uniform single-role workload on $\mathcal{D}$ in which each
$q = (\mathbf{x}, r) \in Q$ draws $r \in \mathcal{R}$ with equal probability. The
estimated average query cost is
\begin{equation}
\textsf{AvgCost}(Q, \mathcal{I}) = \frac{1}{|Q|} \sum_{q \in Q} \sum_{{\tt idx} \in \mathcal{I}(r)} \textsf{Cost}_{\textsf{H}}({\tt idx}, r).
\end{equation}\label{eqn:avg_cost}

\begin{definition}
Given a vector dataset $\mathcal{D}$ with role set $\mathcal{R}$ 
and an \textsf{SA}
budget $\beta$, the {Index Selection} problem constructs an index set
$\mathcal{I}$ over $\mathcal{D}$ via \emph{merge} and \emph{copy} operations that minimizes
$\textsf{AvgCost}(Q, \mathcal{I})$ subject to $\textsf{SA}(\mathcal{I}) \le \beta$.
\end{definition}

The problem may be viewed as the {Budgeted Set Cover} problem~\cite{khuller1999budgeted,martello1987algorithms}: each
index in $\mathcal{I}$ covers a set of exclusive blocks; 
for each role $r$ and its indices $\mathcal{I}(r)$, we must cover every exclusive block accessible to $r$ under a knapsack constraint while
minimizing query cost~\cite{martello1987algorithms}. Since the {Budgeted Set Cover} problem is
NP-hard~\cite{khuller1999budgeted}, we do not seek exact solutions. Instead, we provide greedy heuristics in
\S\ref{sec:veda}--\S\ref{sec:eff_veda}, and give the MILP formulation in Appendix~\S\ref{sec:opt_formulation}.
We note that, in the worst case, the cardinality of $\mathcal{T}$ on $\mathcal{D}$ is asymptotically bounded by the  number of roles in $\mathcal{R}$, i.e., $|\mathcal{T}|\le 2^{|\mathcal{R}|}.$ However, in practice, $|\mathcal{T}|$ is bounded by the number of
\emph{distinct} permissions that are actually assigned, which is usually much smaller, e.g., $641$--$757$ with $|\mathcal{R}|$ = 64--87~(Table~\ref{tab:dataset_summary} in \S\ref{sec:evaluations}).


\begin{figure}[t]
  \centering
  \includegraphics[width=\linewidth]{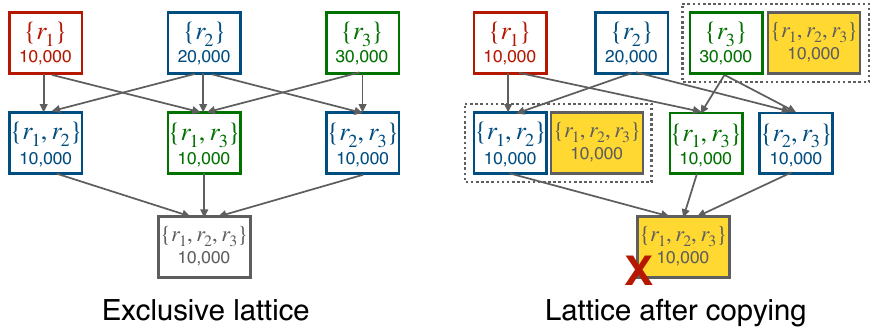}
  \caption{$\mathcal{L}_{\mathit{ex}}$ (left) and $\mathcal{L}$ after copy operations (right).}
  \label{fig:lattice_ex}
\end{figure}

\subsection{Exclusive Lattices}
{Exclusive blocks induce a partial order under role-set inclusion:
$N^{\mathit{ex}}(\tau')$ sits below $N^{\mathit{ex}}(\tau)$ when
$\tau\subset\tau'$, i.e., $\tau'$ is visible to a broader role set. We
organize the blocks by this order in a layered \emph{exclusive lattice} and
manipulate them with copy and merge to form indexable
groups~(\S\ref{sec:veda}--\S\ref{sec:eff_veda}); exclusive lattice construction is in
Appendix~\ref{sec:exclusive_lattice_construction}.}

{Figure~\ref{fig:lattice_ex} shows the original exclusive lattice (left) and the optimized lattice (right) for \textsf{Strategy~4}
of Figure~\ref{fig:indexing_example}. Each block sits on a \emph{layer} given
by $|\tau|$, so three roles yield at most three layers. Edges encode immediate
\emph{parent--child} relations and run from the stricter (fewer-role) parent
to the broader child; e.g., $N^\mathit{ex}(\{r_2, r_3\})$ is a child of both
$N^\mathit{ex}(\{r_2\})$ and $N^\mathit{ex}(\{r_3\})$.}

\begin{definition}\label{def:exclusive-lattice}
  The \emph{exclusive lattice} is a directed acyclic graph
  $\mathcal{L}_\mathit{ex} = (\mathcal{N}_\mathit{ex}, E)$. Each node
  $N^{\mathit{ex}}(\tau) \in \mathcal{N}_\mathit{ex}$ contains data exclusively accessible to $\tau$. Nodes within the same layer have the same cardinality of their role sets~($|\tau|$), and empty nodes and empty layers are omitted.
  $E$ captures access
  scope widening from stricter to broader role sets. 
  An edge $(N^{\mathit{ex}}(\tau), N^{\mathit{ex}}(\tau')) \in E$ defines a
  \emph{parent}--\emph{child} relation when both conditions hold:
  \begin{enumerate}
      \item \textbf{Containment:} $\tau \subset \tau^{\prime}$, so every role
      authorized at the parent is also authorized at the child.
      \item \textbf{Adjacency:} $|\tau| < |\tau^{\prime}|$, and no intermediate node
      $N^{\mathit{ex}}(\tau^{\prime\prime})$ exists with 
      $\tau \subset \tau^{\prime\prime}\subset \tau^{\prime}$.
  \end{enumerate}
  \emph{Descendant--ancestor} and \emph{sibling} relations are in Appendix~\S\ref{app:lattice-relations}.
\end{definition}

Searching based on the lattice requires a \emph{query plan} $\textsf{QP}$
that 
identifies nodes that together cover $\mathcal{D}(r)$ for each role $r\in\mathcal{R}$. Formally, the query plan for $r$, denoted as $\textsf{QP}(r)$, is defined as $\{{N}^\mathit{ex}(\tau) \in \mathcal{L}_\mathit{ex} \mid r \in \tau\}$. 
We leverage the query plan $\textsf{QP}$ when optimizing the lattice in \S\ref{sec:veda} and \S\ref{sec:eff_veda}, and {describe how to construct a minimal-cover $\textsf{QP}$ for each role} in \S\ref{sec:qplan}.
We summarize notations in Table~\ref{tab:notation} in Appendix~\S\ref{sec:notations}.

  \section{\name: The Adaptive Solution}
\label{sec:veda}

\name transforms the exclusive lattice
$\mathcal{L}_{\mathit{ex}}=(\mathcal{N}_{\mathit{ex}},E)$ into an optimized
lattice $\mathcal{L}$ using \emph{copy} and \emph{merge} operations, so that data that might be queried together are grouped. Since brute-force search dominates on small nodes (Figure~\ref{fig:vary-dimensions} and Figure~\ref{fig:vary-dimensions-app} in Appendix~\S\ref{sec:prelim-hnsw}),
 nodes will be indexed with HNSW or kept as leftovers for efficient brute-force search based on their sizes. 

\begin{definition}[Indexing threshold]
\label{def:indexing-threshold}
  The \emph{indexing threshold} $\Lambda$ is the minimum cardinality of nodes for materializing an HNSW index. For each node
  $N \in \mathcal{L}$, \name builds an HNSW index over
  $N$ if $|N| \geq \Lambda$; otherwise, $N$ is kept as leftover vectors.
\end{definition}

\subsection{Algorithm Overview}

\name acts on \emph{descendant--ancestor pairs}
$(N(\tau), N_{a}(\tau'))$ defined by reachability in
$\mathcal{L}_{\mathit{ex}}$: a directed path exists from the ancestor
$N_a(\tau')$ to the child $N(\tau)$ in $\mathcal{L}_{\mathit{ex}}$, equivalently $\tau'\subsetneq\tau$ under
Definition~\ref{def:exclusive-lattice}. A pair remains eligible only while
both nodes are still present in $\mathcal{L}$.
Because $\tau'\subsetneq\tau$, every query that reaches the exclusive block of the ancestor, e.g., $N_{a}^\mathit{ex}(\tau')$,
must also read the descendant block $N_c^\mathit{ex}(\tau)$ in $\mathcal{L}_{\mathit{ex}}$.
Thus, co-locating data selectively enables queries to be served
from fewer nodes. Initially, $\mathcal{L}$ is initialized as $\mathcal{L}_\mathit{ex}$.
Formally, two operations are defined on descendant--ancestor pair
$(N(\tau),\, N_{a}(\tau'))$:

\noindent$\bullet$ \textbf{Copy:} copy vectors in $N_c^\mathit{ex}(\tau)$ to $N_{a}(\tau')$ and keep $N_c^\mathit{ex}(\tau)$ as a separate node; {the source node is preserved} but \textsf{SA} increases.

\noindent$\bullet$ \textbf{Merge:} merge $N(\tau)$ and $N_{a}(\tau')$; {the original nodes $N(\tau)$ and $N_{a}(\tau')$ are removed.} No additional storage overhead; \emph{impurity} is introduced.

\noindent

To navigate the trade-off between \emph{Copy} and \emph{Merge}, \name takes a user-specified \textsf{SA} budget
$\beta$ and ranks candidate operations with \emph{benefit ratio}: query-cost
reduction per unit of added storage.


\begin{definition}[Benefit Ratio]
\label{def:benefit_function}
Let $\mathcal{R}$ be a set of roles and $Q$ be a uniform workload of single-role top-$k$ queries, where each 
query $q=(\mathbf{x},r) \in Q$ is issued with role $r \in \mathcal{R}$.
Let $e$ be an operation applied to a descendant--ancestor pair
$(N(\tau),N_{a}(\tau'))$ in lattice $\mathcal{L}$. Let
$\mathcal{L}^e$ be the lattice obtained after applying $e$, and let
$\Delta S(e)$ denote the additional storage overhead incurred by $e$.
For any lattice $\widetilde{\mathcal{L}}$, let
$\mathcal{I}(\widetilde{\mathcal{L}})$ denote the set of indices represented by
$\widetilde{\mathcal{L}}$. Using $\textsf{AvgCost}$ 
(Equation~\ref{eqn:avg_cost}), the \emph{benefit} of applying $e$ to
$\mathcal{L}$ is defined as
\begin{equation}
\label{eqn:benefit_function}
   f(\mathcal{L}, e)
   =
   \frac{
      \textsf{AvgCost}\!\left(\mathcal{R},\mathcal{I}(\mathcal{L})\right)
      -
      \textsf{AvgCost}\!\left(\mathcal{R},\mathcal{I}(\mathcal{L}^e)\right)
   }{
      \Delta S(e) + 1
   },
\end{equation}
\end{definition}

We set $\Delta S(e)$ to $0$ for merge operations, as merge operations do not incur any additional storage. We also add $1$ at the denominator to prevent division by zero. 
A positive benefit ratio indicates a reduction in query cost, {while} a negative
benefit ratio indicates {that $e$ increases query cost}. {Because} copy operations preserve source nodes
and do not harm queries (see Figure~\ref{fig:indexing_example} and explanations in \S\ref{sec:ac-indexing}), \name applies copy operations first
until it has consumed the storage budget. We justify this prioritization in Theorem~\ref{thm:copy-dominates} in
Appendix~\S\ref{sec:copy_dominance}.

\noindent\textbf{Workflow.}
{Algorithm~\ref{alg:adaptive} (Appendix~\S\ref{app:veda-overview-algorithm}) summarizes \name. Starting from
$\mathcal{L}=\mathcal{L}_{\mathit{ex}}$, \name enumerates the
descendant--ancestor pairs $\mathcal{DA}$ induced by reachability
(Definitions~\ref{def:exclusive-lattice}
and~\ref{def:lattice-ancestor-descendant}) and builds a query plan
$\textsf{QP}$ that defines which node to cover for queries for each role (\S\ref{sec:qplan}). It then iterates \emph{Copy}
(Phase~1) and \emph{Merge} (Phase~2) until the \textsf{SA} budget $\beta$ is
exhausted or no operation has positive benefit under
\cref{eqn:benefit_function}, and cleans up with \emph{Finalization}
(Phase~3).}

\begin{enumerate}[leftmargin=*]
   \item \textbf{Phase~1: Copy (\S\ref{sec:veda-copy-merge}).}
      Copy exclusive blocks to their ancestors; select operations with the highest benefit 
      while
      respecting~$\beta$.

   \item \textbf{Phase~2: Merge (\S\ref{sec:veda-copy-merge}).}
      Merge small nodes into indexable ones (size $\geq \Lambda$);
      reclaiming storage freed due to overlaps. 

   \item \textbf{Phase~3: Finalization (Appendix~\S\ref{sec:veda_finalize}).}
      Decompose remaining unindexable groups into exclusive blocks for brute-force search and handle ``super-impure'' nodes.
\end{enumerate}

During finalization, decomposing unindexable groups might remove duplicated vectors stored across data groups, which may free some storage
budget. \name uses this recovered budget to refine ``super-impure'' nodes: if a
node in a role's query plan contains only a small pure portion for that role,
\name materializes that pure portion as standalone exclusive blocks and
updates the role's query plan to include them. 
To do so, the algorithm check the query plan of each role and computes the purity of each node regarding each role, then sorts a combination of roles and nodes (e.g., $(r, N(\tau))$) based on impurity decreasingly to prioritize nodes with higher impurity.
This avoids searching the
larger impure node when budget permits; any remaining small blocks are kept as leftovers
for efficient brute-force search. 
Appendix~\S\ref{sec:veda_finalize} gives the detailed
procedure, including Algorithm~\ref{alg:adaptive-finalize}.

After Phase~3, 
\name materializes nodes of sizes $\geq\Lambda$ as HNSW indices and keeps other nodes as leftovers. Vectors of same exclusive blocks are grouped together, and each exclusive block is scanned as a unit when queried. 



\subsection{Greedy Copy and Merge}
\label{sec:veda-copy-merge}


{Both phases greedily select the highest-benefit operation each time:
copy duplicates exclusive blocks into ancestors to group co-queried data while
preserving ancestor purity, and merge then grows the resulting nodes into
indexable groups. Algorithms~\ref{alg:veda-copy} and~\ref{alg:veda-merge} give
the pseudocode; detailed helper algorithms are deferred to
Appendix~\S\ref{app:veda-copy-helpers} and~\S\ref{app:veda-merge-helpers}.}
Both phases maintain a candidate set $\textsf{PR}$ for operations
over descendant--ancestor pairs that are still present in $\mathcal{L}$, with available buffer
$\mathit{buf}=\beta|\mathcal{L}_{\mathit{ex}}| - |\mathcal{L}|$ ($|\mathcal{L}|$ counts duplicates). At each step the highest-benefit candidate is selected subject to
\begin{equation}
   \label{eqn:greedyconditions}
   e^*=\arg\max_{e \in \textsf{PR}} f(\mathcal{L},e),\quad
   f(\mathcal{L},e^*) \ge 0 \quad\text{and}\quad
   \Delta S(e^*) \le \mathit{buf}.
\end{equation}
A copy operation
$e(N(\tau), N_{a}(\tau'))$ consumes
$\Delta S(e) = |\mathcal{L}_{\mathit{ex}}[N(\tau)]|$
$+ |N_a(\tau')|$
$- |\mathcal{L}_{\mathit{ex}}[N(\tau)] \cup N_a(\tau')|$;
a merge consumes nothing. 
After each operation, $\mathit{buf}$, the query plan $\mathsf{QP}$, 
 and 
the candidate pairs that
involve the modified nodes are updated. 

\begin{example}
Suppose $\mathcal{L}$ is initialized with the exclusive lattice in Figure~\ref{fig:lattice_ex}. 
Suppose after a set of copy operations, $\mathcal{L}$ contains the following nodes: $N_1=\{N^\mathit{ex}({r_1}), N^\mathit{ex}({r_1, r_2}), N^\mathit{ex}({r_1, r_2, r_3})\}$, $N_2=\{N^\mathit{ex}({r_2})\}$, $N_3=\{N^\mathit{ex}({r_3})\}$, $N_4=\{N^\mathit{ex}({r_1, r_3})\}$, $N_5=\{N^\mathit{ex}({r_2, r_3})\}$, $N_6=\{N^\mathit{ex}({r_1, r_2})\}$, and $N_7$ $=\{N^\mathit{ex}({r_1, r_2, r_3})\}$. Queries with $r_2$ can be issued with $\mathsf{QP}_1(r_2) = \{N_1, N_2, N_5\}$ or $\mathsf{QP}_2(r_2) = \{N_2, N_5, N_6, N_7\}$. $\mathsf{QP}_1(r_2)$ searches one node less than $\mathsf{QP}_2(r_2)$, but has to probe unauthorized vectors in $N_1$. Thus, the query cost with $\mathsf{QP}_1(r_2)$ and $\mathsf{QP}_2(r_2)$  has to be re-evaluated to select the cheaper plan.
\end{example}

\begin{theorem}[Correctness of Greedy Copy Phase]
   \label{thm:copy-phase-correctness} Let $\mathcal{L}_{t}$ be the lattice
   after $t$ greedy copy operations. Then:
   \begin{enumerate}
      \item \emph{Monotonicity:}
         $\textsf{AvgCost}(\mathcal{R}, \mathcal{I}(\mathcal{L}_{t})) \le \textsf{AvgCost}(\mathcal{R},
         \mathcal{I}(\mathcal{L}_{t-1}))$.
      \item \emph{Budget Safety:}
         $|\mathcal{L}_t|/|\mathcal{L}_\mathit{ex}|\le \beta$.
      \item \emph{Termination:} the phase halts when no $e$ satisfies
         Equation~\ref{eqn:greedyconditions}.
   \end{enumerate}
\end{theorem}

\begin{theorem}[Correctness of Merge Phase]
   \label{thm:merge-phase-correctness} Let $\mathcal{L}_{t}$ be the lattice
   after $t$ merge operations. Then:
   \begin{enumerate}
      \item \emph{Monotonicity:}
         $\textsf{AvgCost}(\mathcal{R}, \mathcal{I}(\mathcal{L}_{t})) < \textsf{AvgCost}
         (\mathcal{R}, \mathcal{I}(\mathcal{L}_{t-1}))$.
      \item \emph{Storage Non-Increase:} if $N(\tau)$ and $N_{a}(\tau')$
         share exclusive blocks (from prior copies), then
         $|\mathcal{L}_{t}| < |\mathcal{L}_{t-1}|$; otherwise
         $|\mathcal{L}_{t}| = |\mathcal{L}_{t-1}|$.
      \item \emph{Termination:} the phase halts when no $e$ satisfies
      Equation~\ref{eqn:greedyconditions}.
   \end{enumerate}
\end{theorem}
\noindent Proofs for the theorems above are deferred to Appendix~\S\ref{app:veda-correctness-proofs}.

\noindent\textbf{\textit{Complexity.}}
Each step re-evaluates up to $\mathcal{O}(|\mathcal{N}|^{2})$ candidate pairs;
each evaluation rebuilds the affected roles' query plans in
$\mathcal{O}(|\mathcal{R}|\!\cdot\!|\mathcal{N}|)$. With $\mathcal{O}(|\mathcal{N}|)$
steps, the copy phase runs in $\mathcal{O}(|\mathcal{N}|^{5}|\mathcal{R}|)$ and the merge
phase in $\mathcal{O}(|\mathcal{N}|^{4}|\mathcal{R}|)$. This cost is dominated by the
$\textsf{QP}$ re-derivation inside Equation~\eqref{eqn:benefit_function};
removing that dependency is exactly what \effname does next.

\begin{algorithm}
\caption{\textbf{\name - Copy}}
\label{alg:veda-copy}
\begin{algorithmic}[1]

\Require exclusive lattice $\mathcal{L}_\mathit{ex}$, lattice $\mathcal{L}$, \textsf{SA} budget $\beta$, descendant--ancestor pairs $\mathcal{DA}s$, query plan $\textsf{QP}$ for $\mathcal{L}$.

\State $\mathit{buf}\gets \beta\times|\mathcal{L}_\mathit{ex}| - |\mathcal{L}|$, $\textsf{PR}\gets \emptyset$  
\If{$\mathit{buf}\leq 0$} \textbf{return} $\mathcal{L}$
\EndIf

\For{$(N_c, N_a)\in\mathcal{DA}$}
~$\textsf{PR}[(N_a, N_c)]\gets -1$
\EndFor

\State $\textsf{PR}\gets \textsf{GetCopyPairs}(\mathcal{L}, \textsf{PR}, \texttt{None}, \textsf{QP})$\Comment{Algorithm~\ref{alg:veda-copy-complementary}}

\While{$\mathit{buf}> 0$}
    
    \If{the best pair $\in\textsf{PR}$ has benefit below 0}~\textbf{break}
    \EndIf
    \For{$(N_c, N_a)\in\textsf{PR}$}

     \State $\Delta S(e)\gets |N_a\cup\mathcal{L}_\mathit{ex}[N_c]|-|N_a|$
     \If{$\Delta S(e)\leq \mathit{buf}$}
     \State $N_a.\texttt{add}(\mathcal{L}_\mathit{ex}[N_c])$, $\textsf{QP.renew()}$, $\mathit{buf}.\textsf{adjust}()$, 
     $\textsf{PR}\gets\textsf{GetCopyPairs}(\mathcal{L}, \textsf{PR}, N_a, \textsf{QP})$, \textbf{break} \Comment{Rescore pairs where $N_a$ is the ancestor}
     \EndIf
    \EndFor
    \If{no copy for this round} \textbf{break}
    \EndIf
\EndWhile

\State \textbf{return} $\mathcal{L}$, $\textsf{QP}$
\end{algorithmic}
\end{algorithm}

\begin{algorithm}[h!]
\caption{\textbf{\name - Merge}}
\label{alg:veda-merge}
\begin{algorithmic}[1]
\Require exclusive lattice $\mathcal{L}_\mathit{ex}$, lattice $\mathcal{L}$, descendant--ancestor pairs $\mathcal{DA}s$, query plan $\textsf{QP}$ for $\mathcal{L}$.

    \State $\textsf{PR}\gets \phi$ 

    \For{$(N_c, N_a)\in\mathcal{DA}$}~$\textsf{PR}[(N_a, N_c)]\gets -1$
    \EndFor
    \State $\textsf{PR}\gets \textsf{GetMergePairs}(\mathcal{L}, \textsf{PR}, N_a, N_c, \textsf{QP})$\Comment{Algorithm~\ref{alg:veda-merge-complementary}}
    \While{\texttt{True}}
    
    \State 
    $(N_a, N_c) \gets \arg\max_{(N_a', N_c') \in \textsf{PR}} f(\mathcal{L}, e(N_a', N_c'))$

    \If{$f(\mathcal{L}, e(N_a, N_c)) \leq 0$}~\textbf{break}
    \EndIf
    
    \State $N_a.\textsf{add}(N_c)$, $N_c.\textsf{delete}()$, 
    $\textsf{QP.renew}()$

    \State $\textsf{PR}\gets \textsf{GetMergePairs}(\mathcal{L}, \textsf{PR}, N_a, N_c, \textsf{QP})$\Comment{Algorithm~\ref{alg:veda-merge-complementary}}
    \EndWhile
    \State \textbf{return} $\mathcal{L}$

\end{algorithmic}
\end{algorithm}
  \section{\effname: An Efficient Solution}
\label{sec:eff_veda}

\name suffers from iterations of copy--merge phases and heavy 
re-evaluation after each operation. Candidate node pairs are re-scored with
Equation~\ref{eqn:benefit_function} after each operation, 
and therefore requires refreshing query plan $\textsf{QP}$ after each operation. 
\effname gives the copy and merge phases specific goals. It applies
a fixed traversal order to avoid iterations of copy and merge phases and heavy re-scoring, and achieves search performance comparable to \name (detailed in \S\ref{sec:evaluations}). 
The goals of Copy and Merge in \effname are as follows. 

\noindent$\bullet$ \textbf{Copy. }
Bottom-up traversal and aggressive full-node duplication: each time duplicates a node's \emph{entire} data (not just one exclusive block) into one or
more ancestors while keeping every node \emph{pure} towards its original role set. 

\noindent$\bullet$ \textbf{Merge. }  Target indexability: greedily grow nodes that sit just
below the threshold $\Lambda$ until they become indexable.

Each time, the copy phase identifies \emph{valid partitions}, i.e., a set of ancestors whose role sets cover $\tau$ exactly and disjointly, for each candidate copied node $N(\tau)$. 
Unlike \name, \effname estimates benefit from the marginal cost reduction per role alone and does not maintain a query plan $\textsf{QP}$ or re-score candidate node pairs after each operation. 
Since the merge phase introduces impurity, 
{nodes are no longer pure for their original role sets after merging and cannot be copied again, so \effname omits the copy--merge iteration.} However, it applies the same finalization step (including splitting small nodes into leftovers and handling super-impure nodes; see Appendix~\S\ref{sec:veda_finalize}) as \name to compensate for the omitted copy--merge iterations.
The workflow is summarized in Algorithm~\ref{alg:eff-adaptive} in Appendix~\S\ref{app:effveda-complementary-algorithms}.

\subsection{Phase 1: Copying}
\label{sec:eff_veda_copy}

\effname avoids heavy re-computation with a bottom-up traversal and a restriction of copy operations to a set of ancestors, i.e., a \emph{valid partition} of the copied node. Role combinations of the ancestors in the valid partition tile the child's role set exactly and disjointly. Under this
restriction, the change in $\textsf{AvgCost}$ can be computed from the copied node and the ancestor involved each time (Lemma~\ref{lem:copy_locality}). The benefit function
collapses to a closed-form expression in $C_\theta$, which does not require updating the query plan $\mathsf{QP}$ and the lattice is swept once, bottom-up.


\begin{definition}[Valid Partition]
\label{def:valid_partition} Let $A_c^\tau$ denote the set of ancestors of $N(\tau)$, i.e., $A_c^\tau = \{N(\tau')\in\mathcal{L}:\tau'\subsetneq\tau\}$.
A set of ancestor role sets $\textsf{P}_\tau=\{\tau_1,\ldots,\tau_m\}$, where
$N_a(\tau_j)\in A_c^\tau$ for all $1\leq j\leq m$, is a \emph{valid partition} of
$N(\tau)$ if
\begin{enumerate}
  \item \textbf{coverage:}\; $\bigcup_{j=1}^{m}\tau_j=\tau$,\quad and
  \item \textbf{disjointness:}\; $\tau_j\cap\tau_{j'}=\emptyset$ for all
        $j\neq j'$ \;(i.e., $\sum_j|\tau_j|=|\tau|$).
\end{enumerate}
We write $\mathbb{P}(N(\tau))$ for the set of all valid partitions of $N(\tau)$.
\end{definition}

\begin{example}
\label{ex:valid_partition}
In the lattice of Figure~\ref{fig:lattice_ex},
$\{\{r_1,r_2\},\{r_3\}\}$, $\{\{r_1,r_3\},\{r_2\}\}$ and
$\{\{r_1\},\{r_2\},\{r_3\}\}$ are all valid partitions of node 
$N(\{r_1,r_2,r_3\})$, whereas $\{\{r_1,r_2\},\{r_2,r_3\}\}$ is not because
$r_2$ is covered twice.
\end{example}

{Copying $N(\tau)$ into a valid partition $\textsf{P}_\tau$ replicates
$N(\tau)$ into every $N_a(\tau_j)\in\textsf{P}_\tau$ and deletes $N(\tau)$.
Since the ancestors cover $\tau$ exactly and disjointly, each role $r\in\tau$
finds its authorized vectors from $N(\tau)$ in exactly one ancestor, so
$N(\tau)$ is no longer needed for any $r\in\tau$ and is safe to delete. The added storage is
$|N(\tau)|\cdot(|\textsf{P}_\tau|-1)$. Since role sets of ancestors in $\textsf{P}_\tau$ are disjoint, no query has to retrieve the same vector
twice.}

\begin{theorem}[Purity Preservation]
  \label{thm:node_purity_after_copying}
  Let $\textsf{P}_\tau\in\mathbb{P}(N(\tau))$ and suppose every node in $\mathcal{L}$ is
  pure for its own role set before the copy. After copying $N(\tau)$ into ancestors in
  $\textsf{P}_\tau$ and deleting $N(\tau)$, every surviving node is still pure
  for its own role set.
  \end{theorem}

The proof of Theorem~\ref{thm:node_purity_after_copying} is deferred to Appendix~\S\ref{app:effveda-proofs}.
Since each ancestor $N_a(\tau_j)$ in $\textsf{P}_\tau$ remains pure for $\tau_j$ after receiving vectors from $N(\tau)$, and the role sets of ancestors in $\textsf{P}_\tau$ are disjoint, the change to the query plan $\mathsf{QP}$ is ``local'': only one index in $\mathsf{QP}(r)$ ($r\in\tau$) is affected, and that index is known in advance. For any role $r\in\tau_j$, 
its query probes both $N_a(\tau_j)$ and $N(\tau)$ before the copy, but probes only the enlarged $N_a(\tau_j)$ afterwards; \emph{every other index in $\mathsf{QP}(r)$ is unchanged}.
The global cost reduction is therefore the sum of these per-role reductions over the affected roles, and depends only on the cardinality of each $\tau_j$ and the pre-copy sizes of $N(\tau)$ and $N_a(\tau_j)$.

\begin{definition}[Benefit Ratio for Copy]
\label{def:copy_benefit_effveda}
Let
$I(\cdot)$ denote an HNSW index built on vectors on a node. Given node $N(\tau)$ and its valid partition $\mathsf{P}_\tau$ ($|\mathsf{P}_\tau|\ge 2$), for each $N_a(\tau_j)\in \mathsf{P}_\tau$,
\emph{per-role gain} for each $r\in\tau_j$, denoted as $\Delta_c\!\bigl(N(\tau),N_a(\tau_j)\bigr)$, is defined as
\begin{equation}
\label{eqn:per_role_gain}
\begin{aligned}
  \Delta_c\!\bigl(N(\tau),N_a(\tau_j)\bigr)
  &= C_\theta\!\bigl(I(N_a(\tau_j)),\mathsf{efs}\bigr)
  + C_\theta\!\bigl(I(N(\tau)),\mathsf{efs}\bigr) \\
  &\quad
  - C_\theta\!\bigl(I(N_a(\tau_j)\cup N(\tau)),\mathsf{efs}\bigr).
\end{aligned}
\end{equation}
Thus, \emph{benefit ratio} over all roles in $\tau$, denoted as $f\bigl(N(\tau),\textsf{P}_\tau\bigr)$, is:
\begin{equation}
\label{eqn:copy_benefit_eff_adaptive_simplified}
  f\bigl(N(\tau),\textsf{P}_\tau\bigr)
  \;=\;
  \frac{\sum_{\tau_j\in\textsf{P}_\tau}|\tau_j|\cdot
        \Delta_c\!\bigl(N(\tau),N_a(\tau_j)\bigr)}
       {|N(\tau)|\cdot(|\textsf{P}_\tau|-1)}.
\end{equation}
\end{definition}

\begin{lemma}
\label{lem:copy_locality}
Under a uniform single-role query workload, copying $N(\tau)$ into a valid partition 
$\textsf{P}_\tau$ reduces $\textsf{AvgCost}$ by $\tfrac{1}{|\mathcal{R}|}\sum_{\tau_j\in\textsf{P}_\tau}|\tau_j|\cdot
 \Delta_c(N(\tau),N_a(\tau_j))$,
i.e., Equation~\ref{eqn:copy_benefit_eff_adaptive_simplified} is equivalent with Equation~\ref{eqn:benefit_function}.
\end{lemma}

\begin{lemma}
\label{lem:copy_positivity}
$\Delta_c(N(\tau),N_a(\tau_j))>0$ for any nonempty $N(\tau)$ and $N_a(\tau_j)$, hence
$f(N(\tau),\textsf{P}_\tau)>0$ for every $\textsf{P}_\tau\in\mathbb{P}(N(\tau))$.
\end{lemma}

{Proofs for Lemma~\ref{lem:copy_locality} and Lemma~\ref{lem:copy_positivity} are deferred to Appendix~\S\ref{app:effveda-proofs}.
Lemma~\ref{lem:copy_locality} is the source of \effname's efficiency:
Equation~\ref{eqn:copy_benefit_eff_adaptive_simplified} depends only on
$|N(\tau)|$, $|\mathsf{P}_\tau|$, and the node sizes and role-set cardinalities of the ancestors in
$\textsf{P}_\tau$. Moreover, each time, the algorithm scores all $|\textsf{P}_\tau|$ pairs at once, whereas
\name re-derives a full query plan
($\mathcal{O}(|\mathcal{R}|\!\cdot\!|\mathcal{N}|)$) after every single copy.}

\noindent\textbf{Default two-way valid partitions.} 
{\effname defaults to
two-way valid partitions since they add exactly one extra copy of $N(\tau)$,
the minimum nonzero overhead. Larger partitions expose gains to more ancestors
but each extra ancestor costs a full copy and exhaustive enumeration is
combinatorial. \effname therefore first scores one ancestor at a time. For each ancestor $N_a(\tau_j)$, it records the benefit of copying $N(\tau)$ into $N_a(\tau_j)$. If its complementary
ancestor $N_a(\tau_{j^\prime})$ exists in $A_c^\tau$ (i.e., $\tau_{j^\prime}\cup \tau_j=\tau$), it adds up the benefits of copying $N(\tau)$ into $N_a(\tau_j)$ and $N_a(\tau_{j^\prime})$. If the
highest-benefit candidate contains two ancestors, \effname uses it
directly; otherwise, \effname keeps that
ancestor as the seed and greedily adds further ancestors whose role sets are
disjoint subsets of the uncovered roles.
If $\mathsf{P}_\tau$ covers all roles in $\tau$, the source node $N(\tau)$ is deleted after being copied into the ancestors in $\mathsf{P}_\tau$.
If a residual role set, denoted as $\tau^{\mathsf{res}}$, remains uncovered after the greedy extension, $N(\tau)$ is relabeled
to $\tau^{\mathsf{res}}$. This moves no vectors,
contributes zero gain for roles in $\tau^{\mathsf{res}}$, and preserves
purity, while the copied ancestors still obtain the local reduction of
Lemma~\ref{lem:copy_locality}.}

\begin{algorithm}[t]
\caption{\textbf{\effname-Copy}}
\label{alg:eff-veda-copy}
\begin{algorithmic}[1]
\Require exclusive lattice $\mathcal{L}_\mathit{ex}$ and \textsf{SA} budget $\beta$.
\State $\mathcal{L}\gets \mathcal{L}_\mathit{ex}$;
       $\mathit{buf}\gets (\beta-1)\cdot|\mathcal{D}|$
       \label{ln:copy-init}
\If{$\mathit{buf}\le 0$} \Return $\mathcal{L}$ \EndIf
\For{$\ell \gets \mathsf{Depth}(\mathcal{L}_\mathit{ex})$ \textbf{down to} $2$}
       \label{ln:copy-layer}
       \State $\mathbb{P}_\ell\gets\emptyset$ \Comment{Best valid partitions for all nodes on layer $\ell$}
       \ForAll{$N(\tau)$ on layer $\ell$ of $\mathcal{L}$}
       \label{ln:copy-scan}
       \If{$|N(\tau)|>\mathit{buf}$} \textbf{continue} \EndIf
       \label{ln:copy-skip}
       \State $A_c^\tau\gets\{N_a(\tau')\in\mathcal{L}: \tau'\subsetneq\tau\}$ 
       \label{ln:copy-anc}
       \State $(\textsf{P}_\tau,f_c)\gets
              \textsf{FindBestPartition}(N(\tau),A_c^\tau,\mathcal{L},\mathit{buf})$ \Comment{Algorithm~\ref{alg:eff-veda-copy-find-best-partition}; $\tau^{\mathsf{res}}$ for $N(\tau)$ (if exists) included in $\mathsf{P}_\tau$}
       \label{ln:copy-find}
       \If{$\textsf{P}_\tau\neq\emptyset$}~$\mathbb{P}_\ell\gets\mathbb{P}_\ell\cup\{(N(\tau),\textsf{P}_\tau,f_c)\}$
       \EndIf
       \EndFor
       \State Sort $\mathbb{P}$ by $f_c$ in descending order 
       \label{ln:copy-sort}
       \ForAll{$(N(\tau),\textsf{P}_\tau,f_c)\in\mathbb{P}$}
       \label{ln:copy-commit}
       \State $\Delta S\gets |N(\tau)|\cdot(|\textsf{P}_\tau|-1)$
       \label{ln:copy-ds}
       \If{$\Delta S>\mathit{buf}$} \textbf{continue} \EndIf
       \ForAll{$\tau_j\in\textsf{P}_\tau$}
       \label{ln:copy-dist}
       \If{$N_a(\tau_j)\in\mathcal{L}$} \Comment{$\tau_j$ is not $\tau^{\mathsf{res}}$}
              \State $N_a(\tau_j)\gets N_a(\tau_j)\cup N(\tau)$
       \Else~$N_a(\tau_j)\gets N(\tau)$ \Comment{Relabeling}
       \label{ln:copy-relabel}
       \EndIf
       \EndFor
       \State Delete $N(\tau)$ from $\mathcal{L}$;
              $\mathit{buf}\gets\mathit{buf}-\Delta S$
       \label{ln:copy-delete}
       \EndFor
\EndFor
\State \Return $\mathcal{L}$
\end{algorithmic}
\end{algorithm}

{Algorithm~\ref{alg:eff-veda-copy} processes the lattice bottom-up until layer~$2$, since singleton role sets do not have ancestors. It
initializes $\mathcal{L}$ from $\mathcal{L}_\mathit{ex}$ and the storage
buffer $\mathit{buf}$ from $\beta$ and $|\mathcal{D}|$. On each layer~$\ell$,
every node $N(\tau)$ that fits the remaining buffer is scored against its
ancestor set $A_c^\tau$, and its best partition  $\textsf{P}_\tau$, as well as the benefit ratio, denoted as $f_c$, is
collected in $\mathbb{P}_\ell$. As we process the lattice bottom-up and all nodes on $\ell$ are scored against
the {same unmodified upper layers}, their $f_c$ values are directly
comparable. $\mathbb{P}_\ell$ is then committed in descending $f_c$ order:
each $N(\tau)$ that still fits the buffer is copied into every ancestor in
$\textsf{P}_\tau$ and the source $N(\tau)$ is deleted. If a residual role set
$\tau^{\mathsf{res}}$ remains, a fresh key $\tau^{\mathsf{res}}$ is created
with the vectors of $N(\tau)$, equivalent to relabeling $N(\tau)$
rather than deleting it.}

\noindent\textbf{Best partition selection. }
{Algorithm~\ref{alg:eff-veda-copy-find-best-partition}
(Appendix~\S\ref{app:effveda-complementary-algorithms}) implements
Line~\ref{ln:copy-find} of Algorithm~\ref{alg:eff-veda-copy} without
enumerating all valid partitions of $N(\tau)$, i.e., $\mathbb{P}(N(\tau))$. For each ancestor
$N_a(\tau_j)\in A_c^\tau$, it forms the complement
$\tau_j'=\tau\setminus\tau_j$. If $N_a(\tau_j')$ exists in $\mathcal{L}$,
$\{N_a(\tau_j),N_a(\tau_j')\}$ is a two-way candidate; otherwise
$\{N_a(\tau_j)\}$ is recorded alone. Each candidate is scored with
Equation~\ref{eqn:copy_benefit_eff_adaptive_simplified} and only the best is
kept. If the best partition $\textsf{P}_\tau$ has a single ancestor $N_a(\tau_s)$ ($s$ for seed),
the residual $\tau^{\mathsf{res}}=\tau\setminus\tau_s$ is covered greedily by
adding further disjoint ancestors (largest role set cardinality first) until the buffer
is exhausted or $\tau^{\mathsf{res}}$ is empty. The final $\textsf{P}_\tau$
and any remaining $\tau^{\mathsf{res}}$ are recorded in $\mathbb{P}_\ell$.}


\begin{example}
\label{ex:copy_walkthrough}
{On the lattice in Figure~\ref{fig:lattice_ex}, Algorithm~\ref{alg:eff-veda-copy}
starts at layer~$3$ with $N_c(\{r_1,r_2,r_3\})$. \textsf{FindBestPartition}
pairs $\tau'{=}\{r_1,r_2\}$ with its complement $\{r_3\}$, both in $A_c^\tau$,
giving $\textsf{P}_\tau=\{N_a(\{r_1,r_2\}),N_a(\{r_3\})\}$ with benefit
$(2\Delta_c(N_c(\{r_1,r_2,r_3\}),N_a(\{r_1,r_2\}))$
$+\Delta_c(N_c(\{r_1,r_2,r_3\}),$ $N_a(\{r_3\})))/|N_c(\{r_1,r_2,r_3\})|$; other
ancestors yield analogous candidates. If this partition wins,
$N_c(\{r_1,r_2,r_3\})$ is copied into $N_a(\{r_1,r_2\})$ and $N_a(\{r_3\})$ and then deleted. Queries with
$r_3$ drops from four indices to three, saving
$\Delta_c(N_c(\{r_1,r_2,r_3\}),N_a(\{r_3\}))$.}
\end{example}

\smallskip\noindent\textbf{\textit{Complexity.}} 
\textsf{FindBestPartition} performs $|A_c^\tau|$ iterations with
$\mathcal{O}(1)$ work for each node, i.e., $\mathcal{O}(|\mathcal{N}|)$, where $|\mathcal{N}|$ is the number of nodes in $\mathcal{L}_\mathit{ex}$. The algorithm iterates through all nodes on each layer in $\mathcal{O}(|\mathcal{N}|)$ time in total. The total complexity is
$\mathcal{O}(|\mathcal{N}|^2)$. 

\subsection{Phase 2: Merging}
\label{sec:eff_veda_merging}

The copy phase ensures each node in $\mathcal{L}$ remains pure to their role sets, yet leaves many of them below the indexing 
threshold $\Lambda$. The merging phase grows these pure nodes without consuming storage by absorbing nodes with overlapping role sets (i.e., ancestors, descendants, and siblings; see Appendix~\S\ref{app:lattice-relations}), until
they reach the indexing threshold $\Lambda$. 

Different from copy operations in \effname, merge introduces impurity, thus the benefit score must account for an
inflated beam width (Definition~\ref{def:impurity-factor} in \S\ref{sec:ac-indexing}). 
\effname leverages a lightweight \emph{virtual decomposition} of each node on $\mathcal{L}$, where only roles on the virtual decomposition would be routed to the node, no matter what exclusive blocks the node contains. 

\begin{definition}[Virtual Decomposition]
  \label{def:virtual_decomp}
  Denote each node that survives the copy phase as $N^\rho(\tau_\rho)$ and the
  post-copy lattice as $\mathcal{L}_\rho$. Define $\mathbb{V}$ as the virtual
  decomposition of each current node in $\mathcal{L}$ in terms of frozen
  post-copy nodes from $\mathcal{L}_\rho$.
  At the beginning of the merge phase, $\mathcal{L}$ is initialized with $\mathcal{L}_\rho$; for each node $N(\tau)\in\mathcal{L}$ (i.e., $N(\tau) = N^\rho(\tau_\rho)$ for now), its virtual decomposition is initialized as $\mathbb{V}(N(\tau))=\{N^\rho(\tau_\rho)\}$.
  Whenever $N(\tau)$ absorbs $N(\tau^\prime)$,
  $\mathbb{V}(N(\tau))$ is set to $\mathbb{V}(N(\tau))\cup\mathbb{V}(N(\tau^\prime))$, and
  $\mathbb{V}(N(\tau^\prime))$ is discarded.
\end{definition}

\effname leverages virtual decompositions for estimating the benefit of each operation while tracking changes in nodes in $\mathcal{L}$, without rebuilding any query plan. Since nodes produced by the copy phase are pure, a merged node is a union of pure \emph{virtual decomposition}  whose constituent sizes determine per-role impurity.
Since each merge deletes the absorbed node, every frozen
$N^\rho(\tau_\rho)$ in $\mathcal{L}_\rho$ belongs to exactly one
$\mathbb{V}(N(\tau))$ at any point during the merge phase; the family
$\{\mathbb{V}(N(\tau))\}_{N(\tau)\in\mathcal{L}}$ therefore partitions
$\mathcal{L}_\rho$, and the \emph{inherited plan}, i.e., route each $N^\rho(\tau_\rho)$ to its unique current container, remains a
valid cover of $\mathcal{D}(r)$ for every $r$.

\begin{theorem}[Inherited Routing Invariant]\label{th:routed_query_for_roles}
  Define $\Pi(N(\tau))$ as the set of roles whose post-copy query components are routed to $N(\tau)$ by the merge phase. Then for every current node $N(\tau)\in\mathcal{L}$,
$\Pi(N(\tau))=\bigcup_{N^\rho(\tau_\rho)\in\mathbb{V}(N(\tau))}\tau_\rho.$
\end{theorem}

The proof of Theorem~\ref{th:routed_query_for_roles} is deferred to Appendix~\S\ref{app:effveda-proofs}.
Theorem~\ref{th:routed_query_for_roles} characterizes the routing that
\effname maintains for merge scoring; it is a \emph{valid} plan, not necessarily the minimum-cost one. 
Since a merge will be admitted only when the re-evaluated benefit is positive (similar with \name),
the inherited-plan cost is non-increasing over the merge phase, and since the inherited plan is one
feasible cover for roles in $\Pi(N(\tau))\cup \Pi(N(\tau^\prime))$, we have
\[
\textsf{AvgCost}^{\textsf{opt}}_{\text{final}}
\;\le\;\textsf{AvgCost}^{\textsf{inh}}_{\text{final}}
\;\le\;\textsf{AvgCost}^{\textsf{inh}}_{0}
\;=\;\textsf{AvgCost}_{\text{post-copy}}.
\]
Thus \effname needs no probabilistic claim that roles are routed to merged
nodes. Even when \S\ref{sec:qplan} might select a cheaper alternative container
for some block, the merge phase has already guaranteed no regression relative
to the post-copy layout, and the final plan can only improve on that bound.

\begin{definition}[Impurity from Virtual Decomposition]\label{def:impurity_of_merged_nodes_effveda}
For a current node $N(\tau)$ and a routed role $r\in\Pi(N(\tau))$, define the pure size for $r$ in $N(\tau)$ as
$\omega(N(\tau), r)=
\sum_{N^\rho(\tau_\rho)\in\mathbb{V}(N(\tau)): r\in\tau_\rho}
|N^\rho(\tau_\rho)|.$ The impurity of $N(\tau)$ for role $r$ is
$\lambda_{N(\tau)}^r=
\frac{|N(\tau)|}{\omega(N(\tau), r)}$.
\end{definition}


\begin{definition}[Merge Benefit]
\label{def:merge_benefit_effveda}
Consider two nodes $N(\tau), N(\tau^\prime)$ for merge. 
For any node $N_i$ with roles routed to $N_i$ as $\Pi(N_i)$, define
$H(N_i)
  = \sum_{r\in \Pi(N_i)}\mathsf{Cost}_{\textsf{H}}(\mathcal{I}(N_i), r)
   = |\Pi(N_i)|\,a\log_2(|N_i|+1)
  + \sum_{r\in \Pi(N_i)}\bigl(b\,\lambda_{N_i}^r\mathsf{efs}+c\bigr),$
where $\lambda_N^r$ is computed as in
Definition~\ref{def:impurity_of_merged_nodes_effveda}. The benefit of the
merge is computed as $H(N(\tau))+H(N(\tau^\prime))-H(N(\tau)\cup N(\tau^\prime)).$
\end{definition}

To evaluate a merge, \effname inspects $\mathbb{V}$ to compute
inflation ratios $\lambda$ for the impacted roles in
$\Pi(N(\tau))\cup\Pi(N(\tau^\prime))$. By the same locality argument as
Lemma~\ref{lem:copy_locality}, all other roles keep the same inherited plan.
Since merging
has zero storage overhead, its benefit ratio is exactly the reduction in
$\textsf{AvgCost}$ \emph{under the inherited routing}.

\begin{algorithm}[t]
\caption{\textbf{\effname-Merge}}
\label{alg:eff-veda-merge}
\begin{algorithmic}[1]
\Require post-copy lattice $\mathcal{L}_\rho$; indexing threshold $\Lambda$.
\State $\mathcal{L}\gets\mathcal{L}_\rho$, $\mathbb{V}.\mathsf{initialize}(\mathcal{L})$, $\Pi.\mathsf{initialize}(\mathcal{L})$ \Comment{Definition~\ref{def:virtual_decomp}}
          \label{ln:mg-init}
\State $i\gets 0$, get $\mathcal{N}_\mathit{desc}$: sort nodes in $\mathcal{L}$ by size, descending
          \label{ln:mg-sort}
\While{$i\le|\mathcal{N}_\mathit{desc}|$} \Comment{Prioritize large unindexable nodes}
          \label{ln:mg-loop}
     \State $N(\tau)\gets\mathcal{N}_\mathit{desc}[i]$
     \If{$N(\tau)\notin\mathcal{L}$ \textbf{or} $|N(\tau)|\ge\Lambda$}
          $i\gets i{+}1$; \textbf{continue}
          \label{ln:mg-skip}
     \EndIf
     \State $\textsf{RL}\gets\mathsf{get\_candidates}(N(\tau),\mathcal{L})$, $\mathbb{B}_m\gets\emptyset$
          \label{ln:mg-rel}
     \For{$N(\tau^\prime)\in\textsf{RL}$} \Comment{Compute benefit with Definition~\ref{def:merge_benefit_effveda}} 
          \State $\Delta_m\gets \mathsf{eff\_merge\_benefit}(N(\tau),N(\tau'), \Pi)$ 
          \State $\mathbb{B}_m\gets\mathbb{B}_m\cup\{(N(\tau'),\Delta_m)\}$
     \EndFor
     \State Get $\mathbb{B}_m^{\mathit{desc}}$: sort $\mathbb{B}_m$ by $\Delta_m$ in descending order
          \label{ln:mg-score}
     \ForAll{$(N(\tau'),\Delta_m)\in\mathbb{B}_m^{\mathit{desc}}$ with $N(\tau')\in\mathcal{L}$}\label{ln:mg-recheck}
     \If{$\Delta_m\le 0$} \textbf{break} \EndIf\label{ln:mg-recheck}
     \If{first merge}~$\mathsf{MergeNodes}(N(\tau), N(\tau'), \mathbb{V}, \Pi)$
     \Else~Renew $\Delta_m$, if $>0$, $\mathsf{MergeNodes}(N(\tau), N(\tau'), \mathbb{V}, \Pi)$
     \EndIf
     \If{$|N(\tau)|\ge\Lambda$} \textbf{break} \EndIf\label{ln:mg-thresh}
     \EndFor
     \If{$|N(\tau)|\geq\Lambda$}~$i\gets i{+}1$ \Comment{Process next node}\label{ln:mg-next}
     \Else~$\mathcal{N}_\mathit{desc}[i]\gets N(\tau)$ \Comment{Process $N(\tau)$ again}
          \label{ln:mg-retry}
     \EndIf
\EndWhile
\State \Return $\mathcal{L}$

\Function{MergeNodes}{$N(\tau), N(\tau'), \mathbb{V}, \Pi$}
\State Merge: $N(\tau)\gets N(\tau)\cup N(\tau')$ \Comment{Definition~\ref{def:virtual_decomp}}
\State $\mathbb{V}.\mathsf{update}()$, $\Pi.\mathsf{update}()$, delete $N(\tau')$ from $\mathcal{L}$, $\mathbb{V}$, and $\Pi$ 
\EndFunction
\end{algorithmic}
\end{algorithm}

{Algorithm~\ref{alg:eff-veda-merge} describes the merge phase. It initializes $\mathcal{L}$,
$\mathbb{V}$, and $\Pi$ from the post-copy lattice $\mathcal{L}_\rho$
(Line~\ref{ln:mg-init}) and visits nodes in descending size. Large
unindexable nodes near $\Lambda$ are grown first; indexable or
already-absorbed nodes are skipped (Line~\ref{ln:mg-skip}). For each $N(\tau)$,
ancestors, descendants, and siblings are scored by
Definition~\ref{def:merge_benefit_effveda} and consumed in descending benefit,
ignoring candidates already removed from $\mathcal{L}$ and stopping at the
first nonpositive benefit (Line~\ref{ln:mg-recheck}). The first merge uses the
cached score; later merges re-score $\Delta_m$ since earlier merges may have
changed $\Pi$ and $\mathbb{V}$, and apply only if the updated benefit stays
positive. Processing advances once $|N(\tau)|\geq\Lambda$
(Line~\ref{ln:mg-next}); otherwise the node is revisited with a fresh
candidate set (Line~\ref{ln:mg-retry}).}

\begin{example}
Continuing Example~\ref{ex:copy_walkthrough}, after Phase~1 the lattice holds
$N^\rho(\{r_1\})$, $N^\rho(\{r_2\})$, and $N^\rho(\{r_3\})$ on layer~1 alongside
$N^\rho(\{r_1,r_2\})$, $N^\rho(\{r_1,r_3\})$, and $N^\rho(\{r_2,r_3\})$ on layer~2, each with
$\mathbb{V}(N(\tau))=\{N^\rho(\tau_\rho)\}$. Suppose $\Lambda=15{,}000$, and suppose we merge
$N^\rho(\{r_2,r_3\})$ and $N^\rho(\{r_3\})$. The merged node includes $N^\rho(\{r_2,r_3\})$ and $N^\rho(\{r_3\})$, so $r_3$ reads it with $\lambda^{r_3}=1$ while $r_2$ reads it with $\lambda^{r_2}=5$; both factors are obtained from
$\mathbb{V}$ alone, without touching $\textsf{QP}$.
\end{example}

\smallskip\noindent\textbf{Choice of \name and \effname.}
\name explores the full descendant--ancestor candidate space and yields slightly
lower \textsf{QA} (Exps~\ref{exp:qa-vs-sa} and~\ref{exp:indexing-threshold}) at $\mathcal{O}(|\mathcal{N}|^5|Q|)$
construction cost; \effname reaches a comparable layout in
$\mathcal{O}(|\mathcal{N}|^2)$ per phase. We recommend \effname when
the number of distinct role combinations is large in the dataset or when policies change frequently enough that
re-indexing dominates (Appendix~\ref{app:dynamic-workloads}), and \name
when the number of roles is small and static and the last few percent of \textsf{QA} matter.

  \section{Query Answering}
\label{sec:query-execution}

This section describes query answering with the indices constructed by \name or \effname. 
We note that, after obtaining $\mathcal{L}$ with \name or \effname, indexable nodes in $\mathcal{L}$ ($\geq \Lambda$) materialize HNSW indices, while unindexable nodes ($< \Lambda$) are decomposed into exclusive blocks and kept as leftovers for linear scan. Given a query $q=(\mathbf{x},r)$, two tasks remain:
\textit{(i)}~select a minimal coverage for each role $r$ such that $\mathcal{D}(r)$, the authorized data for $r$ is covered (\S\ref{sec:qplan}); and
\textit{(ii)}~jointly search across the selected HNSW indices and leftovers while filtering
unauthorized data out of impure indices (\S\ref{sec:coordinated-search}).

\begin{figure}
    \centering
    \includegraphics[width=0.498\textwidth]{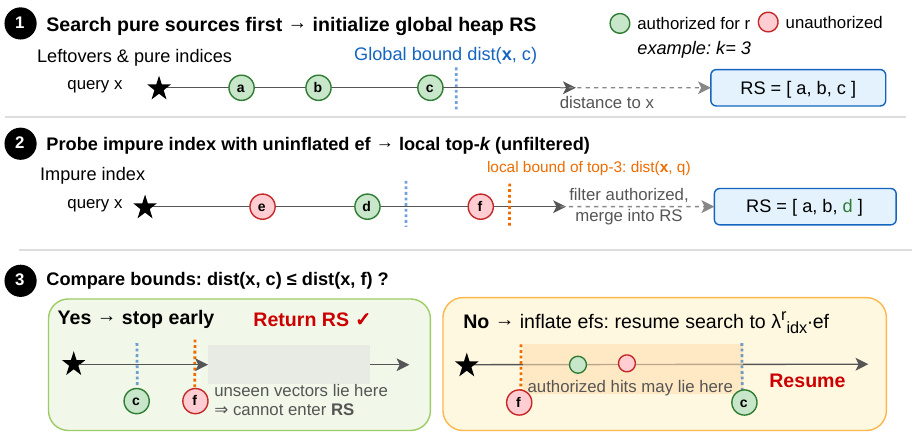}
    \caption{Illustration of coordinated search.}
    \label{fig:coord_search_example}
\end{figure}

\begin{algorithm}[h!]
\caption{\textbf{Coordinated Search}}
\label{alg:coordinated-search}
\begin{algorithmic}[1]
    
\Require query $q$ with role $r$, query plan $\textsf{QP}(r)$, authorized data IDs $\mathcal{D}_\mathit{rids}(r)$ for filtering, search parameter $k$ and $\textsf{efs}$.

\State $\mathcal{U}(r), \mathcal{I}(r) \gets \textsf{QP}(r)$, $\textsf{RS}\gets\textsf{BruteForceSeasrch}(\mathcal{U}(r), k)$ 

Get pure indices$\mathcal{I}^*_\mathit{pure}(r)$ and impure indices $\mathcal{I}^*_\mathit{impure}(r)$.

\State $\mathcal{I}^*_\mathit{impure}(r)\gets \mathcal{I}^(r)\backslash\mathcal{I}^*_\mathit{pure}(r)$


\For{each ${\tt idx}\in \mathcal{I}_\mathit{pure}^*(r)$}
    \State $\textsf{RS}_{{\tt idx}}\gets \textsf{HNSW}({\tt idx}, k, r, \textsf{efs})$, $\textsf{RS}\gets \texttt{merge}(\textsf{RS}_{{\tt idx}}, \textsf{RS})$
\EndFor

\For{each ${\tt idx}\in \mathcal{I}_\mathit{impure}^*(r)$}  
\State  $v_k\gets \textsf{RS}[k]$, $\textsf{RS}_{{\tt idx}}\gets \textsf{HNSW}({\tt idx}, k, r, \textsf{efs})$ 



\State $\textsf{RS}\gets \texttt{merge}(\textsf{RS}_{{\tt idx}}, \textsf{RS}, \mathcal{D}_\mathit{rids}(r), k)$
\If{$\textsf{dist}(\textsf{RS}_{{\tt idx}}[k], \mathbf{x})\geq \textsf{dist}(v_k, \mathbf{x})$}~ \textbf{break} 
    \EndIf

    \State $i\gets 0$, ${\tt idx}.\textsf{PQ.clear}()$, compute $\lambda^r_{{\tt idx}}$\Comment{Definition~\ref{def:impurity-factor}} 
    
    \While{$i < (\lambda^r_{{\tt idx}} - 1)\textsf{efs}$}

    \State ${\tt idx}.\textsf{PQ.add}({\tt idx}.\textsf{NextVector}())$, $i\gets i + 1$
    \EndWhile
    \State $\textsf{RS}_{{\tt idx}}\gets {\tt idx}.\textsf{PQ}()$, $\textsf{RS}\gets \texttt{merge}(\textsf{RS}_{{\tt idx}}, \textsf{RS}, \mathcal{D}_\mathit{rids}(r), k)$
\EndFor

\State \textbf{return} $\textsf{RS}$.\textsf{top}($k$)

    

\end{algorithmic}
\end{algorithm}

\subsection{Query Plan Construction}
\label{sec:qplan}

Given lattice $\mathcal{L}$ produced by \name or \effname, the query plan
identifies nodes that are selected for queries with each role $r\in\mathcal{R}$. We write the plan as
$\textsf{QP}(r)=(\mathcal{I}(r),\mathcal{U}(r))$, where $\mathcal{I}(r)$ is the
set of HNSW indices selected for $r$ and $\mathcal{U}(r)$ is the set of
leftover vectors. $\textsf{QP}(r)$ is valid if it covers $\mathcal{D}(r)$, all authorized
data for $r$; among valid plans we prefer the one with minimum estimated cost,
i.e., a minimal coverage of $\mathcal{D}(r)$.

Let $\mathcal{N}^{\mathit{ex}}(r)=
\{N^{\mathit{ex}}(\tau)\in\mathcal{L}_{\mathit{ex}}\mid r\in\tau\}$
be the exclusive blocks authorized for $r$. Due to copy operations, each
exclusive block may present in several nodes in $\mathcal{L}$. Define $\Phi$, a mapping from exclusive blocks to nodes in $\mathcal{L}$, as follows:
$\Phi(N^{\mathit{ex}}(\tau))=
\{N(\tau')\in\mathcal{L}\mid N^{\mathit{ex}}(\tau)\subseteq N(\tau')\}.$ 
$\Phi$ specifies the ``locations’’ of each $N^\mathit{ex}(\tau)$ in $\mathcal{L}$.
For role $r$, if any $N^\mathit{ex}(\tau)\in\mathcal{N}^{\mathit{ex}}(r)$ exists only in one node $N(\tau')$ in $\mathcal{L}$, $N(\tau')$ is mandatory for $\textsf{QP}(r)$, i.e., 
\[
\{N(\tau')\in\mathcal{L}\mid
N(\tau')\in\Phi(N^{\mathit{ex}}(\tau))~\text{and}~
|\Phi(N^{\mathit{ex}}(\tau))|=1\}\subseteq\textsf{QP}(r).
\]
After fixing mandatory nodes, each remaining block $N^\mathit{ex}(\tau')$ in $\mathcal{N}^{\mathit{ex}}(r)$ is then covered by choosing
among $\Phi(N^{\mathit{ex}}(\tau'))$, which can be solved
with ILP (Algorithm~\ref{alg:query_plan_ILP} in Appendix~\S\ref{app:query-answering-algorithms}) or approximately with a greedy heuristic, i.e., prioritizing nodes that cover more remaining blocks (Algorithm~\ref{alg:query_plan_greedy} in Appendix~\S\ref{app:query-answering-algorithms}).

\begin{table*}[h]
    \small
    \centering
    \caption{Summary of Datasets. Distribution columns report shifted Zipf parameters $(s,\alpha)$; larger $\alpha$ yields higher skew.}
    \label{tab:dataset_summary}
    \setlength{\tabcolsep}{3.5pt}
        \begin{tabular}{@{}l | r | r | l | l | c | c | c | c | c  c@{}}
            \toprule
            \textbf{Dataset} & \makecell{\textbf{Data Size}} & \textbf{Dim.} & \textbf{Type} & \makecell{\textbf{AC}\\\textbf{Rules}} & \makecell{\textbf{\#~Roles}} & \makecell{\textbf{\#~Permissions}} & \makecell{\textbf{Permission}\\\textbf{Dist.}~\textbf{$(s',\alpha')$}} & \makecell{\textbf{Block}\\\textbf{Dist.}~\textbf{$(s,\alpha)$}}
            & \makecell{\textbf{SA w/}\\\textbf{Oracle}~\textbf{Idx}}\\ 
            \midrule
            SIFT-1M~\cite{sift}   & 1,000,000 & 128 & Image & Medium  & 82 & 757 & (2, 1.5) & (1, 1.5)   
            & 11.423\\  
            PAPER~\cite{paper}     & 2,029,997 & 200 & Text  & Hard   & 87 & 676 & (2, 1.5) & (1, 2)   & 5.255\\  
            AMZN~\cite{amazon_books_reviews,amazon_books_reviews2}      &   212,404 & 384 & Text  & Easy   & 64 & 641 & (1, 1.5) & (1, 2)   & 15.084\\  
            \bottomrule
        \end{tabular}
    \end{table*}

\subsection{Coordinated Query Execution}
\label{sec:topk-exec}
\label{sec:coordinated-search}

{$\textsf{QP}(r)$ might mix pure indices, impure indices, and leftovers
(\S\ref{sec:ac-indexing}). Leftovers and pure indices hold only authorized
data and are searched directly; impure indices require filtering and inflated
parameters (Definition~\ref{def:impurity-factor}). The straightforward strategy
searches each component independently, e.g., scanning $\mathcal{U}(r)$, running standard
HNSW on pure indices and inflated HNSW with
$k'=\lceil\lambda^r_{{\tt idx}} k\rceil$,
$\textsf{efs}'=\lceil\lambda^r_{{\tt idx}}\textsf{efs}\rceil$ on impure ones,
then filter and merge into a size-$k$ heap
(Appendix~\S\ref{app:topk-exec-details}; correctness in
Appendix~\S\ref{sec:inflated-k-accuracy}). This wastes work when pure sources
already supply strong authorized candidates.}

{Coordinated search (Figure~\ref{fig:coord_search_example}) is an
instance of threshold-based top-$k$ aggregation~\cite{fagin2001optimal}
adapted to graph ANN: the global $k$-th distance plays the role of Fagin's
threshold~\cite{fagin2001optimal} and each sub-index is a sorted-access source. Leftovers and pure
indices contain only authorized data, so their results initialize a global
heap $\mathsf{RS}$ of size $k$; the distance $\textsf{dist}(\mathbf{x},v_k^g)$
of its $k$-th entry $v_k^g$ is the global bound. Each impure index
${\tt idx}$ is first probed with the \emph{uninflated} $\textsf{efs}$, and 
authorized candidates are merged into $\mathsf{RS}$. Denote $v_k^{l}$ as
the $k$-th entry of the \emph{unfiltered} local result of ${\tt idx}$. If
$\textsf{dist}(\mathbf{x},v_k^g)\le\textsf{dist}(\mathbf{x},v_k^{l})$, then
under Assumption~\ref{assump:hnsw} no unseen vector in ${\tt idx}$ can
improve $\mathsf{RS}$ and the probe stops; otherwise it resumes with
$\mathsf{efs}$ inflated by $\lambda^r_{{\tt idx}}$ and merges the newly
found authorized candidates. Algorithm~\ref{alg:coordinated-search} gives
the procedure; the continued base-layer traversal under the global bound is
in Appendix~\S\ref{app:continued-hnsw}.}

\section{Evaluations}\label{sec:evaluations}

\newcounter{experiment}
\newcommand{\experiment}[2]{%
    \refstepcounter{experiment}%
    \noindent\textbf{Exp~\theexperiment. #1.}\label{#2}%
}

We evaluate \name and \effname with the following experiments: 
Exps~\ref{exp:index-creation-time}--\ref{exp:desired-vs-achieved-sa}
study index construction;
Exps~\ref{exp:qa-vs-sa}--\ref{exp:efs} study query cost, purity, and
parameter sensitivity;
Exps~\ref{exp:qps-recall-sift}--\ref{exp:multi-role} report end-to-end
QPS--recall on three datasets and different query workloads.

\subsection{Setup}
\noindent
\textbf{Setting.}
All experiments were conducted on a machine running macOS~26.3 (Darwin~25.3.0, ARM64) with an Apple M4~Max processor (14 cores) and 36~GiB of unified memory.

\noindent
\textbf{Datasets.}
We use SIFT-1M~\cite{sift}, PAPER~\cite{paper}, and Amazon Books Reviews
(AMZN)~\cite{amazon_books_reviews,amazon_books_reviews2}
(Table~\ref{tab:dataset_summary}). Access-control policies are generated
from OrgAccess~\cite{sanyal2025orgaccess}, which grants each role the
union of a set of departments. We use the \textit{Easy} subset (roles
clustered to a target count), and \textit{Medium}/\textit{Hard} unmodified.
Block sizes follow a shifted Zipf distribution
$(i+s)^{-\alpha}$
~\cite{zipf2016human,mandelbrot1953informational};
the number of blocks assigned to department $j$ follows
$(j+s')^{-\alpha'}$. 
Higher skew in this \textit{permission distribution} means that a
few departments are associated with substantially more data than
the rest.

\noindent
\textbf{Compared Methods.}
We compare against partition-based methods SIEVE~\cite{sieve} and
HoneyBee~\cite{honeybee}, in-search filtering methods ACORN-1 and
ACORN-$\gamma$~\cite{acorn}, and the Global (Baseline~1) and Oracle
(Baseline~2) indices. Oracle is an ideal reference rather than a practical
method: the last column of Table~\ref{tab:dataset_summary} shows that even
single-role oracle indexing incurs $\textsf{SA}{\in}[5.3,15.1]$, and richer
query predicates would require still more pure indices. For \emph{all} partition-based
methods we apply SIEVE's heterogeneous-search rule~\cite{sieve}, downscaling
\textsf{efs} on each sub-index by the log-ratio of its size to
$|\mathcal{D}|$ to avoid small partitions being over-searched. SIEVE is given a held-out workload sampled from the same role
distribution as the test queries.


\noindent\textbf{Metrics.}
{We report \textsf{SA} (total number of vectors to the data size), \textsf{QA} (query cost
normalized to Oracle, isolating algorithmic efficiency from hardware
constants), QPS (end-to-end throughput), and recall@$k$ (fraction of the
brute-force top-$k$ over $\mathcal{D}$ that appears in the returned set).}


\noindent
\textbf{Parameters.}
We set the parameter $M$ of HNSW to $16/32/32$ for
SIFT-1M/PAPER/AMZN. We set $\textsf{efs}$ and $\textsf{efc}$ to $100$ and $200$, respectively. We set $\gamma{=}1/m_s{=}12$ for 
ACORN-$\gamma$, where $m_s$ is {the \textit{minimum selectivity}}, i.e., the smallest fraction of the dataset admitted by any role.
Following~\cite{sieve}, we use a brute-force search bound of
1/12 for ACORN-1 and ACORN-$\gamma$. Without this bound, $\gamma$ reaches $23{,}810$ on
SIFT-1M and each node would store $\gamma M$ neighbors.
{$\Lambda$ is $2{,}900/3{,}000/2{,}600$ for SIFT-1M/PAPER/AMZN per
Figure~\ref{fig:vary-dimensions} and Figure~\ref{fig:vary-dimensions-app} (Appendix~\S\ref{sec:prelim-hnsw}). We set the default \textsf{SA} for all budgeted
methods (SIEVE, HoneyBee, ACORN-1, ACORN-$\gamma$) to $1.1$, 
highlight the advantages of data partitioning in vector search with
access control: even modest storage overhead yields large
gains. We also test under different \textsf{SA} values.
}

\begin{figure*}[!t]
    \centering
    \begin{minipage}{\linewidth}
        \centering
        \includegraphics[width=0.8\linewidth]{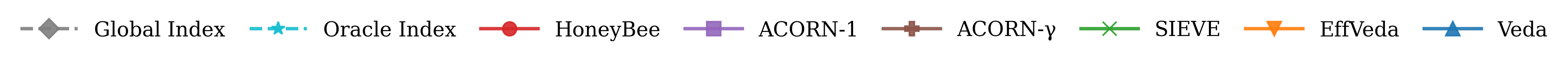}
    \end{minipage}
    \begin{minipage}[t]{0.39\textwidth}
        \centering
        \begin{subfigure}[t]{0.48\linewidth}
            \centering
            \includegraphics[width=\linewidth]{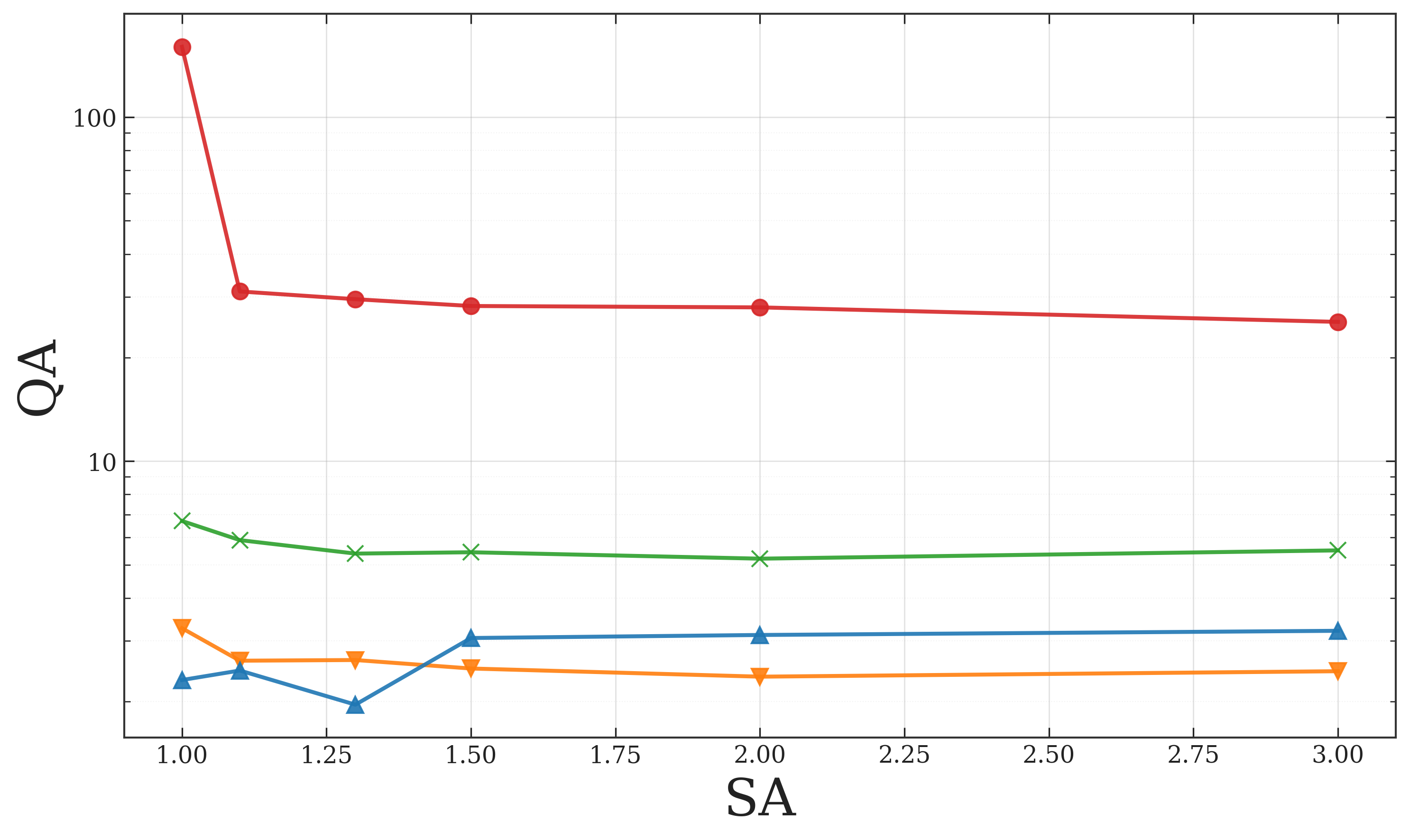}
            \caption{QA vs.\ SA.}
            \label{fig:qa_vs_sa_sift1m_single}
        \end{subfigure}
        \hfill
        \begin{subfigure}[t]{0.48\linewidth}
            \centering
            \includegraphics[width=\linewidth]{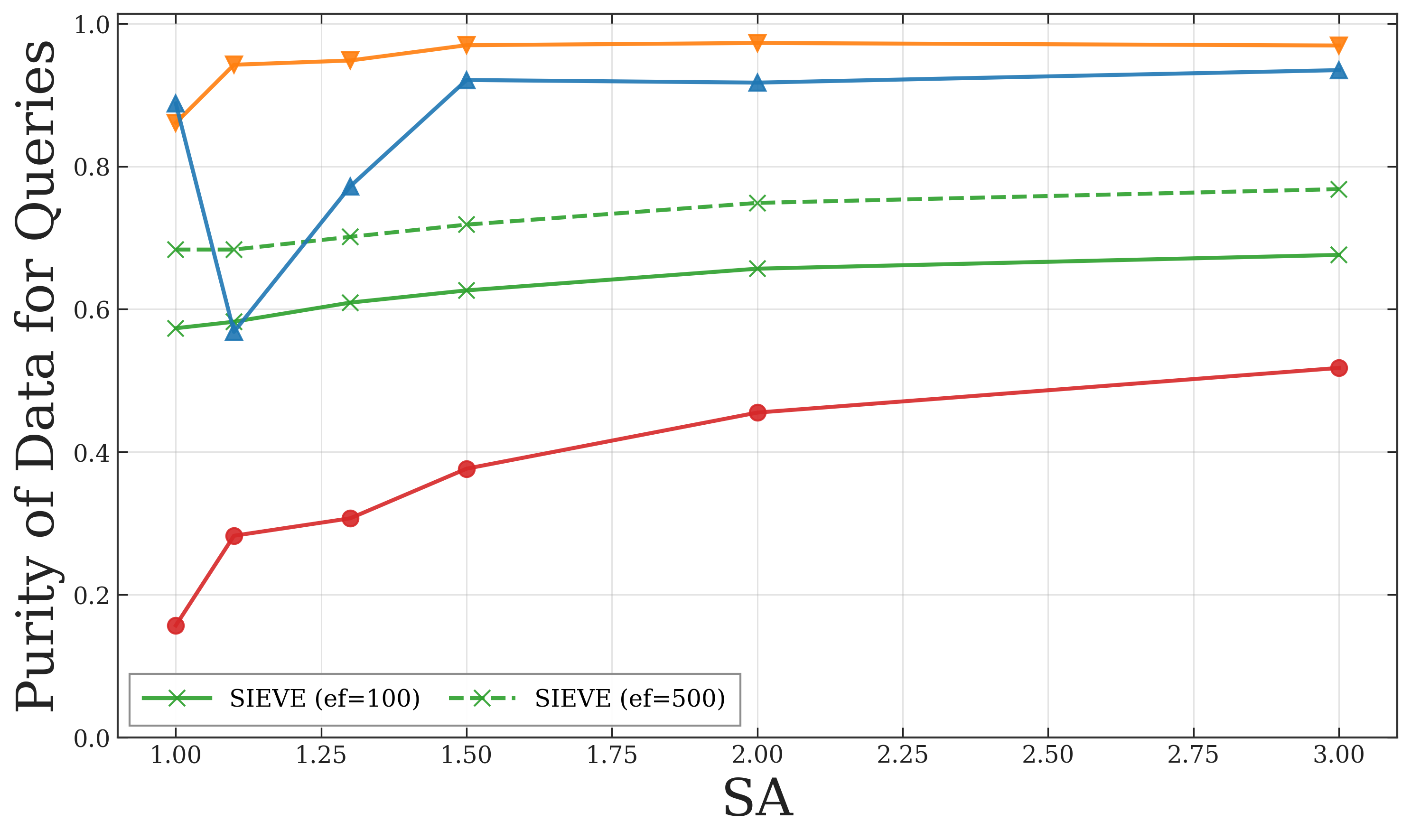}
            \caption{Purity vs.\ SA.}
            \label{fig:purity_vs_sa_ac}
        \end{subfigure}
        \caption{Comparisons with partitioning methods.}
        \label{fig:data_partitioning_comparison_sift}
    \end{minipage}
    \hfill
    \setcounter{subfigure}{0}
    \begin{minipage}[t]{0.59\textwidth}
        \centering
        \begin{subfigure}[t]{0.32\linewidth}
            \centering
            \includegraphics[width=\linewidth]{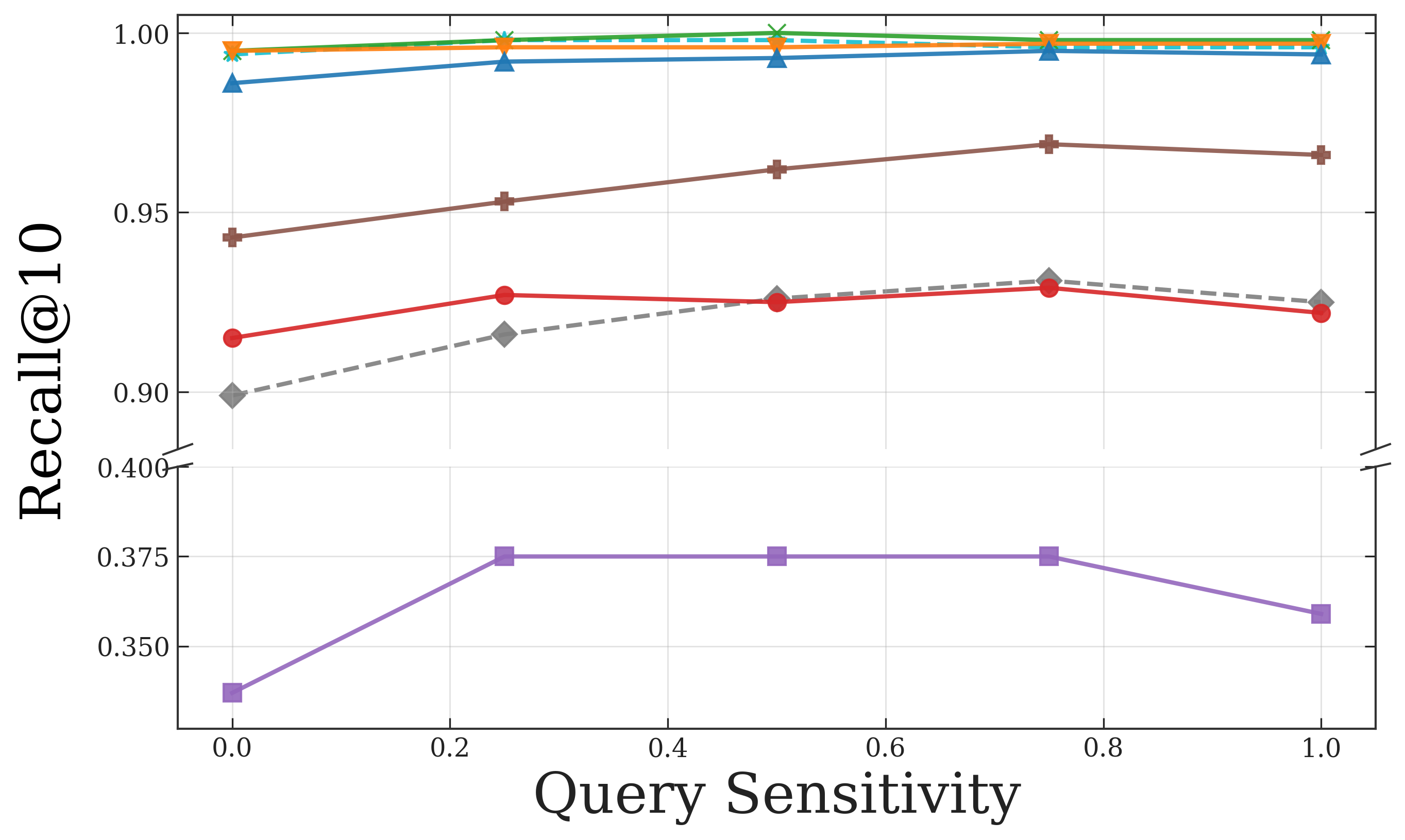}
            \caption{Sensitivity.}
            \label{fig:recall_vs_sensitivity_ac}
        \end{subfigure}
        \hfill
        \begin{subfigure}[t]{0.32\linewidth}
            \centering
            \includegraphics[width=\linewidth]{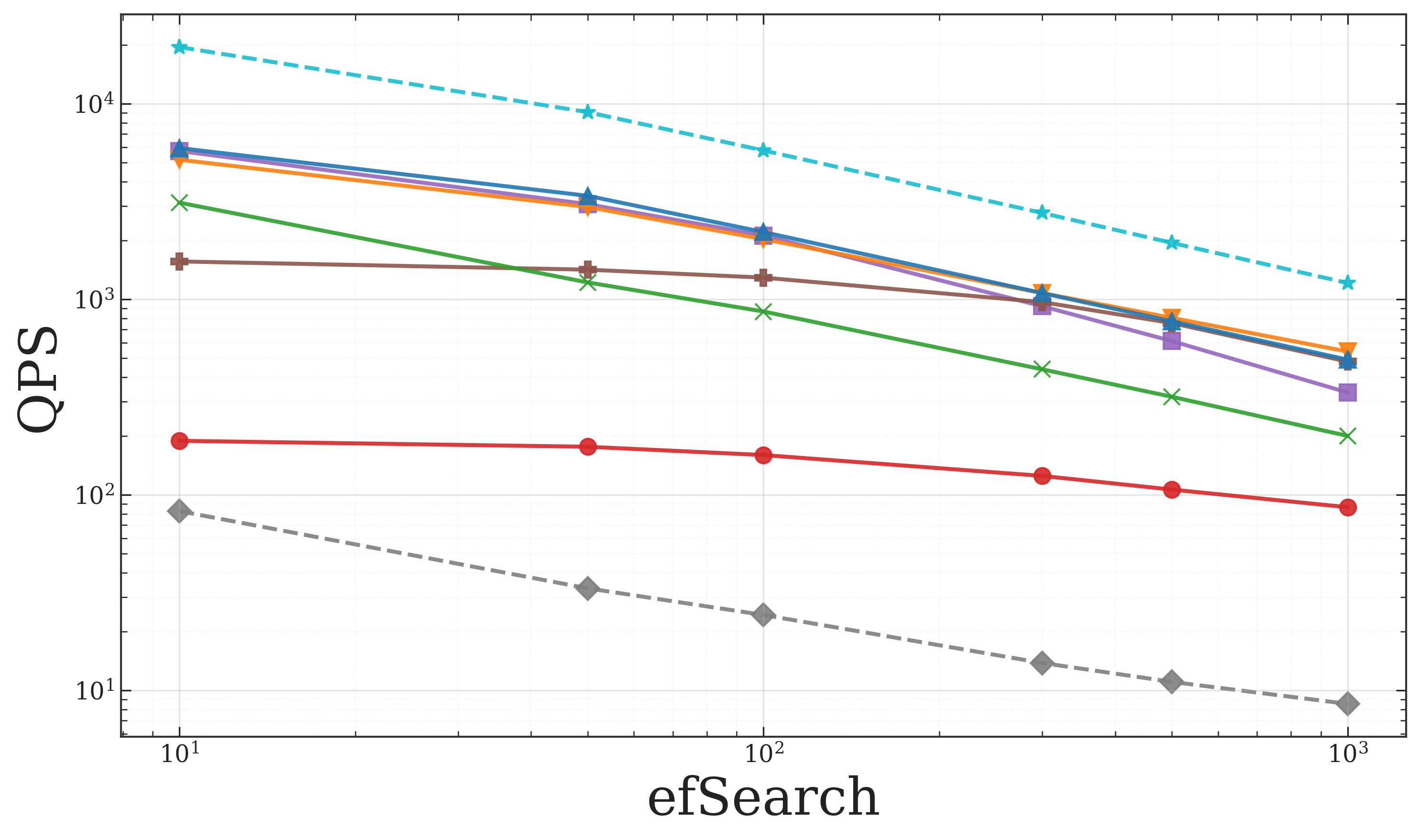}
            \caption{QPS vs.\ \textsf{efs}.}
            \label{fig:qps_vs_ef_ac}
        \end{subfigure}
        \hfill
        \begin{subfigure}[t]{0.32\linewidth}
            \centering
            \includegraphics[width=\linewidth]{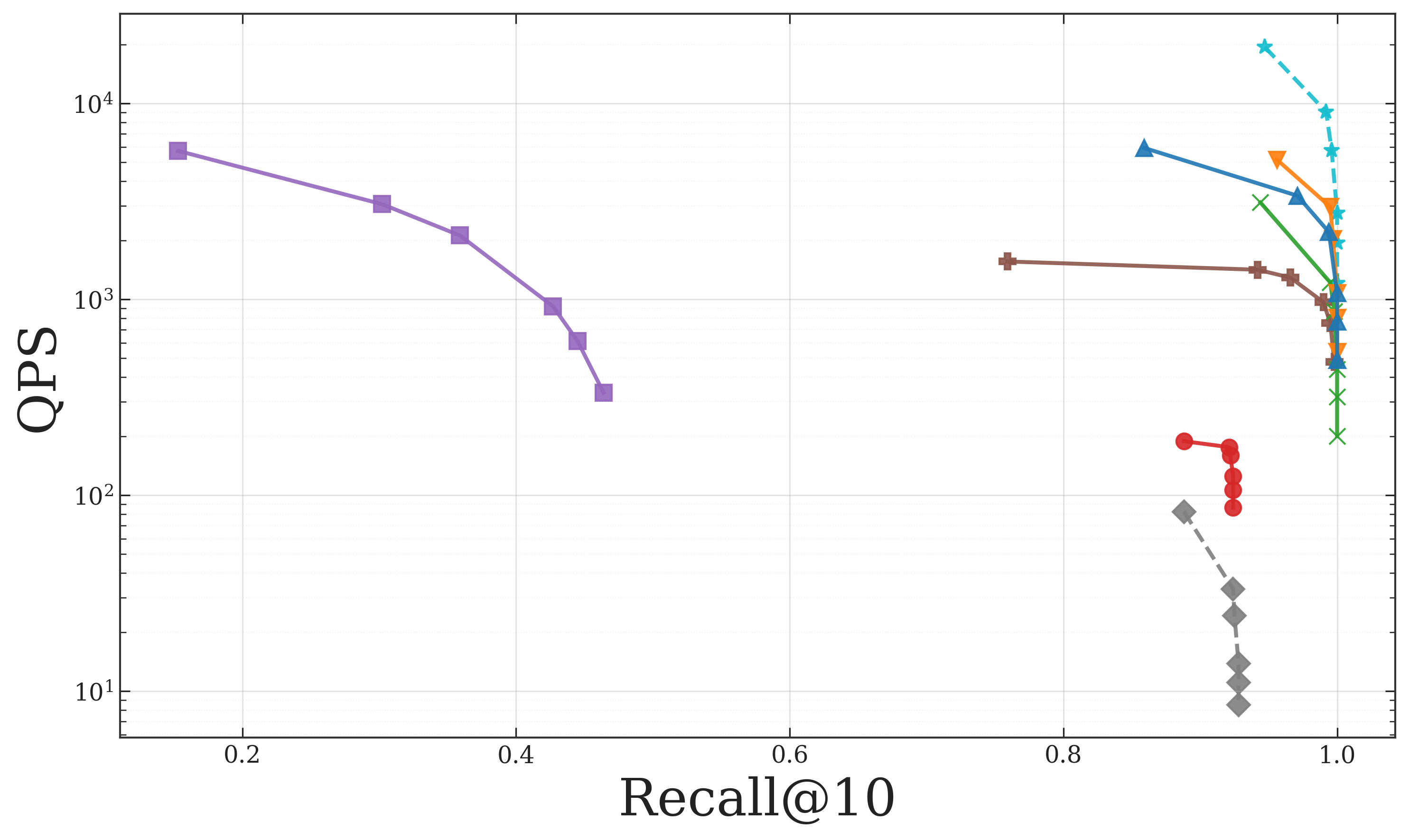}
            \caption{QPS vs.\ recall.}
            \label{fig:qps_vs_recall_sift}
        \end{subfigure}
        \caption{Single-role queries with SIFT-1M.}
        \label{fig:single_role_queries_sift}
    \end{minipage}
\end{figure*}

\begin{figure*}[!t]
    \centering
    \begin{minipage}{\linewidth}
        \centering
        \includegraphics[width=0.8\linewidth]{figures/legend.png}
    \end{minipage}
    \begin{subfigure}[t]{0.195\textwidth}
        \includegraphics[width=\linewidth]{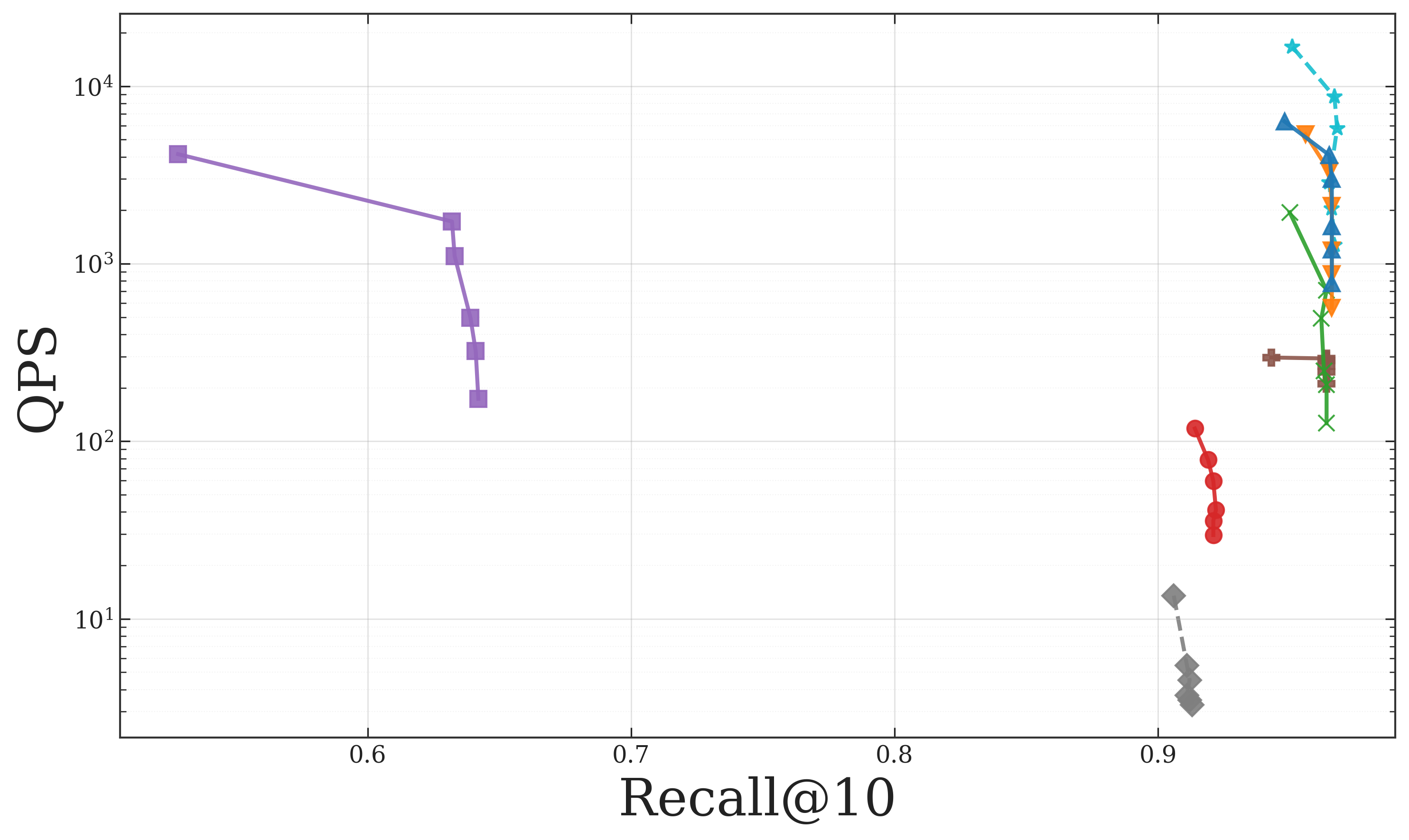}
        \caption{PAPER.}
        \label{fig:qps_vs_recall_paper}
    \end{subfigure}\hfill
    \begin{subfigure}[t]{0.195\textwidth}
        \includegraphics[width=\linewidth]{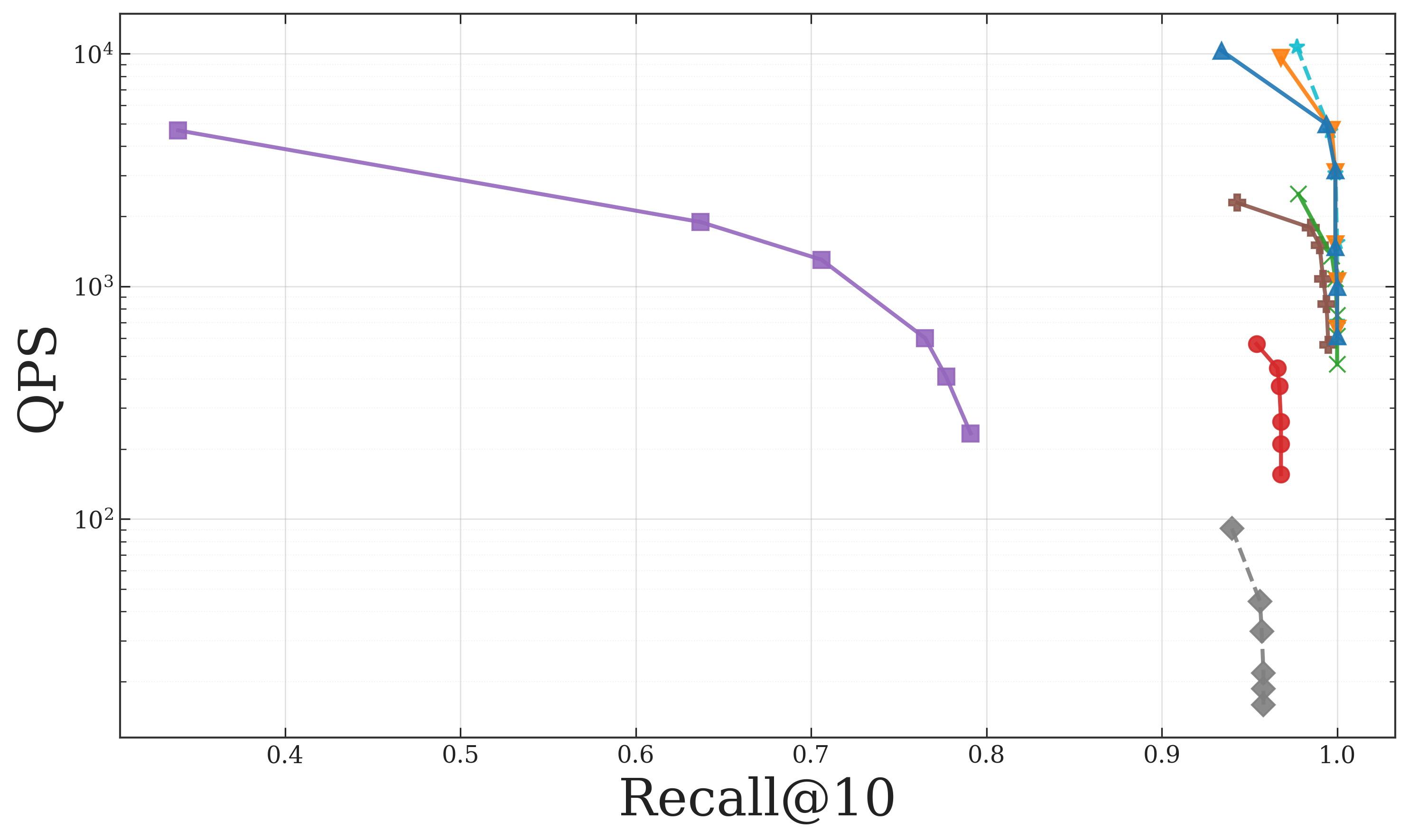}
        \caption{AMZN.}
        \label{fig:qps_vs_recall_amzn}
    \end{subfigure}\hfill
    \begin{subfigure}[t]{0.195\textwidth}
        \includegraphics[width=\linewidth]{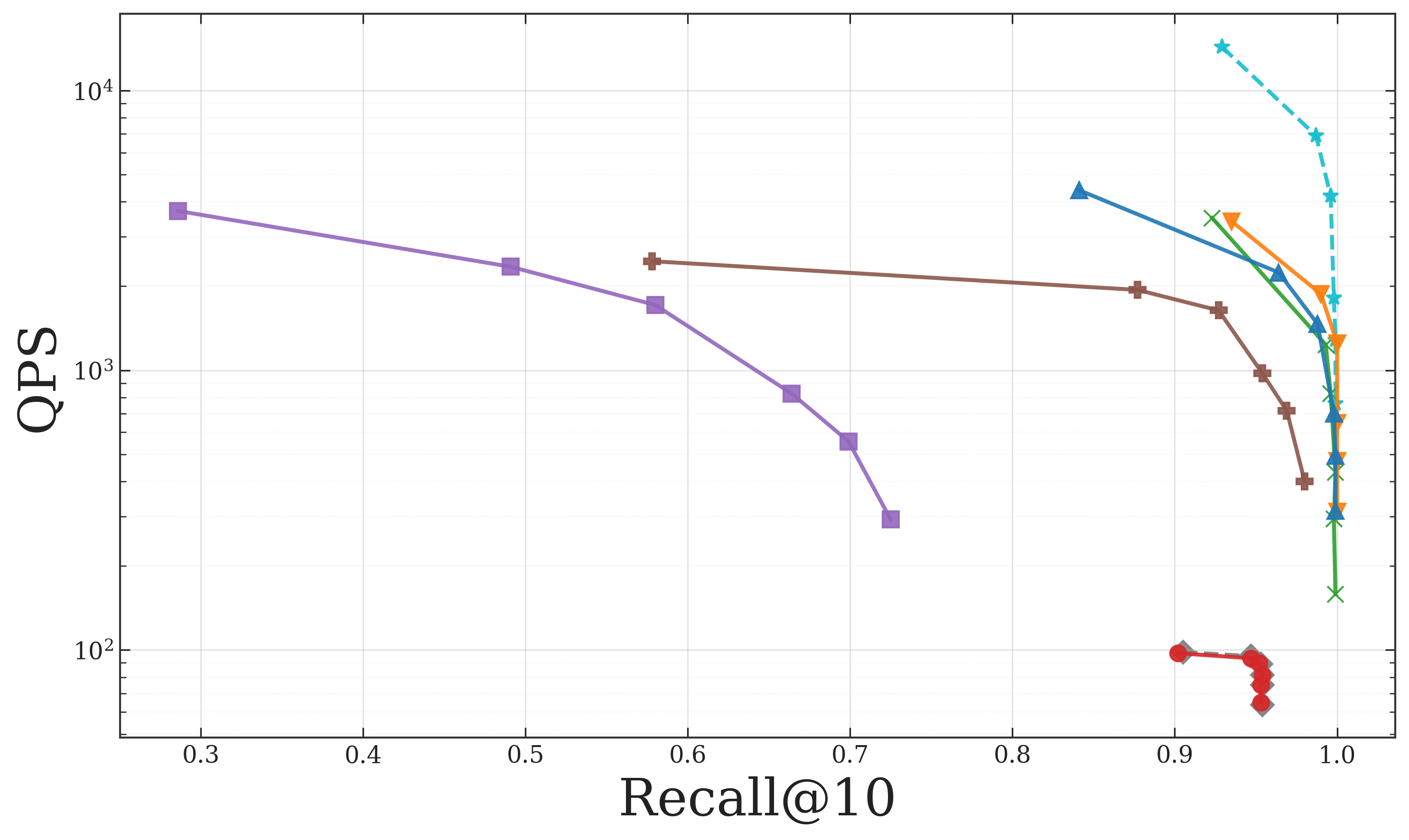}
        \caption{Weighted 1-role.}
        \label{fig:single_role_weighted_query_qps_vs_recall}
    \end{subfigure}\hfill
    \begin{subfigure}[t]{0.195\textwidth}
        \includegraphics[width=\linewidth]{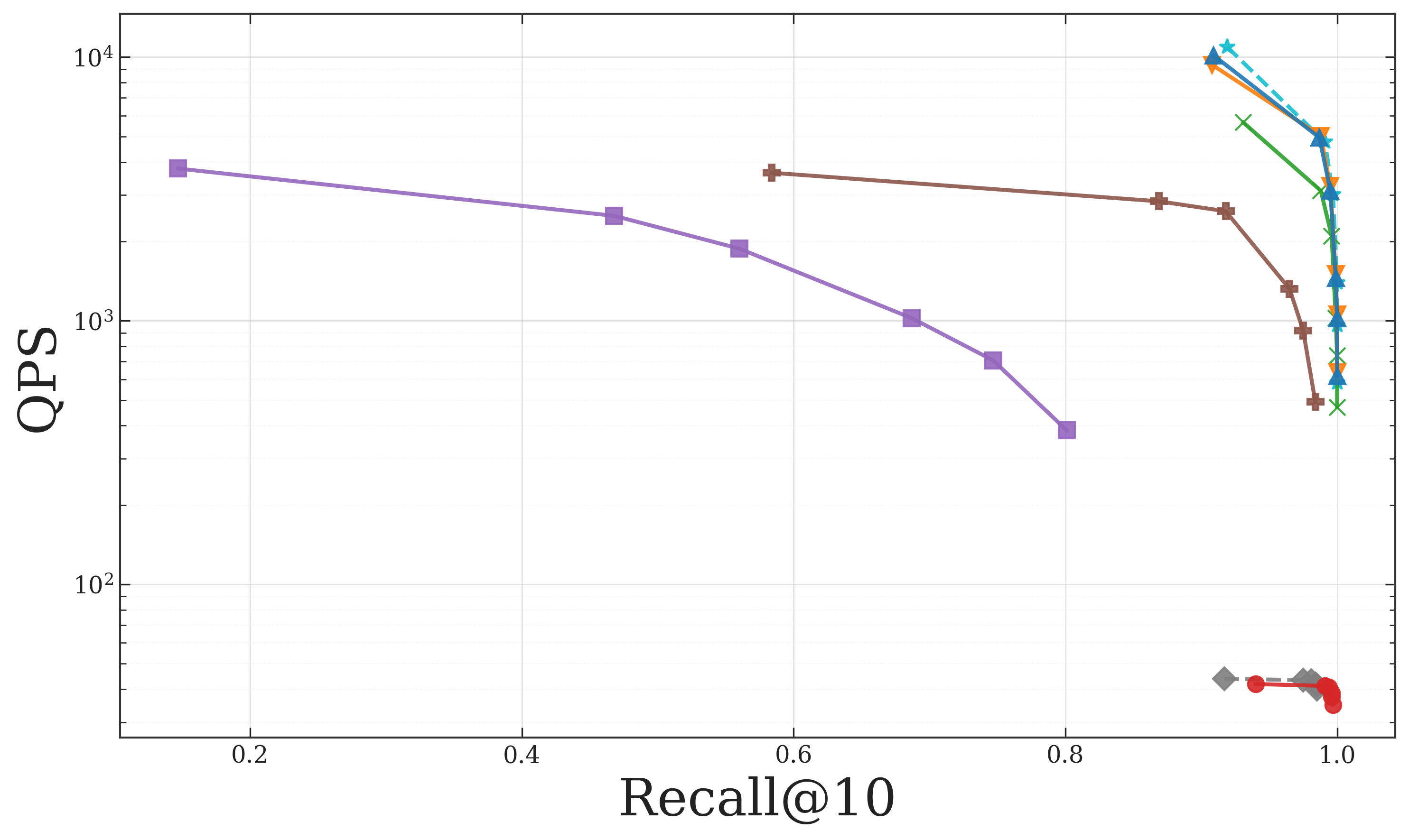}
        \caption{Multi-role wtd.}
        \label{fig:multi_role_weighted_query_qps_vs_recall}
    \end{subfigure}\hfill
    \begin{subfigure}[t]{0.195\textwidth}
        \includegraphics[width=\linewidth]{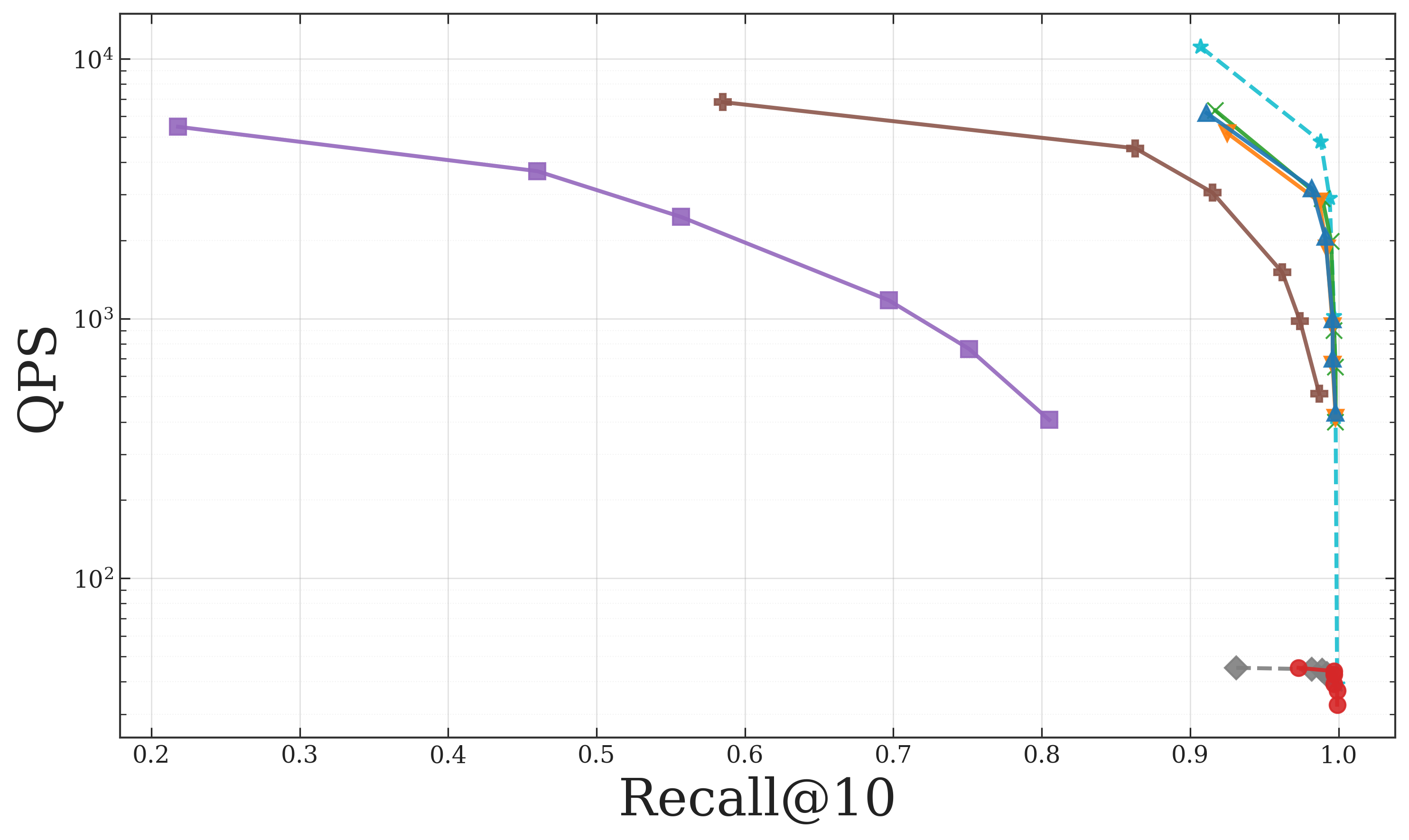}
        \caption{Multi-role unif.}
        \label{fig:multi_role_avg_query_qps_vs_recall}
    \end{subfigure}
    \caption{QPS vs. recall on PAPER and AMZN (a--b), and three alternative workloads on SIFT-1M (c--f).}
    \label{fig:qps_vs_recall_other_datasets}
\end{figure*}



\noindent\textbf{Query workloads.} A query is $q=(\mathbf{x},r)$ or
$q=(\mathbf{x},\tau)$ with $\tau\subseteq\mathcal{R}$. We use four workloads:
(1)~\emph{uniform single-role}, (2)~\emph{weighted single-role} ($\propto|\mathcal{D}(r)|$), (3)~\emph{uniform multi-role} over non-empty
$\tau\in\mathcal{T}$ ($|\tau|{>}1$), and (4)~\emph{weighted multi-role}
($\propto$ exclusive-block sizes). Each workload has 100
queries, averaged over 10 runs. By default we use SIFT-1M with uniform
single-role. 
{The query vector $\mathbf{x}$ may or may not
be drawn from the querying role's authorized data; we capture
this with \textit{query sensitivity}, the fraction of query vectors drawn
from authorized data. For a single-role query with role $r$, the} default value is $1.0$, i.e., all query vectors are drawn from $\mathcal{D}(r)$.

\subsection{Index Creation Evaluation}
\label{sec:eval-index-creation}

\textsf{SA} varies over
$\{1.0,1.1,1.3,1.5,2.0,3.0\}$ in Exp~\ref{exp:index-creation-time}--Exp~\ref{exp:desired-vs-achieved-sa}. Plots and detailed results are deferred to Appendix~\S\ref{app:index-creation-complementary-results}.

\experiment{Index-creation time vs.\ \textsf{SA}}{exp:index-creation-time}
We ran each data partitioning strategy for 3 times to get an average end-to-end building time and data partitioning time. Results are reported in {Figure~\ref{fig:build_time} and Table~\ref{tab:construction-time-ac-sift}
(Appendix~\S\ref{app:index-creation-complementary-results}). \effname cuts data partitioning from $77.412$--$1338.673$s
of \name to $1.983$--$3.794$s by avoiding candidate re-scoring, and is on
average $2117.38$x\,/\,$107.19$x faster than HoneyBee in
data partitioning\,/\,total build and $1.35$x faster than SIEVE in total build.
SIEVE's own data partitioning is negligible since it only selects predicate
subsets from historical workload rather than reorganizing data. \name beats HoneyBee by $5.78$x in
data partitioning and $5.63$x in total build.}


\experiment{Number of indexed vs.\ leftover data}{exp:number-of-indexed-leftover-data}
{Figure~\ref{fig:number-of-indexed-leftover-data}
(Appendix~\S\ref{app:index-creation-complementary-results}) shows that both \name and \effname
index most data and leave only a small leftover fraction, so partitions cover
the data directly rather than relying on brute-force scan. Also, the indexed total
grows with \textsf{SA}.}

\experiment{Number of indices vs.\ \textsf{SA}}{exp:num-indices}
{Figure~\ref{fig:num_indices_built_vs_sa}
(Appendix~\S\ref{app:index-creation-complementary-results}) shows \name and
\effname produce more indices than SIEVE and HoneyBee at every \textsf{SA} budget. This is because the
exclusive lattice exposes fine-grained role-combination blocks and a larger
budget makes more of them indexable, whereas SIEVE partitions data based on historical workload and HoneyBee use coarser role-level partitions.}

\experiment{Desired vs.\ achieved \textsf{SA}}{exp:desired-vs-achieved-sa}
Figure~\ref{fig:input_vs_achieved_sa}  in
Appendix~\S\ref{app:index-creation-complementary-results} shows that \name and \effname track the target \textsf{SA} closely.
HoneyBee and SIEVE account for storage only
\emph{before} admitting a partition, so their final admission can overshoot
the target \textsf{SA}.

\subsection{Query-Based Evaluations}
\label{sec:eval-query}

\experiment{\textsf{QA} vs.\ \textsf{SA}}{exp:qa-vs-sa}
{Figure~\ref{fig:qa_vs_sa_sift1m_single} reports \textsf{QA} to measure the query cost normalized to Oracle indexing (lower is
better). 
\name and \effname consistently outperform HoneyBee and SIEVE, with
the largest gains at small \textsf{SA}. Figure~\ref{fig:qa_vs_sa_sift1m_single} shows that a modest \textsf{SA} budget already removes most
unnecessary search for data partitioning approaches via high-benefit copies and merges, as well as finalization of super-impure nodes. The curves are not strictly monotone
because a larger budget can change which nodes cross $\Lambda$ or which impure
nodes are finalized.}

\experiment{Purity of selected data}{exp:purity}
Figure~\ref{fig:purity_vs_sa_ac} reports the fraction of data touched per
query that is authorized for the queried role for each data partitioning approach. Higher purity indicates
fewer wasted data are selected during query execution. 
\name and \effname achieve high purity across all \textsf{SA} budgets due to our efficient lattice construction strategies. SIEVE's purity depends on \textsf{efs} (we plot
\textsf{efs}$=$100 and 500) because a larger beam (higher \textsf{efs}) pushes more queries to
brute-force scan. HoneyBee achieves low purity across all SA budgets.
Purity directly explains the \textsf{QA} ordering in
Exp~\ref{exp:qa-vs-sa}.

\experiment{Number of indices per query}{exp:n_indcies_per_query}
We report the average number of indices touched per query for \name and \effname in Table~\ref{tab:avg-hnsw-indices-per-query} in 
Appendix~\S\ref{app:ablation-tables}.
Both methods touch fewer than 7 indices per query on average, thus index switching overhead for our approaches is negligible.

\experiment{Indexing threshold}{exp:indexing-threshold}
Tables~\ref{tab:vary-lambda-sift1m-qps-sa11}--\ref{tab:vary-lambda-sift1m-qps-sa15}
and Figure~\ref{fig:big-idx-num-vs-indexing-threshold-sift} in
Appendix~\ref{app:vary-lambda-more-sa} report QPS and the number of indices as
$\Lambda$ varies at $\textsf{SA}\in\{1.1,1.3,1.5\}$. The number of indices reduces as $\Lambda$ increases because fewer nodes exceed $\Lambda$ and become
materialized HNSW indices as $\Lambda$ increases.
Also, strong results occur around $\Lambda{=}2{,}900$, but performance is not
sensitive to the exact value of $\Lambda$. 
Across \name and \effname, the worst QPS remains at least $77.96\%$ of the
best, and in most \textsf{SA} values it is close to $90\%$. This is because the finalization phase for \name and \effname handles the impure nodes efficiently, which absorbs most
threshold-induced differences in the resulting lattices, so neither method is
sensitive to the exact value of $\Lambda$. 

\experiment{Effect of coordinated search}{exp:coordinated-search}
Coordinated search skips phase~2 on at least $67.95\%$ of impure-index visits
at every \textsf{SA} and cuts the inflated \textsf{efs} by
$7.95\%$--$23.67\%$ (Tables~\ref{tab:phase2-skip-rate}
and~\ref{tab:ef-budget-savings}, Appendix~\ref{app:ablation-tables}).
Because finalization removes highly impure nodes, the surviving impure indices
have low impurity and the global-top-$k$ bound suppresses phase~2 with high
probability. Appendix~\ref{app:leftover-effectiveness} further evaluates the
effectiveness of leftover search.

\experiment{Effect of \textsf{efs}}{exp:efs}
{Figure~\ref{fig:qps_vs_ef_ac} reports QPS with \textsf{efs} varied over
$\{10,50,100,300,500,1000\}$. QPS falls for all methods as \textsf{efs} grows because a larger beam explores more candidates in HNSW search.
Global Indexing is the slowest while Oracle is the upper bound since each query hits a single
pure index. \name and \effname are the strongest practical methods throughout. 
ACORN-1 is competitive only at small \textsf{efs}, ACORN-$\gamma$ catches up with \name and \effname at
large \textsf{efs}. HoneyBee stays low due to coarse impure partitions. Figure~\ref{fig:qps_vs_ef_paper} and Figure~\ref{fig:qps_vs_ef_amzn} in
Appendix~\S\ref{app:efs-more-datasets} shows results on PAPER and AMZN.}

\experiment{Comparison with other methods across datasets}{exp:qps-recall-sift}
Figures~\ref{fig:qps_vs_recall_sift}, \ref{fig:qps_vs_recall_paper},
and~\ref{fig:qps_vs_recall_amzn} report QPS vs.\ Recall@10 for $\textsf{efs}\!\in\!\{10,50,$ $100,300,500,1000\}$, and Tables~\ref{tab:qps-vs-efs-sift},~\ref{tab:qps-vs-efs-paper},~\ref{tab:qps-vs-efs-amzn} in Appendix~\ref{app:qps-vs-efs-numbers} show the concrete QPS values. 
\name and \effname occupy the upper-right frontier on all three datasets,
achieving high QPS while maintaining near-oracle recall. 
On average (across all \textsf{efs}), 
\name is $2.24\mathtt{x}$, $5.60\mathtt{x}$, and
$2.61\mathtt{x}$ faster than SIEVE on SIFT-1M, PAPER, and AMZN, respectively;
the corresponding speedups over HoneyBee are $13.17\mathtt{x}$,
$42.61\mathtt{x}$, and $8.69\mathtt{x}$. \effname shows the same trend,
outperforming SIEVE by $2.22\mathtt{x}$, $4.19\mathtt{x}$, and
$2.56\mathtt{x}$, and HoneyBee by $12.55\mathtt{x}$, $32.59\mathtt{x}$, and
$8.46\mathtt{x}$ on the same datasets. 



\experiment{Query sensitivity}{exp:query-sensitivity}
Figure~\ref{fig:recall_vs_sensitivity_ac} varies sensitivity over
0, 0.25, 0.5, 0.75, and 1.0. \name, \effname, SIEVE, and Oracle hold
Recall@10${\geq}0.99$ throughout. ACORN-1 stays below $0.40$ and
ACORN-$\gamma$ at $0.94$--$0.97$, as filtering a global graph fails to recover
all authorized data; Baseline~1 and HoneyBee degrade to $0.90$--$0.93$ once
the query lies outside $\mathcal{D}(r)$. Figure~\ref{fig:recall_vs_sensitivity_additional} in Appendix~\S\ref{app:query-sensitivity-complementary}
shows results on PAPER and AMZN.

\experiment{Weighted single-role queries}{exp:weighted-single-role}
Figure~\ref{fig:single_role_weighted_query_qps_vs_recall} reports QPS vs. Recall@10 for weighted single-role queries; QPS values are in Table~\ref{tab:qps-vs-efs-single-role-weighted} in Appendix~\S\ref{app:qps-vs-efs-numbers}. 
On average (across all \textsf{efs}), 
\name is $1.69\mathtt{x}$ faster than SIEVE and $17.57\mathtt{x}$ faster than HoneyBee, while \effname achieves speedups of $1.51\mathtt{x}$ over SIEVE and $14.67\mathtt{x}$ over HoneyBee.

\experiment{Multi-role queries}{exp:multi-role}
A multi-role query authorizes the union of several roles and can cover most of
$\mathcal{D}$, so we keep one additional global index and route any query
whose authorized region exceeds $80\%$ of $|\mathcal{D}|$ to filtered global
search. Accordingly, we add $1$ to \textsf{SA} for HoneyBee and SIEVE, which are given
$\textsf{SA}{=}2.1$.
Figures~\ref{fig:multi_role_weighted_query_qps_vs_recall}
and~\ref{fig:multi_role_avg_query_qps_vs_recall} show that
\name and \effname stay close to SIEVE and Oracle in the high-recall region on
both uniform and weighted mixed-role workloads. This confirms the crossover argued
in \S\ref{sec:intro}: partitioning wins for selective access, while filtered global
search wins for broad access. 
Oracle itself is sometimes slower than other methods
(Figure~\ref{fig:multi_role_avg_query_qps_vs_recall}) because different roles may route to different pure indices. Maintaining these
role-set-specific indices increases memory overhead, and the resulting
index-switching overhead persists even when the indices are already loaded into memory.
Tables~\ref{tab:qps-vs-efs-multi-role-avg}
and~\ref{tab:qps-vs-efs-multi-role-weighted} in
Appendix~\ref{app:qps-vs-efs-numbers} provide the detailed QPS values across
\textsf{efs} for the two multi-role workloads.

  \section{Related Work}\label{sec:related_work}
General ANN indices supply the building blocks
but assume a single global search space; we therefore focus the remainder of
this section on a qualitative comparison with the four systems closest to our
setting, Filtered-DiskANN~\cite{filtered_diskann}, ACORN~\cite{acorn},
SIEVE~\cite{sieve}, and HoneyBee~\cite{honeybee}, and defer the full landscape of related works to Appendix~\ref{app:detailed_related_work}.

\noindent\textbf{Predicate Model.}
Filtered-DiskANN and UNG~\cite{ung} bake \emph{label} predicates into graph
construction, and ACORN traverses the predicate-induced subgraph of an HNSW
index for \emph{arbitrary} Boolean filters. All three treat the predicate as
query-supplied metadata. SIEVE and HoneyBee instead exploit that the predicate
\emph{distribution} is known offline, SIEVE from a historical workload, HoneyBee from the RBAC role lattice, 
and materialize sub-indexes accordingly.
\name and \effname follow the latter view but represent the policy as exclusive
role-subset blocks, so every role's visible set is an exact union of blocks
rather than an approximate match against a learned filter.

\noindent\textbf{Storage--Latency Trade-off.}
{Filtered-DiskANN and ACORN keep $\textsf{SA}{\approx}1$ and pay at
query time when the predicate is selective; SIEVE and HoneyBee expose an
explicit memory budget and select sub-indexes or partitions under it. \name
exposes the same knob but explores it with two primitives: \emph{merge} trades
latency for storage by combining co-accessed blocks, and \emph{copy} trades
storage for latency by replicating a block into a co-accessed group.}

\noindent\textbf{Query Execution.}
ACORN and Filtered-DiskANN answer every query from \emph{one} graph. SIEVE
routes each query to the single cheapest subsuming sub-index (or brute force),
and HoneyBee routes a role to the partition set implied by its split. None of
them lets results from one index influence search on another. \name and \effname instead
run a coordinated multi-index search: pure partitions are probed first and
their top-$k$ distances tighten a global threshold that prunes the subsequent
beam search on impure partitions and residual scans. This is the mechanism that
lets \name tolerate impure groups without inflating $\textsf{efs}$ to the
worst-case $\lambda\,\textsf{efs}$ on every index.

\noindent\textbf{Recall Under Impurity.}
{When authorized vectors are sparse in the searched structure, ACORN
widens neighbor lists, Filtered-DiskANN stitches label-specific edges, and
HoneyBee/SIEVE over-search or fall back to linear scan. All four raise the
effective $\textsf{efs}$, which our cost model (Definition~\ref{def:hnsw-cost})
shows is linear in latency. \name instead lowers impurity at construction time
(copy/merge) and caps it at execution time via the shared distance bound,
reducing the \emph{need} for inflation rather than enlarging the candidate
set.}

  \section{Conclusion}\label{sec:conclusion}

We presented \name and \effname, two access-aware indexing strategies for
vector databases. \name and \effname organize data in an access-aware lattice and use copy and
merge operations to group co-accessed blocks under a storage budget. Large
nodes are indexed with HNSW and small nodes are stored as leftovers for efficient linear scan. To process queries across indices efficiently, we introduce coordinated
search that probes pure nodes (and leftover vectors) that contain only authorized data first so that the resulting distance bound can be utilized for pruning impure ones. Open directions include incremental maintenance under streaming
inserts and permission revocations and extending the framework to attribute-based and hierarchical policies.










  \bibliographystyle{ACM-Reference-Format}
  \bibliography{reference}

@article{bertino2011access,
  title={Access control for databases: Concepts and systems},
  author={Bertino, Elisa and Ghinita, Gabriel and Kamra, Ashish and others},
  journal={Foundations and Trends{\textregistered} in Databases},
  volume={3},
  number={1--2},
  pages={1--148},
  year={2011},
  publisher={Now Publishers, Inc.}
}

@misc{AI_Act,
  author    = {{European Parliament and Council of the European Union}},
  title     = {Regulation ({EU}) 2024/1689 of the European Parliament and of the Council of 13 June 2024 laying down harmonised rules on artificial intelligence and amending certain Union legislative acts (Artificial Intelligence Act)},
  year      = {2024},
  howpublished = {\url{https://eur-lex.europa.eu/eli/reg/2024/1689/oj/eng}},
  note      = {Accessed: 2025-01-10}
}

@inproceedings{lewis2020retrieval,
  title={Retrieval-Augmented Generation for Knowledge-Intensive {NLP} Tasks},
  author={Lewis, Patrick and Perez, Ethan and Piktus, Aleksandra and Petroni, Fabio and Karpukhin, Vladimir and Goyal, Naman and K\"{u}ttler, Heinrich and Lewis, Mike and Yih, Wen-tau and Rockt\"{a}schel, Tim and Riedel, Sebastian and Kiela, Douwe},
  booktitle={Advances in Neural Information Processing Systems},
  volume={33},
  pages={9459--9474},
  year={2020}
}

@article{hnsw,
  title={Efficient and robust approximate nearest neighbor search using hierarchical navigable small world graphs},
  author={Malkov, Yu A and Yashunin, Dmitry A},
  journal={IEEE transactions on pattern analysis and machine intelligence},
  volume={42},
  number={4},
  pages={824--836},
  year={2018},
  publisher={IEEE}
}

@article{sieve,
  title={SIEVE: Effective Filtered Vector Search with Collection of Indexes},
  author={Li, Zhaoheng and Huang, Silu and Ding, Wei and Park, Yongjoo and Chen, Jianjun},
  journal={arXiv preprint arXiv:2507.11907},
  year={2025}
}

@article{reichman2024retrieval,
  title={Retrieval-Augmented Generation: Is Dense Passage Retrieval Retrieving},
  author={Reichman, Benjamin and Heck, Larry},
  journal={arXiv preprint arXiv:2402.11035},
  year={2024}
}

@article{liu2023lost,
  title={Lost in the middle: How language models use long contexts},
  author={Liu, Nelson F and Lin, Kevin and Hewitt, John and Paranjape, Ashwin and Bevilacqua, Michele and Petroni, Fabio and Liang, Percy},
  journal={arXiv preprint arXiv:2307.03172},
  year={2023}
}

@article{xu2023retrieval,
  title={Retrieval meets long context large language models},
  author={Xu, Peng and Ping, Wei and Wu, Xianchao and McAfee, Lawrence and Zhu, Chen and Liu, Zihan and Subramanian, Sandeep and Bakhturina, Evelina and Shoeybi, Mohammad and Catanzaro, Bryan},
  journal={arXiv preprint arXiv:2310.03025},
  year={2023}
}

@article{sanyal2025orgaccess,
  title={OrgAccess: A Benchmark for Role Based Access Control in Organization Scale LLMs},
  author={Sanyal, Debdeep and Maharana, Umakanta and Sinha, Yash and Tan, Hong Ming and Karande, Shirish and Kankanhalli, Mohan and Mandal, Murari},
  journal={arXiv preprint arXiv:2505.19165},
  year={2025}
}

@article{acorn,
  title={Acorn: Performant and predicate-agnostic search over vector embeddings and structured data},
  author={Patel, Liana and Kraft, Peter and Guestrin, Carlos and Zaharia, Matei},
  journal={Proceedings of the ACM on Management of Data},
  volume={2},
  number={3},
  pages={1--27},
  year={2024},
  publisher={ACM New York, NY, USA}
}

@article{honeybee,
  title={HoneyBee: Efficient role-based access control for vector databases via dynamic partitioning},
  author={Zhong, Hongbin and Lentz, Matthew and Narodytska, Nina and Szekeres, Adriana and Rong, Kexin},
  journal={arXiv preprint arXiv:2505.01538},
  year={2025}
}

@article{faiss,
  title={Billion-scale similarity search with GPUs},
  author={Johnson, Jeff and Douze, Matthijs and J{\'e}gou, Herv{\'e}},
  journal={IEEE transactions on big data},
  volume={7},
  number={3},
  pages={535--547},
  year={2019},
  publisher={IEEE}
}

@inproceedings{diskann,
  title={DiskANN: Fast accurate billion-point nearest neighbor search on a single node},
  author={Subramanya, Suhas Jayaram and Devvrit and Kadekodi, Rohan and Krishnaswamy, Ravishankar and Simhadri, Harsha Vardhan},
  booktitle={Advances in Neural Information Processing Systems},
  volume={32},
  year={2019}
}

@inproceedings{scann,
  title={Accelerating large-scale inference with anisotropic vector quantization},
  author={Guo, Ruiqi and Sun, Philip and Lindgren, Erik and Geng, Quan and Simcha, David and Chern, Felix and Kumar, Sanjiv},
  booktitle={International Conference on Machine Learning},
  pages={3887--3896},
  year={2020},
  organization={PMLR}
}

@inproceedings{filtered_diskann,
  title={Filtered-diskann: Graph algorithms for approximate nearest neighbor search with filters},
  author={Gollapudi, Siddharth and Karia, Neel and Sivashankar, Varun and Krishnaswamy, Ravishankar and Begwani, Nikit and Raz, Swapnil and Lin, Yiyong and Zhang, Yin and Mahapatro, Neelam and Srinivasan, Premkumar and others},
  booktitle={Proceedings of the ACM Web Conference 2023},
  pages={3406--3416},
  year={2023}
}

@article{ung,
  title={Navigating labels and vectors: A unified approach to filtered approximate nearest neighbor search},
  author={Cai, Yuzheng and Shi, Jiayang and Chen, Yizhuo and Zheng, Weiguo},
  journal={Proceedings of the ACM on Management of Data},
  volume={2},
  number={6},
  pages={1--27},
  year={2024},
  publisher={ACM New York, NY, USA}
}

@article{sift,
  title={Product quantization for nearest neighbor search},
  author={Jegou, Herve and Douze, Matthijs and Schmid, Cordelia},
  journal={IEEE transactions on pattern analysis and machine intelligence},
  volume={33},
  number={1},
  pages={117--128},
  year={2010},
  publisher={IEEE}
}

@article{paper,
  title={Navigable proximity graph-driven native hybrid queries with structured and unstructured constraints},
  author={Wang, Mengzhao and Lv, Lingwei and Xu, Xiaoliang and Wang, Yuxiang and Yue, Qiang and Ni, Jiongkang},
  journal={arXiv preprint arXiv:2203.13601},
  year={2022}
}

@inproceedings{amazon_books_reviews,
  title={Ups and downs: Modeling the visual evolution of fashion trends with one-class collaborative filtering},
  author={He, Ruining and McAuley, Julian},
  booktitle={proceedings of the 25th international conference on world wide web},
  pages={507--517},
  year={2016}
}

@inproceedings{amazon_books_reviews2,
  title={Image-based recommendations on styles and substitutes},
  author={McAuley, Julian and Targett, Christopher and Shi, Qinfeng and Van Den Hengel, Anton},
  booktitle={Proceedings of the 38th international ACM SIGIR conference on research and development in information retrieval},
  pages={43--52},
  year={2015}
}

@book{zipf2016human,
  title={Human behavior and the principle of least effort: An introduction to human ecology},
  author={Zipf, George Kingsley},
  year={2016},
  publisher={Ravenio books}
}

@article{mandelbrot1953informational,
  title={An informational theory of the statistical structure of language},
  author={Mandelbrot, Benoit},
  journal={Communication theory},
  volume={84},
  number={21},
  pages={486--502},
  year={1953},
  publisher={London}
}

@article{aumuller2020ann,
  title={ANN-Benchmarks: A benchmarking tool for approximate nearest neighbor algorithms},
  author={Aum{\"u}ller, Martin and Bernhardsson, Erik and Faithfull, Alexander},
  journal={Information Systems},
  volume={87},
  pages={101374},
  year={2020},
  publisher={Elsevier}
}

@inproceedings{fagin2001optimal,
  title={Optimal aggregation algorithms for middleware},
  author={Fagin, Ronald and Lotem, Amnon and Naor, Moni},
  booktitle={Proceedings of the twentieth ACM SIGMOD-SIGACT-SIGART symposium on Principles of database systems},
  pages={102--113},
  year={2001}
}

@article{khuller1999budgeted,
  title={The budgeted maximum coverage problem},
  author={Khuller, Samir and Moss, Anna and Naor, Joseph Seffi},
  journal={Information processing letters},
  volume={70},
  number={1},
  pages={39--45},
  year={1999},
  publisher={Elsevier}
}

@article{martello1987algorithms,
  title={Algorithms for knapsack problems},
  author={Martello, Silvano and Toth, Paolo},
  journal={North-Holland Mathematics Studies},
  volume={132},
  pages={213--257},
  year={1987},
  publisher={Elsevier}
}
  
  \newpage
\appendix

\section{Detailed Introduction of HNSW}
\label{sec:prelim-hnsw}

HNSW is one of the most effective and widely used approximate
nearest-neighbor data structures. It organizes vectors into a multi-layer
hierarchy in which each layer is a navigable proximity graph~\cite{hnsw}. At
insertion time, each vector is assigned a maximum level $\ell\ge 0$ by
sampling from an exponential (geometric) distribution. On every upper layer
$\ell>0$ a node retains up to $M$ bidirectional links chosen by a
diversity-preserving heuristic; on the base layer ($\ell=0$) the limit is
$M_0\le 2M$ (often configurable in implementations).

To answer a query, HNSW starts at the topmost layer and greedily walks toward
the query vector, descending one layer at a time until it reaches the base
layer in $\mathcal{O}(\log|{\tt idx}|)$ expected steps. At the base layer the
search switches to a best-first expansion: candidates are kept in a priority
queue of capacity $\textsf{efs}$ and expanded in increasing distance, each
expansion visiting up to $M_0$ neighbors. The parameter $\textsf{efs}$ governs
the recall--latency trade-off. Counting comparisons, the asymptotic query cost
has three terms: upper-layer descent, base-layer priority-queue maintenance,
and final selection of $k$ results from the $\textsf{efs}$ candidates:
$
\mathcal{O}(M\log |{\tt idx}|)
+ \mathcal{O}(\textsf{efs}\log\textsf{efs})
+ \mathcal{O}(\textsf{efs}\log k),$


In role-based settings the comparison-count bound is not the right proxy for
latency. When an index is impure for the issuing role, the query must inflate
$k$ to $k'=\lceil\lambda k\rceil$ and $\textsf{efs}$ to
$\lceil\lambda\,\textsf{efs}\rceil$ (\S\ref{sec:ac-indexing}), pushing
$\textsf{efs}$ into a regime where the base-layer term dominates.
Figure~\ref{fig:vary_k} sweeps $k$ (and the corresponding
$\textsf{efs}=\alpha k$) and shows that wall-clock latency grows
\emph{linearly} once $k$ exceeds a small threshold $k_o$, not as
$\textsf{efs}\log\textsf{efs}$. The cost model used throughout the paper
therefore charges $C_{\theta}({\tt idx},\textsf{efs})=a\log |{\tt idx}| + b\,\textsf{efs} + c$
(Definition~\ref{def:hnsw-cost}); Appendix~\ref{sec:cost-function-selection}
derives this form from the per-expansion hardware cost and gives the
calibration procedure for $\theta=(a,b,c)$.

\begin{figure}[h!]
  \centering
  \includegraphics[width=0.5\linewidth]{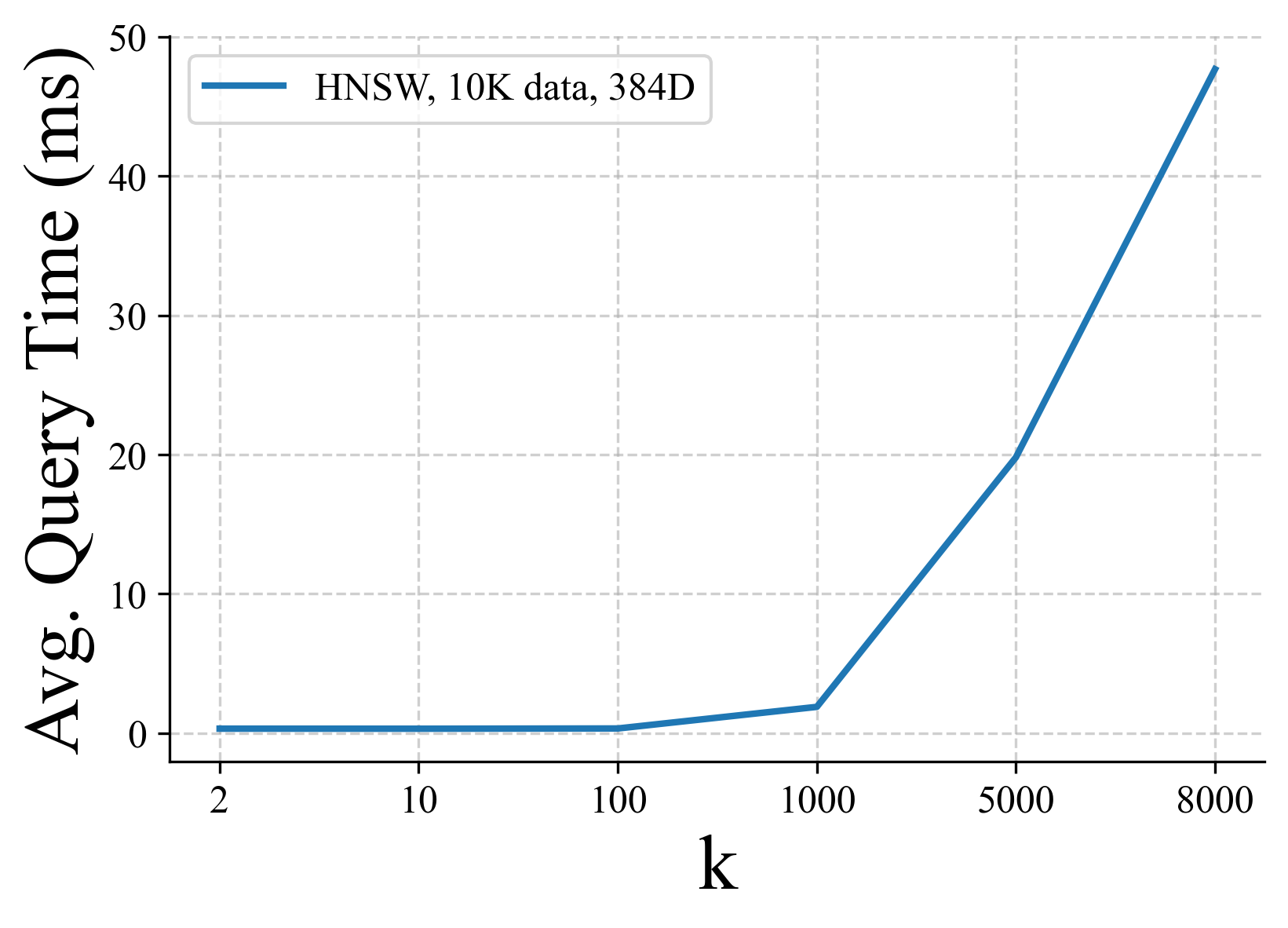}
  \caption{Query time vs.\ top-$k$ (log-$k$ scale).}
  \label{fig:vary_k}
\end{figure}

\begin{figure}[h!]
  \centering
  \includegraphics[width=0.58\linewidth]{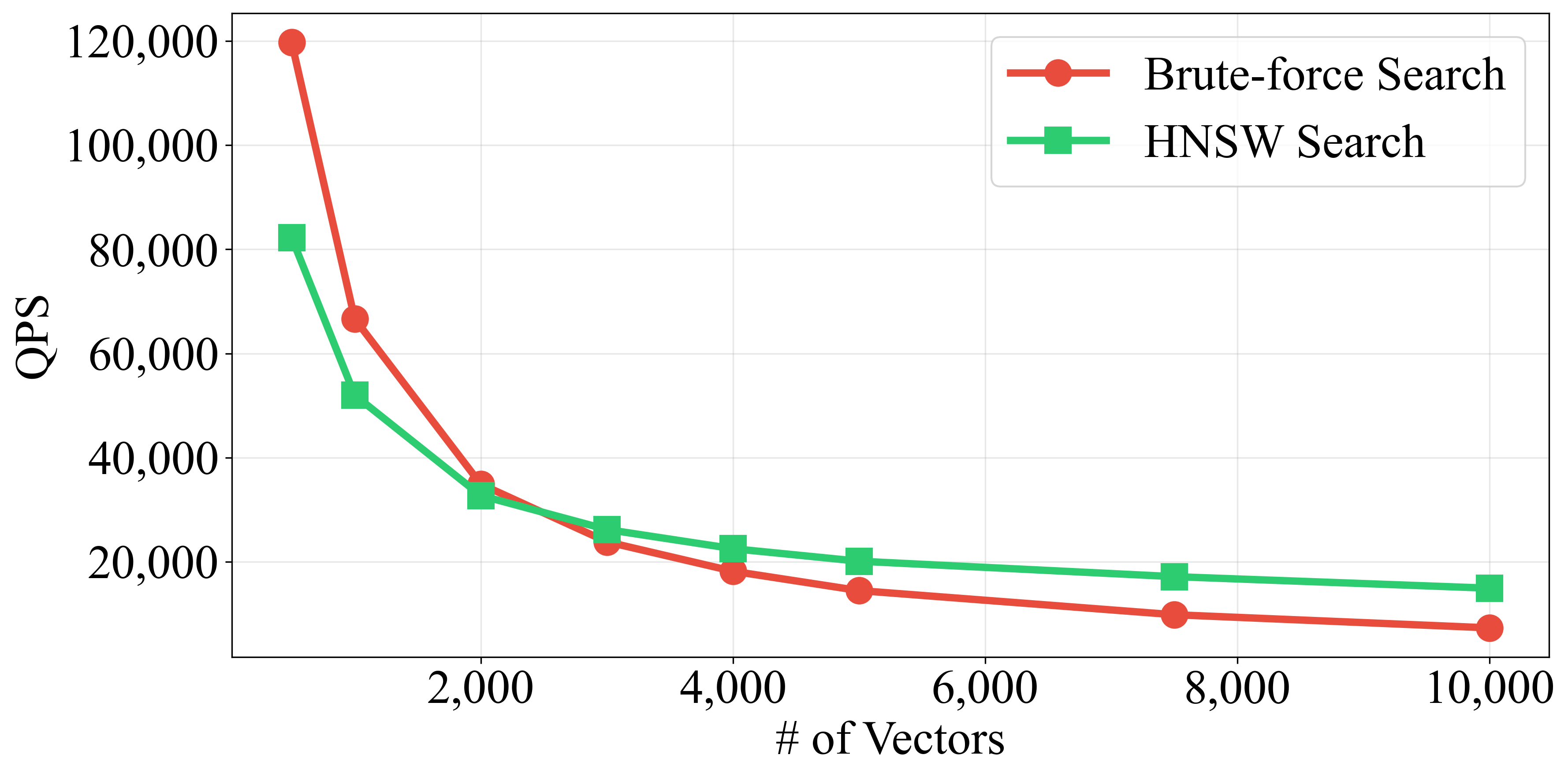}
  \caption{HNSW vs.\ linear scan ($d{=}384$).}
  \label{fig:vary-dimensions-app}
\end{figure}

Figure~\ref{fig:vary-dimensions-app} complements
Figure~\ref{fig:vary-dimensions}. The trend is consistent with the main
figure: for small partitions, the overhead of entering and traversing an
HNSW index outweighs its pruning benefit, so a linear scan is faster; once
the partition contains a few thousand vectors, HNSW becomes preferable.

\section{Cost Function Construction}
\label{sec:cost-function-selection}

Our partitioning algorithms compare candidate index layouts by their
predicted query latency, so they require a closed-form model
$C_{\theta}({\tt idx},\textsf{efs})$ that maps an index ${\tt idx}$ searched
with beam width $\textsf{efs}$ to wall-clock time. This
appendix justifies the functional form used in
Definition~\ref{def:hnsw-cost}, describes how its coefficients are calibrated,
and validates the choice empirically.

\subsection{
Wall-Clock Latency of HNSW Search}
\label{sec:cost-function-form}

The textbook complexity of an HNSW query decomposes into an upper-layer
greedy descent of $\mathcal{O}(\log |{\tt idx}|)$ hops and a base-layer beam search that
performs $\mathcal{O}(\textsf{efs})$ expansions while maintaining a priority
queue of capacity $\textsf{efs}$. Charging each expansion an
$\mathcal{O}(\log\textsf{efs})$ heap update yields the familiar proxy. Let
$a$ weight upper-layer descent, $b$ weight base-layer work, and $c$ capture
fixed overhead:
$C_\theta^\mathsf{cmp}({\tt idx},\textsf{efs}) \;=\; a\log |{\tt idx}| \;+\; b\,\textsf{efs}\log\textsf{efs} \;+\; c.$
This expression counts \emph{comparisons}. It is the wrong currency for
latency prediction because the heap update is not where time is spent.

Each base-layer expansion also (i)~reads the $M$ neighbor IDs of the popped
node, (ii)~fetches the corresponding $M$ vectors from memory, and
(iii)~computes $M$ distances of $d$ dimensions each. Steps~(ii)--(iii) cost
$\mathcal{O}(Md)$ arithmetic operations plus, once the index outgrows the CPU
cache, $\mathcal{O}(M)$ cache-missing random reads. For a representative
configuration ($d{=}128$, $M{=}16$, $\textsf{efs}{=}100$) a single expansion
performs $M\!\cdot\!d = 2{,}048$ multiply--adds and up to $16$ vector fetches,
against only $\log_2\textsf{efs}\approx 7$ heap comparisons. The heap term is
therefore two to three orders of magnitude cheaper than the distance term and
vanishes into the per-step constant. Since the per-step work does not depend
on $\textsf{efs}$, the base-layer cost is, to first order, \emph{linear} in the
number of expansions:
$C_\theta({\tt idx},\textsf{efs}) \;=\; a\log |{\tt idx}| \;+\; b\,\textsf{efs} \;+\; c.$
Here the constants $a$, $b$, $c$ depend on $d$, $M$, the SIMD width, and the
memory hierarchy of the host, none of which are known a~priori. We therefore
fit them on the deployment machine rather than derive them analytically.

\subsection{Calibrating the Cost Function}
\label{sec:cost-function-calibration}

We isolate the two variable terms with two one-dimensional sweeps.

\smallskip
\noindent\textbf{Upper-layer sweep.} Fix $\textsf{efs}=1$, $k=1$ and vary $|{\tt idx}|$.
With a single-slot beam the base layer does negligible work, so the measured
latency reflects the greedy descent plus fixed overhead. A least-squares fit
of $T_{\textsf{size}}(|{\tt idx}|)=a\log |{\tt idx}| + c_1$ recovers $a$. 

\smallskip
\noindent\textbf{Base-layer sweep.} Fix $|{\tt idx}|=|{\tt idx}_0|$ and vary $\textsf{efs}$. We fit
both candidates,
$T_{\textsf{efs}}^{\text{lin}}(\textsf{efs})=b\,\textsf{efs}+c_2$ and
$T_{\textsf{efs}}^{\text{log}}(\textsf{efs})=b'\,\textsf{efs}\log\textsf{efs}+c_2'$,
and select whichever attains the higher $R^2$.

\smallskip
\noindent\textbf{Combining the intercepts.} Each sweep holds the other
variable fixed, so each intercept already contains the other term's
contribution at the held value. We remove that contribution from both sides,
$c^{(1)} = c_1 - b\cdot 1$ and $c^{(2)} = c_2 - a\log |{\tt idx}_0|$, and average the
two estimates, $c = \tfrac{1}{2}\bigl(c^{(1)}+c^{(2)}\bigr)$, to damp
measurement noise from either sweep.

\label{sec:cost-function-validation}

Figure~\ref{fig:time-vs-ef} shows the base-layer sweep on the default
configuration ($d{=}128$, $M{=}16$). The linear model tracks the measurements
across the full range of $\textsf{efs}$, whereas the
$\textsf{efs}\log\textsf{efs}$ model bends upward and over-predicts at large
$\textsf{efs}$---precisely the regime that impure indices enter when
$\lambda\,\textsf{efs}$ is inflated (Definition~\ref{def:hnsw-cost}). The
linear fit attains $R^2=0.9938$ versus $R^2=0.9811$ for the log-linear fit,
confirming the hardware argument of \S\ref{sec:cost-function-form}. We
therefore adopt
\[
    {\,C_{\theta}(|{\tt idx}|,\textsf{efs}) \;=\; a\log |{\tt idx}| \;+\; b\,\textsf{efs} \;+\; c\,}
\]
with $\theta=(a,b,c)$ fitted once per dataset configuration. For the
configuration of $d{=}128$ and $M{=}16$, the calibrated model is
\[
    C_{\theta}(|{\tt idx}|,\textsf{efs}) \;=\; 0.0821\log_2 |{\tt idx}| \;+\; 0.1159\,\textsf{efs} \;+\; 2.3110,
\]
in which the base-layer term dominates for any practical $\textsf{efs}$ while
the $\log |{\tt idx}|$ term contributes a small but non-negligible offset. The logarithm
base is absorbed into $a$; our implementation uses $\log_2$.



\begin{figure}[t]
    \centering
    \includegraphics[width=0.26\textwidth]{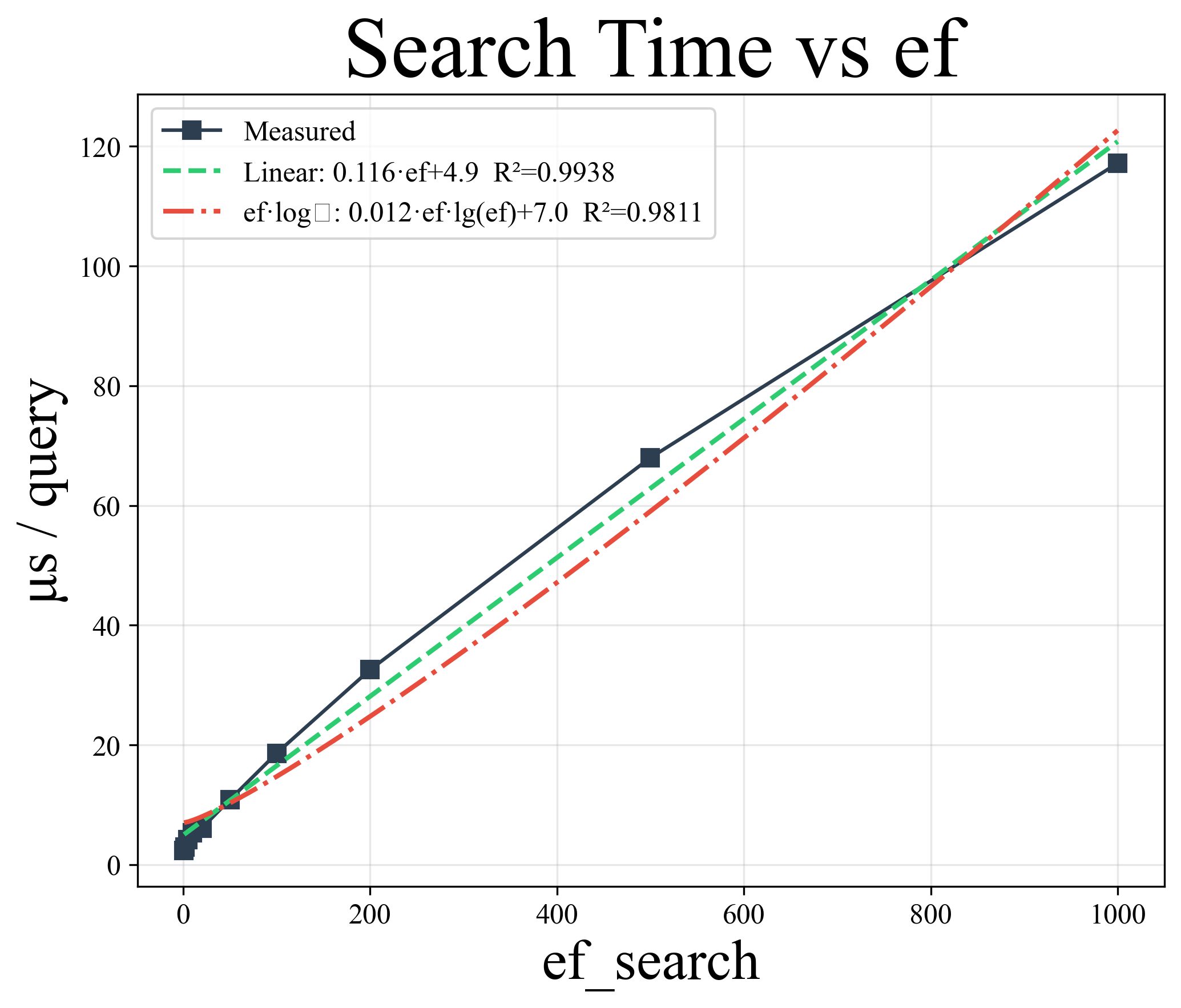}
    \caption{Base-layer sweep: search time as $\textsf{efs}$ varies at fixed
    $|{\tt idx}|$. The linear model ($R^2{=}0.994$) fits better than
    $\textsf{efs}\log\textsf{efs}$ ($R^2{=}0.981$), which over-predicts at large
    $\textsf{efs}$.}
    \label{fig:time-vs-ef}
\end{figure}

\section{Notation}
\label{sec:notations}
Table~\ref{tab:notation} summarizes the notation used throughout the paper.
The symbols are grouped by the section that introduces them: dataset and
access-control primitives from \S\ref{sec:preliminary}, lattice and indexing
quantities from \S\ref{sec:ac-indexing}, and the cost-model parameters from
Definition~\ref{def:hnsw-cost}.

\begin{table}[h]
        \centering
        \caption{Frequently used notation.}
        \label{tab:notation}
        \begin{tabular}{@{}llp{5cm}@{}}
        \toprule
        \textbf{Symbol} & \textbf{Domain} & \textbf{Description} \\
        \midrule
        $\mathcal{D}$ & Vectors & Entire vector dataset \\
        $\mathcal{R}$ & Roles & Set $\{ r_1, r_2, \dots \}$ of user roles \\
        $\tau$ & $\mathcal{P}(\mathcal{R})$ & A combination of roles \\
        $N^{\mathit{ex}}(\tau)$ & Vectors & Vectors exclusively accessible to $\tau$\\
        $\mathcal{L}_{\mathit{ex}}$ & DAG & Exclusive lattice. Nodes $N^{\mathit{ex}}(\tau)\in \mathcal{L}_{\mathit{ex}}$ store vectors exclusively accessible to $\tau$ \\
        $\mathcal{L}$ & DAG & Optimized lattice; each node $N(\tau)$ stores one or more exclusive blocks \\
        $\mathcal{D}(r)$ & Vectors & Set of all vectors accessible to role $r$\\
        $\mathcal{I}$ & Indices & Set of HNSW indices \\
        $\mathcal{D}({\tt idx})$ & Vectors & Set of vectors contained in index ${\tt idx}$ \\
        $\mathcal{I}(r)$ & Indices & HNSW indices touched for queries with role $r$ \\
        $\textsf{QP}(r)$ &  & Query plan for role $r$; data blocks required for queries with role $r$ \\
        $\textsf{Cost}(q)$ & $\mathbb{R}^+$ & Query cost for $q = (\mathbf{x},r)$ \\
        $\textsf{SA}(\mathcal{I})$ & $\mathbb{R}^+$ & Storage amplification: $\frac{\sum_{{\tt idx}\in \mathcal{I}} |\mathcal{D}({\tt idx})|}{|\mathcal{D}|}$ \\
        $\beta$ & $\mathbb{R}^+$ & Storage amplification budget \\
        $\Lambda$ & $\mathbb{N}$ & HNSW indexing threshold \\
        \bottomrule 
        \end{tabular}
        \end{table}

\section{Exclusive Lattice Construction}
\label{sec:exclusive_lattice_construction}

\begin{algorithm}
    \caption{\textbf{Exclusive Lattice Construction}}
    \label{alg:exclusive_lattice_construction}
    \begin{algorithmic}[1]
        \Require $\mathcal{D}$: the access control dataset. 

        \Function{GetExclusiveLattice}{$\mathcal{D}$}

        \State $\mathit{areas}\gets \textsc{GetDataAreas}(\mathcal{D})$

        \State $\mathcal{N}\gets \mathsf{sort}(\mathit{areas}.\mathit{nodes}, \prec_{\mathit{role\_num}})$, $\mathcal{L}_\mathit{ex}\gets \phi$

        \For{$N_i\in \mathcal{N}$}
            \For{$N_c\in \mathcal{N}$}
            \If{$N_c$ is a child of $N_i$}
            \State $\mathit{areas}[N_i]\gets \mathit{areas}[N_i] \backslash \mathit{areas}[N_c]$
            
            \EndIf
            \EndFor
            \If{$\mathit{areas}[N_i]$.$\mathit{size}>0$}
            \State $\mathcal{L}_\mathit{ex}[N_i]\gets\mathit{areas}[N_i]$
            \EndIf
        \EndFor
        \State \textbf{return} $\mathcal{L}_\mathit{ex}$
        \EndFunction

        \Function{GetDataAreas}{$\mathcal{D}$}
        
        \State $\mathit{roles}, \mathit{data} \gets \textsc{GetRolesAndData}(\mathcal{D})$\label{ln:find_single_category_and_data_ids} 

        \State $\mathit{cache}_{\mathit{single}}\gets \textsc{GetSingleRoleData}(\mathit{role}, \mathit{data})$

        \State $\mathit{areas}, \mathit{cache}\gets \mathit{cache}_{\mathit{single}}, \mathit{cache}_{\mathit{single}}$\label{nl:initiate_node_areas}  

        \State $\ell_\mathit{max}\gets \mathit{roles}.\mathit{length}$ \Comment{Max role combinations}

        \For{$i \in \mathit{range}(2, \ell_\mathit{max})$}\label{nl:process_in_layers_lattice_construction}\Comment{From Layer 2 to $\ell_\mathit{max}$}
            \State $\mathit{policies}\gets \mathit{cache}_\mathit{last}.\mathit{keys}()$

            \For{$\mathit{p}\in \mathit{policies}$}
                \For{$\mathit{r}\in \mathit{cache}_{\mathit{single}}.\mathit{keys}()$}
                    \If{$\mathit{r}\notin \mathit{p}$}
                    \State $\mathit{p}^{\prime}\gets \mathit{p} | \mathit{r}$, $d\gets \mathit{cache}[\mathit{p}]\cap\mathit{\mathit{cache}_{\mathit{single}}[r]}$

                    \If{$d$}
                        $\mathit{areas}[{p^{\prime}}]\gets d$, $\mathit{cache}[{p^{\prime}}]\gets d$
                    \EndIf
                    \EndIf
                \EndFor
                \State $\mathit{cache}[\mathit{p}]$.$\mathit{delete}()$
            \EndFor
        \EndFor
        \State \textbf{return} $\mathit{areas}$
        \EndFunction
    \end{algorithmic}
\end{algorithm}

Exclusive lattice can be constructed by processing the dataset twice; see Algorithm~\ref{alg:exclusive_lattice_construction} for details.
Below we describe a method that designs the metadata for the access control dataset properly to construct the exclusive lattice efficiently.
The access-control metadata can be represented as a sparse binary matrix in Compressed Sparse Row (CSR) format. Each row corresponds to one data vector, and each column corresponds to one filter, policy label, or access-control attribute. A nonzero entry indicates that the corresponding data vector satisfies that filter. In CSR format, the matrix is stored using row offsets and column indices: for row $i$, the active filters are given by
$\texttt{indices}[\texttt{indptr}[i] : \texttt{indptr}[i+1]].$
Before constructing the lattice, the matrix is converted to boolean form, duplicate entries are merged, and column indices within each row are sorted. This ensures that each row has a canonical representation of its filter set.

The exclusive lattice groups data vectors by their exact filter combination. For each nonempty row, the sorted tuple of active filter IDs is treated as a  signature for a role set $\tau$, and the row ID is inserted into the lattice cell corresponding to that signature. Therefore, each data vector belongs to exactly one lattice cell, and the cell representing its role set $\tau$. 

\section{The Optimization Problem}\label{sec:opt_formulation}
We now give a precise optimization view of the index grouping problem. Given the lattice $\mathcal{L}_{\mathit{ex}}$ with a set of exclusive nodes $\mathcal{N}_{\mathit{ex}}$ and letting $\mathcal{T}$ be the set of role subsets (e.g., the set of access control policies, or, the set of combinations of accessible roles)
corresponding to $\mathcal{N}_{\mathit{ex}}$.
we build at most one index per $\tau \in \mathcal{T}$ by merging and copying exclusive nodes upward in the lattice.
The decision variables are:
\begin{itemize}
  \item $z_{\tau',\tau} \in \{0,1\}$: $N^{\mathit{ex}}(\tau')$ is included in the index for $\tau$;
  \item $x_\tau \in \{0,1\}$: the index for $\tau$ is materialized;
  \item $y_{q,\tau} \in \{0,1\}$: query $q$ consults the index for $\tau$.
\end{itemize}

The candidate size of an index for a set $\tau$ of roles is
$s_\tau = \sum_{\tau'} z_{\tau',\tau} \cdot |N^{\mathit{ex}}(\tau')|.$
For a query $q$ with role set $q.\rho$, the authorized blocks are
$\mathcal{N}_q = \{N^{\mathit{ex}}(\tau') : \tau' \cap q.\rho \ne \emptyset\}$.
If the query uses indices $\{\tau : y_{q,\tau}=1\}$ to look for its nearest neighbors, the query-time cost is:
\[
\textsf{Cost}(q) = \sum_{\tau : y_{q,\tau} = 1} \log s_\tau + \alpha \sum_{N^{\mathit{ex}}(\tau') \in \mathcal{N}_q} \, |N^{\mathit{ex}}(\tau')| \cdot \prod_{\tau : y_{q,\tau} = 1} (1 - z_{\tau', \tau})
\]
The second term inflates the cost proportionally to any portion of $q$'s authorized blocks that are not covered by the selected indices.

The overall workload objective is:
\begin{equation}
\label{eq:full-minlp-obj}
\min_{x,z,y} \sum_{q \in \mathcal{Q}} q.\omega
\left(
  \sum_{\tau \in \mathcal{T}} y_{q,\tau} \log s_\tau
  + \alpha \sum_{N^{\mathit{ex}}(\tau') \in \mathcal{N}_q} \, |N^{\mathit{ex}}(\tau')| \cdot \prod_{\tau : y_{q,\tau} = 1} (1 - z_{\tau', \tau})
\right)
\end{equation}
subject to:
\begin{align}
  & \forall \tau', \tau: \quad z_{\tau',\tau} \le x_\tau
    && \text{(include in materialized index)} \label{con:minlp-use} \\
  & \forall \tau', \tau: \quad 0 \le z_{\tau',\tau} \le a_{\tau',\tau}
    &&  a_{\tau',\tau} = 1 \text{ if } \tau' \cap \tau \ne \emptyset \text{ else } 0 \label{con:minlp-lattice} \\
  & \forall \tau: \quad s_\tau = \sum_{\tau'} z_{\tau',\tau} |N^{\mathit{ex}}(\tau')| \label{con:minlp-size} \\
  & \sum_{\tau} x_\tau s_\tau \le \beta |\mathcal{D}| && \text{(storage budget)} \label{con:minlp-store}
\end{align}
All variables are binary; $\beta \ge 1$ bounds the storage amplification.
Because $s_\tau$ and the cost expression contain products and logs, the objective \ref{eq:full-minlp-obj} is non-convex.

A common simplification is to pre-enumerate a finite catalogue $\mathcal{G}_c$ of valid groups (from merge/copy operations), treat $|g|$ as a constant, and decide which groups to index.
Let $x_g, y_{q,g}, u_{q,N^{\mathit{ex}}(\tau')}$ be binary variables as in \S\ref{sec:problem-spec}.
With $\log|g|$ constant, the objective becomes linear:
\begin{equation}
\label{eq:milp-obj-compact}
\min_{x,y,u} \sum_{q \in \mathcal{Q}} q.\omega
\left(
  \sum_{g \in \mathcal{G}_c} y_{q,g} \log|g|
  + \alpha \sum_{N^{\mathit{ex}}(\tau') \in \mathcal{N}_q} u_{q,N^{\mathit{ex}}(\tau')} |N^{\mathit{ex}}(\tau')|
\right)
\end{equation}
subject to:
\begin{align}
  y_{q,g} &\le x_g && \forall q,g \label{con:milp-use} \\
  \sum_{g : N(\tau') \in g} y_{q,g} + u_{q,N(\tau')} &\ge 1 && \forall q, N(\tau') \in \mathcal{N}_q \label{con:milp-cover} \\
  x_g &= 0 && \text{if } |g| < \Lambda \label{con:milp-thresh} \\
  \sum_g x_g |g| &\le \beta |\mathcal{D}| \label{con:milp-store}
\end{align}
Equations~\ref{eq:milp-obj-compact}--\ref{con:milp-store} define a 0--1 MILP.

Enumerating $\mathcal{G}_c$ reduces the problem to a MILP, but scale remains a challenge:
$|\mathcal{G}_c|$ can reach tens of thousands; $y_{q,g}$ spans millions of binaries for realistic workloads,
and runtimes remain prohibitive. Moreover, fixed catalogues cannot adapt to shifting role distributions or data updates.

\section{Node Relations in Lattices}
\label{app:lattice-relations}

Below we define the \emph{descendant--ancestor} and \emph{relative} relations used by \name and \effname.

\begin{definition}[Descendant--ancestor relation]
\label{def:lattice-ancestor-descendant}
  For two exclusive blocks $N^{\mathit{ex}}(\tau)$ and
  $N^{\mathit{ex}}(\tau')$ in $\mathcal{L}_{\mathit{ex}}$,
  $N^{\mathit{ex}}(\tau)$ is an \emph{ancestor} of
  $N^{\mathit{ex}}(\tau')$, and $N^{\mathit{ex}}(\tau')$ is a
  \emph{descendant} of $N^{\mathit{ex}}(\tau)$, when both conditions hold:
  \begin{enumerate}
      \item \textbf{Containment:} $\tau \subset \tau'$, so every role
      authorized at the ancestor $N^{\mathit{ex}}(\tau)$ is also authorized at the descendant $N^{\mathit{ex}}(\tau')$.
      \item \textbf{Reachability:} there is a directed path from
      $N^{\mathit{ex}}(\tau)$ to $N^{\mathit{ex}}(\tau')$ in
      $\mathcal{L}_{\mathit{ex}}$.
  \end{enumerate}
  Unlike the parent--child relation defined in Definition~\ref{def:exclusive-lattice}, \emph{descendant--ancestor} relation does not require adjacency.
\end{definition}

\begin{definition}[Sibling relation]
\label{def:lattice-sibling}
  Two exclusive blocks $N^{\mathit{ex}}(\tau)$ and
  $N^{\mathit{ex}}(\tau')$ in $\mathcal{L}_{\mathit{ex}}$ are
  \emph{siblings} when both conditions hold:
  \begin{enumerate}
      \item \textbf{Same layer:} $|\tau| = |\tau'|$ and $\tau \ne \tau'$.
      \item \textbf{Shared roles:} $\tau \cap \tau' \ne \emptyset$.
  \end{enumerate}
\end{definition}

\section{Copy Operation Dominance}\label{sec:copy_dominance}

To justify prioritizing copy operations, we compare their effect on query workloads.
Let $\textsf{Cost}(Q,\mathcal{L})$ denote the cost of evaluating a workload $Q$ over lattice $\mathcal{L}$.
Consider a child node $N_c=N(\tau_c)$ and an ancestor $N_a=N(\tau_a)$.
Let $\mathcal{L}_{\mathit{copy}}$ be the lattice where $N_c$ is copied into $N_a$ (so both remain), and $\mathcal{L}_{\mathit{merge}}$ the lattice where $N_c$ is merged into $N_a$ (and removed). The following result shows that copying never performs worse than merging, and may strictly improve cost for workloads that isolate $N_c$.

\begin{theorem}[Copy Dominance]
\label{thm:copy-dominates}
Suppose both copying and merging $N_c$ into $N_a$ yield positive cost reduction for queries that involve roles in $\tau_a$. Then the following hold for all queries in $Q$:
\begin{itemize}
    \item We have $\textsf{Cost}(Q,\mathcal{L}_{\mathit{copy}})$  $< \textsf{Cost}(Q,\mathcal{L}_{\mathit{merge}})$ if a query accesses only $N_c$.

    \item We have $\textsf{Cost}(Q,\mathcal{L}_{\mathit{copy}}) = \textsf{Cost}(Q,\mathcal{L}_{\mathit{merge}})$ if a query accesses both $N_c$ and $N_a$.
\end{itemize}
\end{theorem}

\begin{proof}
After either operation the ancestor becomes
$N_a'=N_a\cup N_c$; the two lattices differ only in whether $N_c$ remains.
For a query $q$ authorized for both $N_c$ and $N_a$, the optimal plan in
either lattice routes $q$ to $N_a'$, which covers both blocks at identical
cost, so $\textsf{Cost}(q,\mathcal{L}_{\mathit{copy}})=
\textsf{Cost}(q,\mathcal{L}_{\mathit{merge}})$.
For a query $q$ authorized for $N_c$ but not $N_a$, $\mathcal{L}_{\mathit{copy}}$
still offers the original pure node $N_c$ at cost
$C_\theta(|N_c|,\textsf{efs})$. In $\mathcal{L}_{\mathit{merge}}$ the only node
covering $N_c$'s vectors is $N_a'$, which is impure for $q$ with
$\lambda=\lceil|N_a'|/|N_c|\rceil\ge 2$ and incurs cost
$C_\theta(|N_a'|,\lambda\,\textsf{efs})>C_\theta(|N_c|,\textsf{efs})$ under
Definition~\ref{def:hnsw-cost}. Hence we have
$\textsf{Cost}(q,\mathcal{L}_{\mathit{copy}})$ 
$<\textsf{Cost}(q,\mathcal{L}_{\mathit{merge}})$.
\end{proof}

\section{Supplementary Materials for \name}
\label{app:veda-supplementary}
This section provides supplementary materials for \name, including the overall
workflow pseudocode, details of the finalization phase, correctness proofs for
the copy and merge phases, and complementary algorithms for the copy and merge
phases.

\subsection{Overall Workflow}
\label{app:veda-overview-algorithm}

Algorithm~\ref{alg:adaptive} summarizes how \name converts the exclusive
lattice into the final vector storage layout. The algorithm alternates between
budgeted copy operations and merge operations, maintains the query plans
affected by each refinement, and then finalizes the lattice by separating
unindexable nodes into leftovers and building the corresponding HNSW indices.

\begin{algorithm}[h!]
\caption{\textbf{\name - Overview}}
\label{alg:adaptive}
\begin{algorithmic}[1]

\Require exclusive lattice $\mathcal{L}_\mathit{ex}$, indexing threshold $\Lambda$, storage amplification budget $\beta$, role set $\mathcal{R}$ on $\mathcal{D}$.

\State $\mathcal{L}\gets \mathcal{L}_\mathit{ex}$, $\mathcal{DA}\gets \mathsf{get\_child\_ancestor\_pairs}(\mathcal{L}_\mathit{ex})$

\State $\textsf{QP} = \textsf{GetQueryPlans}(\mathcal{R}, \mathcal{L})$\Comment{Detailed in \S\ref{sec:qplan}}

\While{\texttt{True}} \Comment{Iterate Algorithm~\ref{alg:veda-copy} and Algorithm~\ref{alg:veda-merge}}
\If{$\beta > 0$}~$\mathcal{L}, \textsf{QP}\gets \mathsf{Copy}(\mathcal{L}, \mathcal{L}_\mathit{ex}, \beta, \mathcal{DA}, \textsf{QP})$ 
\EndIf
\If{not first round and no copy applied}~\textbf{break}
\EndIf
\State $\mathcal{L}, \textsf{QP}\gets \mathsf{Merge}(\mathcal{L}, \mathcal{L}_\mathit{ex}, \mathcal{DA}, \textsf{QP})$ 
\If{no merge applied}~\textbf{break}
\EndIf
\EndWhile
\State $\mathcal{L}, \textsf{QP}\gets \mathsf{split\_small\_nodes\_into\_leftovers}(\mathcal{L}, \Lambda)$

\If{$\frac{|\mathcal{L}|}{|\mathcal{D}|}<\beta$} \Comment{Storage reclaimed after merging}
\State $\mathcal{L}, \textsf{QP}\gets \mathsf{HandleSuperImpureNodes}(\mathcal{L}, \mathcal{L}_\mathit{ex}, \beta, \mathcal{DA}, \textsf{QP})$
\EndIf

\State $\mathcal{U}$, $\mathcal{I}\gets \mathsf{build\_vector\_storage}(\mathcal{L})$ 

\State \textbf{return} $\mathcal{U}$, $\mathcal{I}$
\end{algorithmic}
\end{algorithm}

\newpage

\begin{algorithm}
    \caption{\textbf{\name -- Super-Impure Node Refinement}}
    \label{alg:adaptive-finalize}
    \begin{algorithmic}[1]

    \Require current lattice $\mathcal{L}$, exclusive lattice $\mathcal{L}_{\mathit{ex}}$, 
    query plan $\textsf{QP}$, and remaining storage budget $\mathit{buf}$.

    \State \textbf{// Step 1: Collect impure $(r, N)$ candidates}
    \State $\mathcal{C} \gets \emptyset$, $\textsf{Copied} \gets \emptyset$
    \State $\mathit{ref}[N] \gets |\{r : N \in \textsf{QP}(r)\}|$ for each $N \in \mathcal{L}$ \Comment{\# of roles that have $N$ in $\textsf{QP}(r)$}
    \For{role $r \in \textsf{QP}$, node $N \in \textsf{QP}(r)$}
        \State $\mathit{pure}^{\mathit{ex}}\gets \{N^{\mathit{ex}}(\tau) \in N : r \text{ is permitted by } \tau\}$
        \State $\mathit{pure}_s\gets \sum_{N^{\mathit{ex}}(\tau) \in \mathit{pure}^{\mathit{ex}}} |N^{\mathit{ex}}(\tau)|$
        \If{$0 < \mathit{pure}_s < |N|$}
            \State Add $(r, N, \mathit{pure}^{\mathit{ex}}, \mathit{imp}{=}\frac{|N|}{\mathit{pure}_s}, \mathit{pure}_s)$ to $\mathcal{C}$
        \EndIf
    \EndFor
    \State Sort $\mathcal{C}$ by $(\mathit{imp}, -\mathit{pure}_s)$ decreasingly \Comment{Prioritize impure nodes}

    \State \textbf{// Step 2: Refine each candidate, most-impure first}
    \For{$(r, N, \mathit{pure}^{\mathit{ex}}, \cdot, \cdot) \in \mathcal{C}$}
        \If{$N \notin \textsf{QP}(r)$ \textbf{or} $N \notin \mathcal{L}$}~\textbf{continue}
        \EndIf
        \State $\mathit{copy}_s \gets \sum_{N^{\mathit{ex}}(\tau) \in \mathit{pure}^{\mathit{ex}} \setminus \textsf{Copied}}\bigl|N^{\mathit{ex}}(\tau)\bigr|$
        \If{$\mathit{buf} < \mathit{copy}_s$}~\textbf{continue}
        \EndIf
        \For{$N^{\mathit{ex}}(\tau) \in \mathit{pure}^{\mathit{ex}}$} \Comment{Algorithm~\ref{alg:separate_pure_parts}}
            \State $(N^{\mathit{ex}}(\eta), \mathit{buf}) \gets \mathsf{Separate}(N^{\mathit{ex}}(\tau), \mathit{buf}, \mathcal{L}, \textsf{Copied})$
            \State $\textsf{QP}(r).\textsf{add}(N^{\mathit{ex}}(\eta))$
        \EndFor
        \State $\textsf{QP}(r).\textsf{remove}(N)$; $\mathit{ref}[N] \gets \mathit{ref}[N] - 1$
        \If{$\mathit{ref}[N] = 0$}~delete $N$ from $\mathcal{L}$; $\mathit{buf} \gets \mathit{buf} + |N|$
        \EndIf
    \EndFor
    \State \Return $\mathcal{L}$, $\textsf{QP}$

\end{algorithmic}
\end{algorithm}

\subsection{Phase 3: Finalization}
\label{sec:veda_finalize}

In the finalization phase, any remaining nodes that are not indexable are decomposed into exclusive blocks for efficient linear scan.
As there might be duplicated data between data groups that free storage after decomposition, \name identifies nodes that are ``super impure'' to  queried roles and create copies of the pure parts to utilize the freed storage.

Algorithm~\ref{alg:adaptive-finalize} formalizes this procedure.
The algorithm first computes each node $N$'s reference count in $\mathcal{L}$, i.e., the number of roles that access a node $N$, denoted $\mathit{ref}[N]$. Then, for each role $r$, the algorithm computes its impurity of each node in $\textsf{QP}(r)$ with the pure part of $N$ for $r$, denoted $\mathit{pure}^{\mathit{ex}}$.
If this pure part is smaller than $N$, $N$ is impure for $r$, and the pair is added to a candidate list with an impurity score.
Candidates are processed in decreasing impurity order.
Before refining a candidate, the algorithm computes the additional copy size $\mathit{copy}_s$ by counting only exclusive blocks in $\mathit{pure}^{\mathit{ex}}$ that have not already been copied by earlier refinements.
If the remaining budget is sufficient, the algorithm materializes each exclusive block into a standalone node that can be used directly in the query plan.
When the block is already standalone in $\mathcal{L}$, the algorithm reuses it without spending budget.
Otherwise, the block is copied into $\mathcal{L}$: it keeps its original key if that key is free, and uses a renamed key if the original key is already occupied by a non-standalone node.
The query plan $\textsf{QP}(r)$ is redirected to these materialized blocks, and the original impure node is removed from $\textsf{QP}(r)$.
When this update makes $\mathit{ref}[N]=0$, the node is not referenced by any query and is safe to delete from $\mathcal{L}$ to reclaim storage. Algorithm~\ref{alg:separate_pure_parts} describes the details of splitting the pure parts from impure nodes.

\begin{algorithm}
    \caption{\textbf{Separate Pure Parts from Impure Node}}
    \label{alg:separate_pure_parts}
    \begin{algorithmic}[1]

    \Require pure exclusive node $N^{\mathit{ex}}(\tau)$, current lattice $\mathcal{L}$, exclusive lattice $\mathcal{L}_{\mathit{ex}}$, copied-node set $\textsf{Copied}$, and remaining storage budget $\mathit{buf}$.

    \Function{Separate}{$N^{\mathit{ex}}(\tau)$, $\mathit{buf}$, $\mathcal{L}$, $\textsf{Copied}$}
        \Comment{Ensure $N^{\mathit{ex}}(\tau)$ is queryable as a standalone key}
        \If{$N^{\mathit{ex}}(\tau)$ is standalone in $\mathcal{L}$}~\Return $N^{\mathit{ex}}(\tau),\ \mathit{buf}$
        \EndIf
        \State $\eta\gets(-1,\tau)$ \Comment{In case $\tau$ occupied by a non-standalone node}
        \If{$N^{\mathit{ex}}(\tau) \in \textsf{Copied}$}~\Return $N^{\mathit{ex}}(\eta),\ \mathit{buf}$
        \EndIf
        \State $\textsf{Copied}.\textsf{add}(N^{\mathit{ex}}(\tau))$
        \If{$N^{\mathit{ex}}(\tau) \in \mathcal{L}$}
        \State Copy $N^{\mathit{ex}}(\tau)$ to $\mathcal{L}$ as $N^{\mathit{ex}}(\eta)$
        \State \Return $N^{\mathit{ex}}(\eta),\ \mathit{buf} - |\mathcal{L}_{\mathit{ex}}[N^{\mathit{ex}}(\tau)]|$
        \Else
        \State Copy $N^{\mathit{ex}}(\tau)$ to $\mathcal{L}$ as $N^{\mathit{ex}}(\eta)$
        \State \Return $N^{\mathit{ex}}(\tau),\ \mathit{buf} - |\mathcal{L}_{\mathit{ex}}[N^{\mathit{ex}}(\tau)]|$
        \EndIf
    \EndFunction

    \end{algorithmic}
\end{algorithm}

\subsection{\name Correctness Proofs}
\label{app:veda-correctness-proofs}
Below we provide the proofs for Theorem~\ref{thm:copy-phase-correctness} and Theorem~\ref{thm:merge-phase-correctness}.

\begin{proof}[Proof of Theorem~\ref{thm:copy-phase-correctness}]
(1) A copy is applied only if $f(\mathcal{L}_{t-1},e^*)\ge 0$, i.e., the
numerator of Equation~\ref{eqn:benefit_function} is nonnegative, so
$\textsf{AvgCost}$ is non-increasing. (2) Each copy decrements $\mathit{buf}$
by exactly $\Delta S(e^*)$ and is admitted only when
$\Delta S(e^*)\le\mathit{buf}$; by induction
$|\mathcal{L}_t|\le\beta|\mathcal{L}_{\mathit{ex}}|$. (3) Each copy strictly
enlarges some ancestor, so a given ordered pair $(N_c,N_a)$ is applied at most
once; $\textsf{PR}$ is finite and the guard in
Equation~\ref{eqn:greedyconditions} eventually fails.
\end{proof}

\begin{proof}[Proof of Theorem~\ref{thm:merge-phase-correctness}]
(1) A merge is applied only when its benefit is strictly positive, and since
$\Delta S(e)=0$ the numerator of Equation~\ref{eqn:benefit_function} equals
the benefit. (2) Merging takes the set union of the two nodes' vectors and
deletes the child, so total stored vectors drop by exactly the overlap.
(3) Each merge removes one node from $\mathcal{L}$; at most
$|\mathcal{N}_\mathit{ex}|$ merges can occur.
\end{proof}

\subsection{Copy-Phase Complementary Algorithms}
\label{app:veda-copy-helpers}

Algorithms~\ref{alg:veda-copy-complementary} and~\ref{alg:veda-copy-benefit}
describe complementary algorithms for the copy-phase in \name called from
Algorithm~\ref{alg:veda-copy}. \textsf{GetCopyPairs} refreshes the priority
list of available descendant--ancestor copy pairs, rescoring only pairs whose
ancestor was touched by the previous copy. \textsf{GetCopyBenefit} evaluates the benefit ratio for each candidate copy pair based on
Equation~\ref{eqn:benefit_function}, 
and
\textsf{GetDiff} in Algorithm~\ref{alg:veda-copy-benefit} computes the change in a single role's plan cost when its
query plan switches from the current cover to one that uses the enlarged
ancestor.

\begin{algorithm}
\caption{\textbf{GetCopyPairs} for \name}
\label{alg:veda-copy-complementary}
\begin{algorithmic}[1]


\Require lattice $\mathcal{L}$, available copy pairs $\textsf{PR}$, last updated ancestor
$N_a^\mathit{last}$, and query plan $\textsf{QP}$ for $\mathcal{L}$.
    \For{$(N_c, N_a)\in\textsf{PR}$}
        \If{$!N_a^\mathit{last}$ or $N_a^\mathit{last}=N_a$}
        \If{$N_c \notin \mathcal{L}$}~$\textsf{PR}[(N_a, N_c)].\textsf{delete}()$\EndIf
        \State $\textsf{PR}[(N_a, N_c)]\gets \mathsf{GetCopyBenefit}(N_a, N_c,  \mathcal{L}, \mathcal{L}_{\mathit{ex}}, \textsf{QP})$
        \EndIf
    \EndFor
    \State Sort $\textsf{PR}$ by their benefits
    \State \textbf{return } $\textsf{PR}$

\end{algorithmic}
\end{algorithm}

\begin{algorithm}
    \caption{\textbf{GetCopyBenefit} for \name}
    \label{alg:veda-copy-benefit}
    \begin{algorithmic}[1]
\Require ancestor node $N_a$, descendant node $N_c$, exclusive lattice
$\mathcal{L}_{\mathit{ex}}$, lattice $\mathcal{L}$, and query plan $\textsf{QP}$
for $\mathcal{L}$.

\State $\Delta\gets 0$, $\Delta S(e(N_c, N_a))\gets |N_a\cup\mathcal{L}_{\mathit{ex}}[N_c]| - |N_a|$
\For{$r\in \textsf{QP}$}
\If{$N_a\in \textsf{QP}(r)$} 
\If{$N_c\in\textsf{QP}(r)$}
\State $\Delta\gets \Delta + \log (|N_a| + 1) + \log(|\mathcal{L}_{\mathit{ex}}[N_c]| + 1)$
\State $\Delta\gets \Delta - \log (|N_a + \mathcal{L}_{\mathit{ex}}[N_c]| + 1)$
\Else 
\State $\mathit{qp}^r\gets \mathsf{GetCoverage}(r, \mathcal{L})$
\State $\Delta\gets \mathsf{RecomputeBenefit}(qp^r, \textsf{QP}(r), \mathcal{L}, N_a, N_c^\mathit{ex})$
\EndIf
\EndIf
\EndFor
    \State \textbf{return } $\Delta/\Delta S(e(N_c, N_a))$

\end{algorithmic}
\end{algorithm}

\begin{algorithm}
    \caption{\textbf{RecomputeBenefit} for \name}
    \label{alg:veda-get-changes-in-benefit-for-renewed-qp}
    \begin{algorithmic}[1]

\Require recomputed query plan $qp^r$, current query plan $qp$, lattice $\mathcal{L}$, ancestor node $N_a$, and data to be copied into $N_a$, $N_c^\mathit{ex}$.

\Function{RecomputeBenefit}{$qp^r$, $qp$, $\mathcal{L}$, $N_a$, $N_c^\mathit{ex}$}
    \State $\mathit{qp}_\mathit{unq}\gets \mathit{qp}\setminus \mathit{qp}^r$, 
    $\mathit{qp}_\mathit{unq}^r\gets \mathit{qp}^r\setminus \mathit{qp}$
    \If{$N_a\in \mathit{qp}_\mathit{unq}$~and~$N_a\in \mathit{qp}_\mathit{unq}^r$}
    \State $\mathit{qp}_\mathit{unq}.\texttt{add}(N_a)$, $\mathit{qp}_\mathit{unq}^r.\texttt{add}(N_a)$
    \EndIf
    \For{$\tau\in \mathit{qp}_\mathit{unq}$}~$\Delta\gets \Delta + \log(|\mathcal{L}[\tau]| + 1)$
    \EndFor
    \For{$\tau\in \mathit{qp}_\mathit{unq}^r$}
    \If{$\tau\neq N_a$}~$\Delta\gets \Delta - \log(|\mathcal{L}[\tau]| + 1)$
    \Else~$\Delta\gets \Delta - \log(|N_c^\mathit{ex}\cup N_a| + 1)$
    \EndIf
    \EndFor
    \State \textbf{return}~$\Delta$
\EndFunction

\end{algorithmic}
\end{algorithm}

\subsection{Merge-Phase Complementary Algorithms}
\label{app:veda-merge-helpers}

Algorithm~\ref{alg:veda-merge-complementary} supplies the merge-phase
analogues called from Algorithm~\ref{alg:veda-merge}.
\textsf{GetMergeBenefit} differs from the copy case in that it must also
charge the impurity introduced for roles that previously read the child as a
pure node.

\begin{algorithm}[h!]
\caption{\textbf{GetMergePairs} for the Merge Phase in \name}
\label{alg:veda-merge-complementary}
\begin{algorithmic}[1]

\Require lattice $\mathcal{L}$, available merge pairs $\mathit{pairs}$, last merged ancestor
$a^\mathit{last}$, last merged descendant $c^\mathit{last}$, and query plan
$\textsf{QP}$.

\Function{GetMergePairs}{$\mathcal{L}$, $\mathit{pairs}$, $a^\mathit{last}$, $c^\mathit{last}$, $\textsf{QP}$} 
    \For{$(\mathit{c}, \mathit{a})\in\mathit{pairs}$}
        \If{$c\notin \mathcal{L}$~or~$a\notin \mathcal{L}$}~delete $(c, a)$; \textbf{continue}
        \EndIf
        \If{$a \in [a^\mathit{last}, \texttt{None}]$ or  $c \in [c^\mathit{last}s, \texttt{None}]$}
        \State $\mathit{pairs}[(a, c)]\gets \textsf{GetMergeBenefit}(\mathit{a}, \mathit{c}, \mathcal{L})$
        \EndIf
    \EndFor
    \State \textbf{return } $\mathit{pairs}$
\EndFunction

\Function{GetMergeBenefit}{$a$, $c$, $\mathcal{L}$, $\textsf{QP}$} 
\label{alg:adaptive_merge_benefit_beg}
\State $\Delta\gets 0$
\For{$r\in \textsf{QP}$}
\If{$a\in \textsf{QP}(r)$ and $c\in\textsf{QP}(r)$} 

\State $\Delta\gets \Delta + \log (|\mathcal{L}[a]| + 1) + \log(|\mathcal{L}[c]| + 1)$
\State $\Delta\gets \Delta - \log (|\mathcal{L}[a] + \mathcal{L}[c]| + 1)$
\ElsIf{$a\in \textsf{QP}(r)$ or $c\in\textsf{QP}(r)$}
\State $\mathit{qp}^r\gets \textsf{GetCoverage}(r, \mathcal{L})$

\State $\Delta\gets \textsf{RecomputeBenefit}(qp^r, \textsf{QP}(r), \mathcal{L}, a, \mathcal{L}[c])$
\EndIf
\EndFor
    \State \textbf{return } $\Delta$
    \label{alg:adaptive_merge_benefit_end}
\EndFunction

\end{algorithmic}
\end{algorithm}

\section{Supplementary Materials for \effname}
\label{app:effveda-helpers}
This section provides supplementary materials for \effname: the overall
workflow, the helper routines for the copy and merge phases, and the proofs
deferred from \S\ref{sec:eff_veda}.

\subsection{Complementary Algorithms}
\label{app:effveda-complementary-algorithms}
This subsection describes the complementary algorithms for Copy and Merge in \effname in \S\ref{sec:eff_veda}.
Algorithm~\ref{alg:eff-adaptive} gives the
overall workflow of \effname.
Algorithm~\ref{alg:eff-veda-copy-find-best-partition} defines
\textsf{FindBestPartition}, whose goal is to choose the best valid
partition for a node under the remaining storage budget.

\begin{algorithm}[h!]
\caption{\textbf{\effname - Overview}}
\label{alg:eff-adaptive}
\begin{algorithmic}[1]

\Require exclusive lattice $\mathcal{L}_{\mathit{ex}}$, role set  $\mathcal{R}$, indexing threshold $\Lambda$, and \textsf{SA} budget $\beta$.

\State $\mathcal{L}\gets \mathcal{L}_{\mathit{ex}}$, $\mathcal{DA}\gets \mathsf{get\_child\_ancestor\_pairs}(\mathcal{L}_{\mathit{ex}})$ \Comment{All ancestor--descendant pairs induced by reachability in $\mathcal{L}_{\mathit{ex}}$}

\State $\textsf{QP} = \mathsf{GetQueryPlans}(\mathcal{R}, \mathcal{L})$

\If{$\beta > 0$}
\State $\mathcal{L}, \textsf{QP}\gets \mathsf{Copy}(\mathcal{L}, \mathcal{L}_{\mathit{ex}}, \beta, \mathcal{DA}, \textsf{QP})$\Comment{Algorithm~\ref{alg:eff-veda-copy}}
\EndIf

\State $\mathcal{L}, \textsf{QP}\gets \mathsf{Merge}(\mathcal{L}, \mathcal{L}_{\mathit{ex}}, \mathcal{DA}, \textsf{QP})$ \Comment{Algorithm~\ref{alg:eff-veda-merge}}

\State $\mathcal{L}, \textsf{QP}\gets \mathsf{split\_small\_nodes\_into\_leftovers}(\mathcal{L}, \Lambda)$

\If{$\frac{|\mathcal{L}|}{|\mathcal{D}|}<\beta$} \Comment{Storage reclaimed after merging}
\State $\mathcal{L}, \textsf{QP}\gets \mathsf{HandleSuperImpureNodes}(\mathcal{L}, \mathcal{L}_{\mathit{ex}}, \beta, \mathcal{DA}, \textsf{QP})$
\EndIf



\State $\mathcal{U}$, $\mathcal{I}\gets \mathsf{build\_vector\_storage}(\mathcal{L})$ 

\State \textbf{return} $\mathcal{U}$, $\mathcal{I}$
\end{algorithmic}
\end{algorithm}

\begin{algorithm}[t]
       \caption{\textbf{{FindBestPartition}} for \effname-Copy}
       \label{alg:eff-veda-copy-find-best-partition}
       \begin{algorithmic}[1]
       \Require node $N(\tau)$, ancestor set $A_c^\tau$ for $N(\tau)$
                     (sorted by role-set cardinality desc.), lattice $\mathcal{L}$,
                     remaining buffer $\mathit{buf}$.
       \If{$A_c^\tau=\emptyset$ or $|N(\tau)|>\mathit{buf}$} \Return $(\emptyset,0)$ \EndIf
       \State $\textsf{P}_\tau^*\gets\emptyset$, $\Delta^*\gets 0$
       \Statex \textit{// Stage 1: find two-way covers; keep best singleton as fallback}
       \ForAll{$N_a(\tau')\in A_c^\tau$}\label{ln:fbp-two}
              \State  $\textsf{P}_\tau \gets \{N_a(\tau')\}$, $\Delta'\gets|\tau'|\cdot\Delta_c\,\bigl(N(\tau),N_a(\tau')\bigr)$,$\tau''\gets\tau\setminus\tau'$ \label{ln:fbp-compl}
              \If{$N_a(\tau'')\in A_c^\tau$}\label{ln:fbp-compl-hit}
              \State $\textsf{P}_\tau.add(N_a(\tau''))$, $\Delta'\gets\Delta'
                     +|\tau''|\cdot\Delta_c\,\bigl(N(\tau),N_a(\tau'')\bigr)$
              \EndIf
              \If{$\Delta'>\Delta^*$}~$(\textsf{P}_\tau^*,\Delta^*)\gets(\textsf{P}_\tau,\Delta')$
              \EndIf
              \label{ln:fbp-seed}
       \EndFor
       \If{$|\textsf{P}_\tau^*|=2$}~\textbf{return} $(\textsf{P}_\tau^*,\Delta^*/|N(\tau)|)$ \EndIf
       \Statex \textit{// Stage 2: greedy disjoint extension if no two-way cover found}
       \State 
              $N_a(\tau_s)\gets\textsf{P}_\tau^*[0]$, $\tau^{\mathsf{res}}\gets\tau\setminus\tau_s$, $\Delta S\gets 0$
              \label{ln:fbp-extend}
       \ForAll{$N_a(\tau')\in A_c^\tau$} \Comment{$A_c^\tau$ is sorted by role-set cardinality desc.}
              \label{ln:fbp-greedy}

              \If{$\tau'\subseteq \tau^{\mathsf{res}}$}
              \label{ln:fbp-disjoint}
              \State $\Delta S\gets\Delta S + |N(\tau)|$
              \If{$\Delta S>\mathit{buf}$} \textbf{break} \EndIf
              \State $\tau^{\mathsf{res}}\gets\tau^{\mathsf{res}} \setminus \tau'$, $\textsf{P}_\tau\gets\textsf{P}_\tau\cup\{N_a(\tau')\}$
              \State $\Delta^*\gets\Delta^*
                     +|\tau'|\cdot\Delta_c\,\bigl(N(\tau),N_a(\tau')\bigr)$
              \EndIf
       \EndFor
       \If{$\tau^{\mathsf{res}}\neq\emptyset$}
              \label{ln:fbp-residual}
              ~$\textsf{P}_\tau\gets\textsf{P}_\tau\cup\{\tau^{\mathsf{res}}\}$
                     \Comment{Residual roles}
       \EndIf
       \State \Return $\bigl(\textsf{P}_\tau, \Delta^*/(|N(\tau)|\cdot(|\textsf{P}_\tau|-1))\bigr)$
              \label{ln:fbp-return}
       \end{algorithmic}
       \end{algorithm}

\subsection{Supplementary Proofs for \effname}\label{app:effveda-proofs}

This subsection collects the proofs deferred from \S\ref{sec:eff_veda}.
We begin with Theorem~\ref{thm:node_purity_after_copying}.

\begin{proof}[Proof of Theorem~\ref{thm:node_purity_after_copying}]
  Only the ancestors in $\textsf{P}_\tau$ change after copying $N(\tau)$ into ancestors in $\textsf{P}_\tau$. 
  For each ancestor $N_a(\tau_j)\in\textsf{P}_\tau$, 
  before taking vectors from $N(\tau)$, all vectors in $N_a(\tau_j)$ are authorized for $\tau_j$. 
  Since all vectors in $N(\tau)$ are authorized for $\tau$, and $\tau_j\subsetneq\tau$ (i.e., $N_a(\tau_j)$ is an ancestor of $N(\tau)$), every data in $N(\tau)$ is also authorized for $\tau_j$. Thus, $N_a(\tau_j)$ remains pure for $\tau_j$ after copying $N(\tau)$ into $N_a(\tau_j)$.
\end{proof}

Below we provide the proof for Lemma~\ref{lem:copy_locality}.

\begin{proof}[Proof of Lemma~\ref{lem:copy_locality}]
  After copying $N(\tau)$ into ancestors in $\textsf{P}_\tau$, only the query plans of roles in $\tau$ are affected.
For any role $r\in\mathcal{R}$ and its query plan $\textsf{QP}(r)$, if $r\notin\tau$ then $r$'s query plan and its query cost remain unchanged. For role $r\in \tau$, since the role set of each ancestor in $\textsf{P}_\tau$ is disjoint, there is a unique $N_a(\tau_j)\in\textsf{P}_\tau$ such that $r\in\tau_j$. Before the
copy, $r$'s plan contains both $N_a(\tau_j)$ and $N(\tau)$; after the copy it
contains the single node $N_a(\tau_j)\cup N(\tau)$, and all remaining nodes in
the plan are unchanged because purity is preserved
(Theorem~\ref{thm:node_purity_after_copying}). The cost difference for $r$ is
exactly $\Delta_c(N(\tau),N_a(\tau_j))$. Summing over $r\in\tau$ and grouping
by the $\tau_j$ gives
$\sum_{\tau_j}|\tau_j|\cdot\Delta_c(\cdot)$; dividing by $|\mathcal{R}|$
yields the change in $\textsf{AvgCost}$.
\end{proof}

Below we provide the proof for Lemma~\ref{lem:copy_positivity}.
\begin{proof}[Proof of Lemma~\ref{lem:copy_positivity}]
Denote the size of $N_a(\tau_j)$ and $N(\tau)$ as $n_a$ and $n_c$ respectively. According to
Equation~\ref{eqn:derivedcostmodel} and~Equation~\ref{eqn:per_role_gain}, we have 
$\Delta_c(N(\tau),N_a(\tau_j)) = a\bigl(\log n_a+\log n_c-\log(n_a{+}n_c)\bigr)+b\,\textsf{efs}+c
          = a\log\!\tfrac{n_an_c}{n_a+n_c}+b\,\textsf{efs}+c$.
The first term is nonnegative whenever $n_a,n_c\ge 2$ (in fact $2 << \Lambda$ according to Figure~\ref{fig:vary-dimensions} and Figure~\ref{fig:vary-dimensions-app}), and
$b\,\textsf{efs}+c>0$ under any calibration that assigns positive cost to a
probe. Thus $\Delta_c(N(\tau),N_a(\tau_j))>0$; combined with the positive numerator $|N(\tau)|\cdot(|\textsf{P}_\tau|-1)$
in Equation~\ref{eqn:copy_benefit_eff_adaptive_simplified} is a positive combination
of positive terms.
\end{proof}

Below we provide the proof for Theorem~\ref{th:routed_query_for_roles}.

\begin{proof}[Proof of Theorem~\ref{th:routed_query_for_roles}]
  At initialization each $N(\tau)\in\mathcal{L}$ coincides with one frozen $N^\rho(\tau_\rho)$ and $\mathbb{V}(N(\tau))=\{N^\rho(\tau_\rho)\}$. The inherited plan at this point is the natural post-copy plan, which probes $N^\rho(\tau_\rho)$ for role $r$ iff $r\in\tau_\rho$; this plan is valid because the post-copy nodes are pure (Theorem~\ref{thm:node_purity_after_copying}), so the invariant holds.
  Assume the invariant holds before a merge operation where $N(\tau)$ absorbs $N(\tau')$. The absorbed node $N(\tau')$ is  deleted from $\mathcal{L}$, so any role whose post-copy component was routed to $N(\tau')$ must be redirected to the surviving node $N(\tau)$. According to Definition~\ref{def:virtual_decomp}, the virtual decomposition of the survivor $\mathbb{V}(N(\tau))$ becomes $\mathbb{V}(N(\tau))\cup\mathbb{V}(N(\tau')).$
  Thus, the routed roles of the survivor node $N(\tau)$ become $\Pi(N(\tau))\cup\Pi(N(\tau'))$.
  All other nodes keep both their routed roles and their virtual decompositions unchanged. Thus the invariant is preserved after every merge.
\end{proof}







\section{Supplementary Materials for Query Answering}
\label{app:query-exec-supplementary}

This section supplements \S\ref{sec:query-execution}. \S\ref{app:query-answering-algorithms} gives the ILP and greedy
algorithms for query-plan construction; \S\ref{app:topk-exec-details}
details independent top-$k$ execution over a mixed plan;
\S\ref{app:continued-hnsw} gives the continued base-layer HNSW search used by
coordinated search on impure indices; and \S\ref{sec:inflated-k-accuracy}
proves that inflated-$k$ search returns the true authorized top-$k$ with high
probability.

\subsection{Query Plan Construction}
\label{app:query-answering-algorithms}

We give two algorithms for selecting a minimal cover of $\mathcal{D}(r)$ from
the materialized lattice (\S\ref{sec:qplan}).
Algorithm~\ref{alg:query_plan_ILP} solves the exact set-cover ILP and yields
the lowest-cost plan; Algorithm~\ref{alg:query_plan_greedy} greedily picks the
node that covers the most remaining blocks, trading optimality for speed. We
recommend the greedy variant inside the \name and \effname construction loops,
where plans are recomputed many times, and the ILP variant once at the end to
fix the final per-role plans.

\begin{algorithm}[t]
\caption{\textbf{Query Plan Construction (ILP)}}
\label{alg:query_plan_ILP}
\begin{algorithmic}[1]
\Require single roles $\mathcal{R}$, exclusive blocks $\mathcal{N}_{\mathit{ex}}$, location map $\Phi$, where $\Phi(E)$ is the set of nodes in $\mathcal{L}$ containing exclusive block $E$, and node sizes $|N(\tau)|$.
\State $\textsf{QP}\gets\emptyset$
\ForAll{$r\in\mathcal{R}$}
        \State $\mathcal{N}^{\mathit{ex}}(r)\gets\{N^{\mathit{ex}}(\tau)\in\mathcal{N}_{\mathit{ex}}\mid r\in\tau\}$
        \State $\mathsf{cov}\gets\emptyset$, $\mathcal{P}\gets\emptyset$
                \Comment{Selected nodes and pending nodes}
        \ForAll{$N^{\mathit{ex}}(\tau)\in\mathcal{N}^{\mathit{ex}}(r)$}
                \If{$|L_E|=1$}
                        \State $\mathsf{cov}\gets\mathsf{cov}\cup \Phi(N^{\mathit{ex}}(\tau))[0]$
                        \Comment{The node is mandatory}
                \Else~$\mathcal{P}[N^{\mathit{ex}}(\tau)]\gets \Phi(N^{\mathit{ex}}(\tau))$
                \EndIf
        \EndFor
        \State Remove from $\mathcal{P}$ every $N^{\mathit{ex}}(\tau)$ with $\mathcal{P}[N^{\mathit{ex}}(\tau)]\cap\mathsf{cov}\neq\emptyset$
        \State Create binary variable $x_{N(\tau)}$ for each $N(\tau)\in\bigcup_{N^{\mathit{ex}}(\tau)\in\mathcal{P}}\mathcal{P}[N^{\mathit{ex}}(\tau)]$
        \State Minimize $\sum_{N(\tau)} \log_2(|N(\tau)|+1)\cdot x_{N(\tau)}$
        \ForAll{$N^{\mathit{ex}}(\tau)\in\mathcal{P}$}
                \State Add constraint $\sum_{N(\tau)\in\mathcal{P}[N^{\mathit{ex}}(\tau)]}x_{N(\tau)}\ge 1$
        \EndFor
        \State $\mathsf{status}\gets\mathsf{SolveILP}()$
        \If{$\mathsf{status}=\mathsf{optimal}$}
                \ForAll{$N(\tau)$ with binary variable $x_{N(\tau)}$}
                        \If{$x_{N(\tau)}=1$}~$\mathsf{cov}\gets\mathsf{cov}\cup\{N(\tau)\}$ \EndIf
                \EndFor
        \Else \Comment{Fallback to the greedy algorithm (Algorithm~\ref{alg:query_plan_greedy_select})}
                \State $\mathsf{cov}\gets\textsf{GreedySelectCoverage}(\mathsf{cov},\mathcal{P})$
        \EndIf
        \State $\textsf{QP}(r)\gets\mathsf{cov}$
\EndFor
\State \Return $\textsf{QP}$
\end{algorithmic}
\end{algorithm}

\begin{algorithm}[t]
\caption{\textbf{Query Plan Construction (Greedy)}}
\label{alg:query_plan_greedy}
\begin{algorithmic}[1]
\Require single roles $\mathcal{R}$, exclusive blocks $\mathcal{N}_{\mathit{ex}}$, location map $\Phi$, where $\Phi(N^\mathit{ex}_i)$ is the set of nodes in $\mathcal{L}$ containing exclusive block $N^\mathit{ex}_i$, and node sizes $|N(\tau)|$.
\Ensure query plan $\textsf{QP}$ mapping each role $r$ to selected nodes.
\State $\textsf{QP}\gets\emptyset$
\ForAll{$r\in\mathcal{R}$}
       \State $\mathcal{N}^{\mathit{ex}}(r)\gets\{N^{\mathit{ex}}(\tau)\in\mathcal{N}_{\mathit{ex}}\mid r\in\tau\}$
       \State $\mathsf{cov}\gets\emptyset$, $\mathcal{P}\gets\emptyset$
              \Comment{Selected nodes and pending nodes}
       \ForAll{$N^{\mathit{ex}}(\tau)\in\mathcal{N}^{\mathit{ex}}(r)$}
              \If{$|L_E|=1$}
                     \State $\mathsf{cov}\gets\mathsf{cov}\cup \Phi(N^{\mathit{ex}}(\tau))[0]$
                     \Comment{The node is mandatory}
              \Else~$\mathcal{P}[N^{\mathit{ex}}(\tau)]\gets \Phi(N^{\mathit{ex}}(\tau))$
              \EndIf
       \EndFor
       \State Remove from $\mathcal{P}$ every $N^{\mathit{ex}}(\tau)$ with $\mathcal{P}[N^{\mathit{ex}}(\tau)]\cap\mathsf{cov}\neq\emptyset$
       \State $\mathsf{cov}\gets\textsf{GreedySelectCoverage}(\mathsf{cov},\mathcal{P})$\Comment{Algorithm~\ref{alg:query_plan_greedy_select}}
       \State $\textsf{QP}(r)\gets\mathsf{cov}$
\EndFor
\State \Return $\textsf{QP}$
\end{algorithmic}
\end{algorithm}

\begin{algorithm}[t]
       \caption{\textbf{GreedySelectCoverage}}
       \label{alg:query_plan_greedy_select}
       \begin{algorithmic}[1]

\Require current coverage $\mathsf{cov}$ and pending coverage map $\mathcal{P}$ from
exclusive blocks to candidate nodes in $\mathcal{L}$.

\Function{GreedySelectCoverage}{$\mathsf{cov},\mathcal{P}$}
       \ForAll{$N^\mathit{ex}_i\in\mathcal{P}$}
              \If{$\mathcal{P}[N^\mathit{ex}_i]\cap\mathsf{cov}=\emptyset$}
              \State $\mathsf{selected}\gets \mathsf{get\_the\_minimus\_node}(\mathcal{P}[N^\mathit{ex}_i])$
                     \State $\mathsf{cov}\gets\mathsf{cov}\cup\{\mathcal{P}[\mathsf{selected}]\}$
              \EndIf
       \EndFor
       \State \Return $\mathsf{cov}$
\EndFunction
\end{algorithmic}
\end{algorithm}


\subsection{\texorpdfstring{Independent Top-$k$ Query Execution}{Independent Top-k Query Execution}}
\label{app:topk-exec-details}

\begin{algorithm}
\caption{\textbf{Top-$k$ Query Execution}}
\label{alg:topk-execution}
\begin{algorithmic}[1]

\Require Query $q=(\mathbf{x},r)$; plan $\textsf{QP}(r)=(\mathcal{I}^*(r),\mathcal{U}^*(r))$; parameter $k$; ground truth data IDs $\mathcal{D}_\mathit{rids}(r)$ for filtering.

\State $\textsf{RS} \gets [\,]$
\For{each ${\tt idx}\in \mathcal{I}^*(r)$}
    \If{${\tt idx}$ is pure} \Comment{$\mathcal{D}({\tt idx}) \subseteq \mathcal{D}(r)$} 
        \State $\textsf{RS}_{{\tt idx}}\gets \texttt{HNSW}({\tt idx},k)$
    \Else 
        \State $\lambda^r_{{\tt idx}}\gets \lceil \tfrac{|\mathcal{D}({\tt idx})|}{|\mathcal{D}({\tt idx})\cap\mathcal{D}(r)|}\rceil$, $k'\gets \lceil \lambda^r_{{\tt idx}}\cdot k\rceil$
        \State $\textsf{RS}_{{\tt idx}}\gets \texttt{HNSW}({\tt idx},k')$
    \EndIf
    \State Filter $\textsf{RS}_{{\tt idx}}$ with $\mathcal{D}_\mathit{rids}(r)$ and append to $\textsf{RS}$
\EndFor

\For{each $v\in \mathcal{U}^*(r)$} 
    \State Compute $\textsf{dist}(v,\mathbf{x})$ and add $v$ to $\textsf{RS}$
\EndFor
\State Sort $\textsf{RS}$ by distances in descending order

\State \Return Top-$k$ data in $\textsf{RS}$

\end{algorithmic}
\end{algorithm}

Algorithm~\ref{alg:topk-execution} gives the independent (non-coordinated)
execution baseline referenced in \S\ref{sec:coordinated-search}.
Given a query $q=(\mathbf{x},r)$ and a query plan $\textsf{QP}(r)=(\mathcal{I}(r),\mathcal{U}(r))$, the baseline execution strategy searches every selected component independently.
Leftover vectors in $\mathcal{U}(r)$ are stored by block and scanned linearly.
For each HNSW index ${\tt idx}\in\mathcal{I}(r)$, the query uses standard top-$k$ search when ${\tt idx}$ is pure for $r$.
If ${\tt idx}$ is impure, the query inflates the request size to $k'=\lceil\lambda^r_{{\tt idx}} k\rceil$, where $\lambda^r_{{\tt idx}}$ is defined in Definition~\ref{def:impurity-factor}, and filters the returned candidates against $\mathcal{D}(r)$.
The final answer is the top-$k$ closest authorized vectors after merging the filtered HNSW results and the linear-scan results.

\begin{example}
Consider the merged group in \textsf{Strategy~3} of
Figure~\ref{fig:indexing_example} that builds one HNSW index
${\tt idx}(\{r_1,r_2\})$ over
$N^{\mathit{ex}}(\{r_1\})\cup N^{\mathit{ex}}(\{r_1,r_2\})$, each of size
$10{,}000$. For a query $q=(\mathbf{x},r_2)$ only
$N^{\mathit{ex}}(\{r_1,r_2\})$ is authorized inside this index, so
$\lambda^{r_2}_{{\tt idx}}=\lceil 20{,}000/10{,}000\rceil=2$. A top-$2$ query
therefore searches with $k'=4$; if HNSW returns
$[(u_1,0.03),(v_1,0.04),(u_2,0.06),(v_2,0.09)]$ with
$u_i\in N^{\mathit{ex}}(\{r_1\})$ and $v_i\in N^{\mathit{ex}}(\{r_1,r_2\})$,
filtering keeps $[(v_1,0.04),(v_2,0.09)]$ as the authorized top-$2$ from this
component, to be merged with the results of any other components in
$\textsf{QP}(r_2)$.
\end{example}

\subsection{\texorpdfstring{Accuracy Analysis of Inflated-$k$ Search}{Accuracy Analysis of Inflated-k Search}}
\label{sec:inflated-k-accuracy}

We formalize when inflated-$k$ search over an impure index returns the true authorized top-$k$.
For a query $q=(\mathbf{x},r)$ on an impure index ${\tt idx}$, the goal is to retrieve the $k$ closest vectors in $\mathcal{D}_r({\tt idx})$ under $\textsf{dist}(\cdot,\cdot)$.
Because unauthorized vectors may appear before authorized ones in the unfiltered ranking, the $k$-th authorized vector can occur after position $k$.

\begin{definition}[Authorized Position Threshold]
Let $u_1,u_2,\ldots$ denote all vectors in $\mathcal{D}({\tt idx})$ sorted by increasing $\textsf{dist}(\mathbf{x},u_i)$.
Define
\[
\tau_r = \min\Bigl\{t \mid \bigl|\{u_1,\ldots,u_t\}\cap \mathcal{D}_r({\tt idx})\bigr| \ge k \Bigr\},
\]
i.e., the position in the full unfiltered ranking where the $k$-th authorized vector appears.
\end{definition}

Retrieving at least the first $\tau_r$ true unfiltered neighbors is therefore sufficient to include all $k$ authorized nearest neighbors.

\begin{assumption}[HNSW Top-$t$ Success Condition]
\label{assump:hnsw-top-t}
For any index ${\tt idx}$ and any integer $t \ge 1$, there exists a non-decreasing function $f(t)$ such that running HNSW on ${\tt idx}$ with parameter $\textsf{efs}\ge f(t)$ returns the true top-$t$ nearest neighbors by distance with probability at least $1-\delta$.
\end{assumption}

This assumption abstracts the empirical behavior that larger $\textsf{efs}$ values yield higher HNSW recall.
In our implementation, $f(t)=\alpha t$ with a small constant $\alpha$ (e.g., $5\le\alpha\le10$) satisfies the condition with high probability.

\begin{theorem}[Accuracy of Inflated-$k$ Search]
\label{thm:coordinated-inflated-k-accuracy}
\label{thm:inflated-k-accuracy}
\label{thm:inflated-k-sufficiency}
Fix a query $q=(\mathbf{x},r)$ on an impure index ${\tt idx}$.
Let the search retrieve the true top-$k'$ unfiltered neighbors of $\mathbf{x}$ from $\mathcal{D}({\tt idx})$ with probability at least $1-\delta$, for some $k'\ge\tau_r$.
Then, after filtering unauthorized vectors from this result set, the remaining vectors contain exactly the true top-$k$ authorized neighbors of $\mathbf{x}$ in $\mathcal{D}_r({\tt idx})$, with probability at least $1-\delta$.
\end{theorem}

\begin{proof}
By definition of $\tau_r$, the $k$-th authorized neighbor $a_k\in\mathcal{D}_r({\tt idx})$ appears within the first $\tau_r$ elements of the true unfiltered ranking.
If the algorithm returns the true top-$k'$ unfiltered neighbors with $k'\ge\tau_r$, then all $u_1,\ldots,u_{\tau_r}$ are included.
Filtering unauthorized vectors from this set retains all authorized neighbors up to $a_k$.
No authorized vector after position $\tau_r$ can precede $a_k$ in distance.
Hence, the first $k$ vectors after filtering are exactly the true top-$k$ authorized neighbors, with probability at least $1-\delta$.
\end{proof}

\begin{corollary}[Inflation Rule]
\label{cor:operational-inflation}
Under Assumption~\ref{assump:hnsw-top-t}, Theorem~\ref{thm:inflated-k-sufficiency} holds if the search parameters satisfy
$k'\ge\tau_r$ and $\textsf{efs}'\ge f(k')$.
\end{corollary}

In practice, $\tau_r$ is unknown.
Coordinated search estimates an inflation factor $\lambda^r_{{\tt idx}}\ge1$ from the observed relation between the local impure-index ranking and the global result heap $\mathsf{RS}$.
Setting
\[
k' = \lceil \lambda^r_{{\tt idx}} \cdot k \rceil,
\qquad
\textsf{efs}' = \lceil \lambda^r_{{\tt idx}} \cdot \textsf{efs} \rceil
\]
grows the target size and the HNSW beam together.
If the chosen $\lambda^r_{{\tt idx}}$ is large enough that $k'\ge\tau_r$, Corollary~\ref{cor:operational-inflation} guarantees that filtering the returned set yields the exact authorized top-$k$ with probability at least $1-\delta$.

\subsection{Continued Base-Layer Search on Impure Indices}
\label{app:continued-hnsw}

Algorithm~\ref{alg:L0-hnsw} gives the base-layer (layer-$0$) HNSW traversal
used inside Algorithm~\ref{alg:coordinated-search} when an impure index must
be searched under the global distance bound. The traversal carries the current
global $k$-th distance $d_k^{(g)}$ from $\mathsf{RS}$ as an admission
threshold: a neighbor is pushed onto the candidate heap only if it can still
enter the global top-$k$, and the loop terminates as soon as the best
remaining candidate is farther than $d_k^{(g)}$. Authorization is checked
before a vector is admitted to the local result heap, so unauthorized vectors
steer the traversal but never reach $\mathsf{RS}$.

\begin{algorithm}[h!]
\caption{\textbf{Continual L0 Search in HNSW}}
\label{alg:L0-hnsw}
\begin{algorithmic}[1]

\Require Query vector $\mathbf{x}$, 
a max heap for global top-k results $\textsf{RS}$ before searching on the impure nodes, an impure index ${\tt idx}_{\mathit{ip}}$ of role $r$, a list of authorized data ids $\mathcal{D}_\mathit{rids}(r)$ of role $r$, 
current entry point $\mathit{ep}$, distance $\mathit{d\_ep}$, candidate min-heap $C$, , counters $\mathit|\textsf{RS}_l|$, $\mathit{ef}$, $k$, $\mathit{desired\_ef}$, thresholds $\mathit{worst\_dist}$, $d_\mathit{max}^{(l)}$, and optional filter predicate $\mathit{isAllowed}(\cdot)$

\Ensure Top-$k$ authorized results for query vector $\mathbf{x}$ with role $r$.

\State Traverse on ${\tt idx}_{\mathit{ip}}$ to reach the bottom layer of ${\tt idx}_{\mathit{ip}}$
\State $(d_v, v)\gets \textsc{GetEntryPoint}()$, $C\gets []$\Comment{Candidate list}

    \State $C.\textsf{push}(\mathit{d_v}, v)$, $\mathit{ef} \gets 1$, $d_k^{g}\gets \textsf{RS}.\texttt{max}()$, $v.\texttt{visited}\gets \texttt{True}$
    
    \If{$v\in \mathcal{D}_\mathit{rids}(r) \land d_v <d_k^{(g)}$}~$\textsf{RS}_l.\textsf{push}(d_v, v)$, $d_\mathit{max}^{(l)} \gets d_v$
    \EndIf

\State $\textsf{RS}_l$, $C\gets \textsf{BaseLayerSearch}(\mathbf{x}, C, k, 1, \mathit{ef}_\mathit{max}, \mathcal{D}_\mathit{rids}(r))$ \Comment{Algorithm~\ref{alg:base_layer_search_hnsw}}

\State $\textsf{RS}\gets \textsc{MergeRS}(\textsf{RS}_l, \textsf{RS})$

\If{\texttt{STOP\_FLAG}}~\textbf{return } $\textsf{RS}$
\EndIf

\State $\textsf{RS}_l$, $C\gets \textsf{BaseLayerSearch}(\mathbf{x}, C, k, \mathit{ef}_\mathit{default}, \mathit{ef}_\mathit{max}, \mathcal{D}_\mathit{rids}(r))$ \Comment{Algorithm~\ref{alg:base_layer_search_hnsw}}

\State $\textsf{RS}\gets \textsc{MergeRS}(\textsf{RS}_l, \textsf{RS})$

\State \textbf{return } $\textsf{RS}$

\end{algorithmic}
\end{algorithm}

\begin{algorithm}[h!]
    \caption{\textbf{BaseLayerSearch for Algorithm~\ref{alg:L0-hnsw}}}
    \label{alg:base_layer_search_hnsw}
    \begin{algorithmic}[1]
\Require Query vector $\mathbf{x}$, candidate min-heap $C$, top-$k$ parameter $k$,
current expansion factor $\mathit{ef}$, maximum expansion factor $\mathit{ef}_\mathit{max}$,
and authorized data ids $\mathcal{D}_\mathit{rids}(r)$ for role $r$.
\State $\textsf{RS}_{l}\gets []$
\While{$C$ is not empty}
    \State $(d_v, v) \gets C.\textsf{top}()$, $C.\textsf{pop}()$
    
    \If{$d_v >d_k^{(g)} \lor (d_v > d_\mathit{max}^{(l)} \land |\textsf{RS}_l| \geq k)$}
        \State \textbf{break}
    \EndIf
    
    \For{each neighbor $u$ of $v$ in base layer}
        \If{$u$.$\texttt{visited}$}~\textbf{Continue}\EndIf
        \If{$\mathit{ef} \geq \mathit{ef}_\mathit{max}$}~\textbf{break}\EndIf
        \State $u$.$\texttt{visited}\gets\texttt{True}$, $\mathit{ef} \gets \mathit{ef} + 1$, $d_u \gets \textsf{dist}(\mathbf{x}, u)$
        
        \If{$d_u >d_k^{(g)} \lor (d_u > d_\mathit{max}^{(l)} \land \mathit|\textsf{RS}_l| \geq k)$}
            \State \textbf{continue}
        \EndIf
        
        \State $C.\textsf{push}(d_u, u)$

        \If{$u\notin\mathcal{D}(r)$}~\textbf{continue}\EndIf
        \If{$d_u < d_\mathit{max}^{(l)}$}~$\textsf{RS}_l.\textsf{push}(d_u, u)$
        \ElsIf{$d_u <d_k^{(g)}\land \mathit|\textsf{RS}_l| < k$} 
        \State $\textsf{RS}_l.\textsf{push}(d, u)$, $d_\mathit{max}^{(l)} \gets d_u$
        \EndIf
    \EndFor
\EndWhile
\State \textbf{return} $\textsf{RS}_l$, $C$

\end{algorithmic}
\end{algorithm}

\section{Complementary Evaluations}
\label{app:complementary-results}

\subsection{Supplementary Results for Index Creation Evaluation}
\label{app:index-creation-complementary-results}

\begin{figure*}[!t]
  \centering
  \begin{subfigure}[t]{0.24\textwidth}
      \centering
      \includegraphics[width=\linewidth]{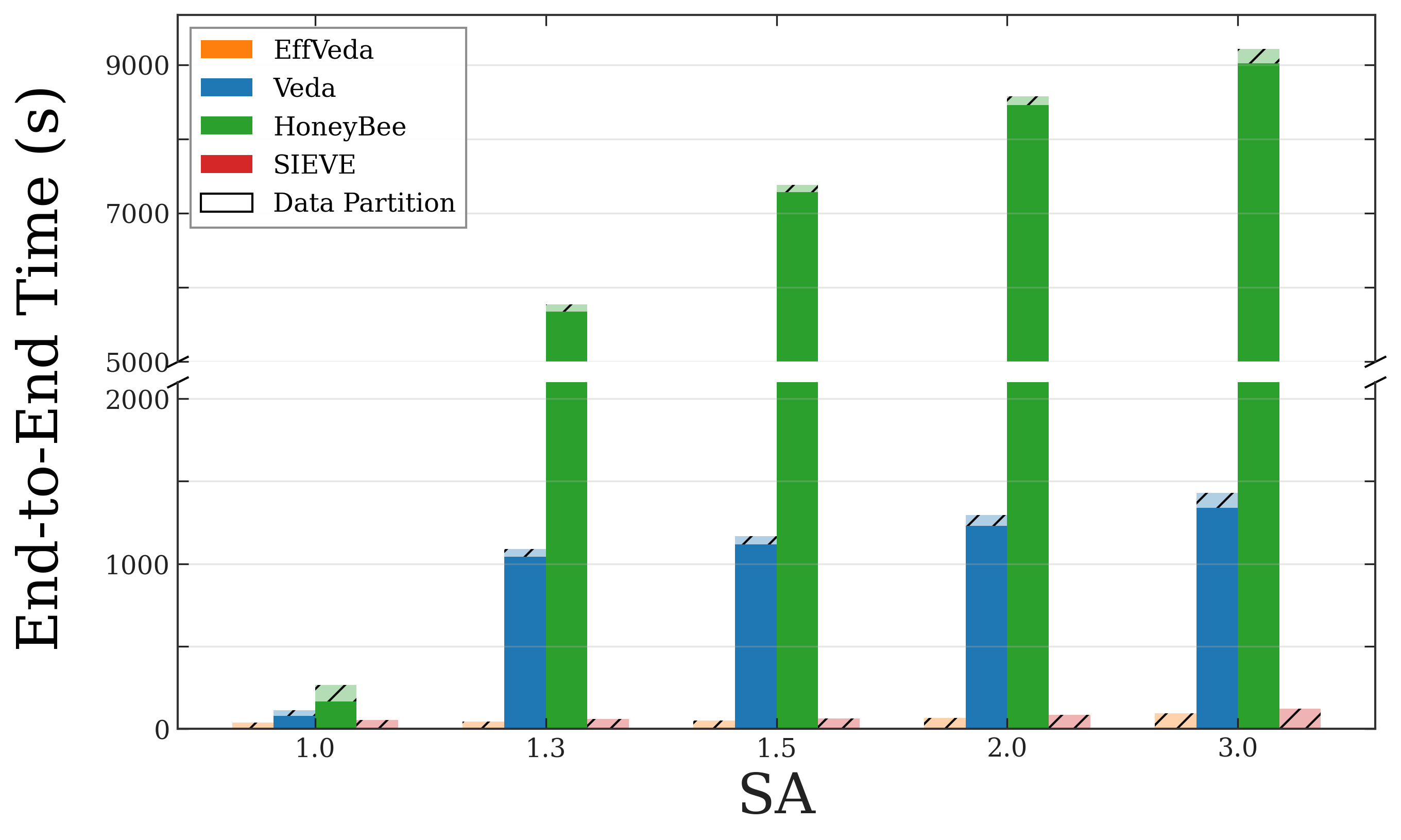}
      \caption{End-to-end build time vs.\ SA.}
      \label{fig:build_time}
  \end{subfigure}
  \hfill
  \begin{subfigure}[t]{0.24\textwidth}
      \centering
      \includegraphics[width=\linewidth]{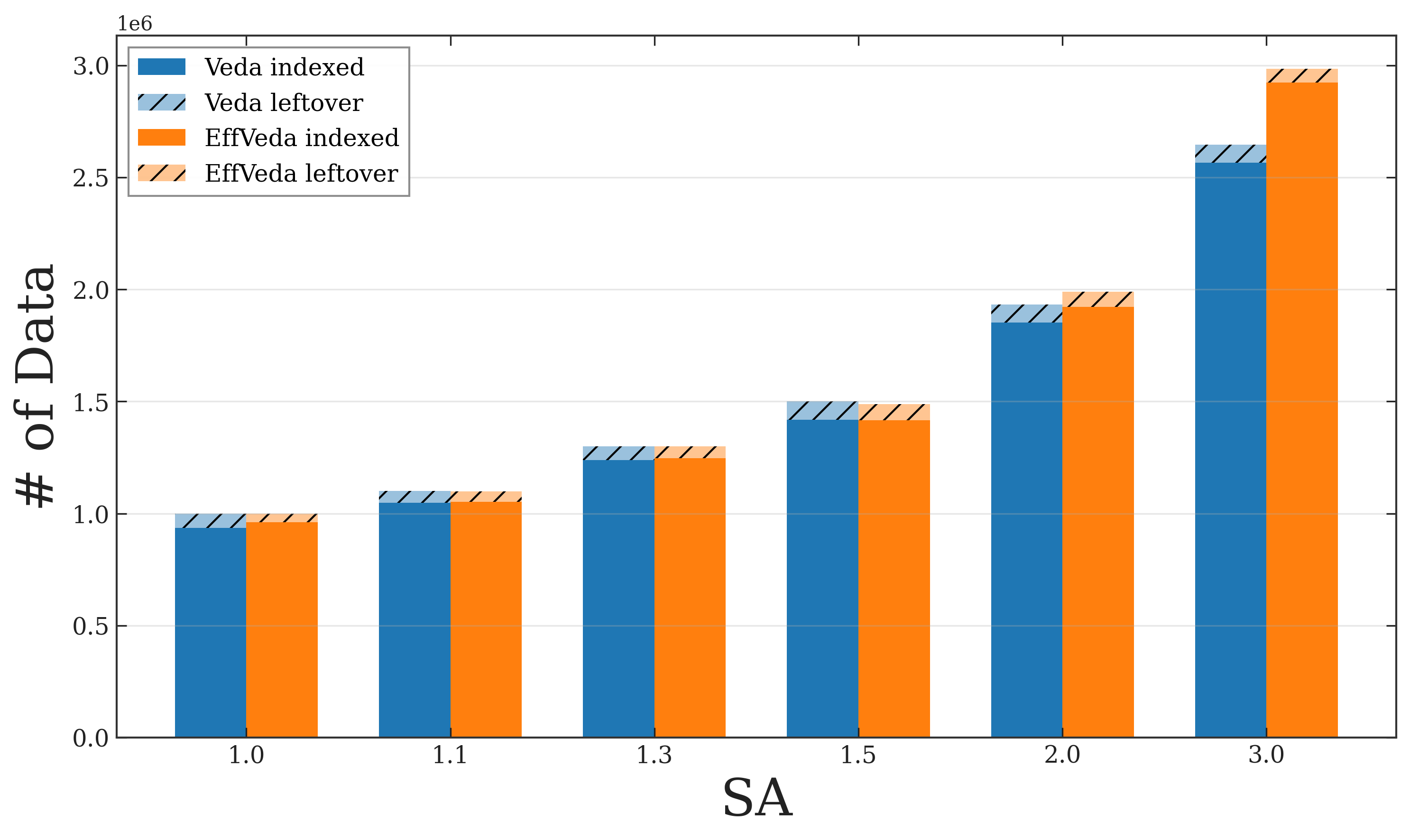}
      \caption{\# Indexed vs.\ \# leftover.}
      \label{fig:number-of-indexed-leftover-data}
  \end{subfigure}
  \hfill
  \begin{subfigure}[t]{0.24\textwidth}
      \centering
      \includegraphics[width=\linewidth]{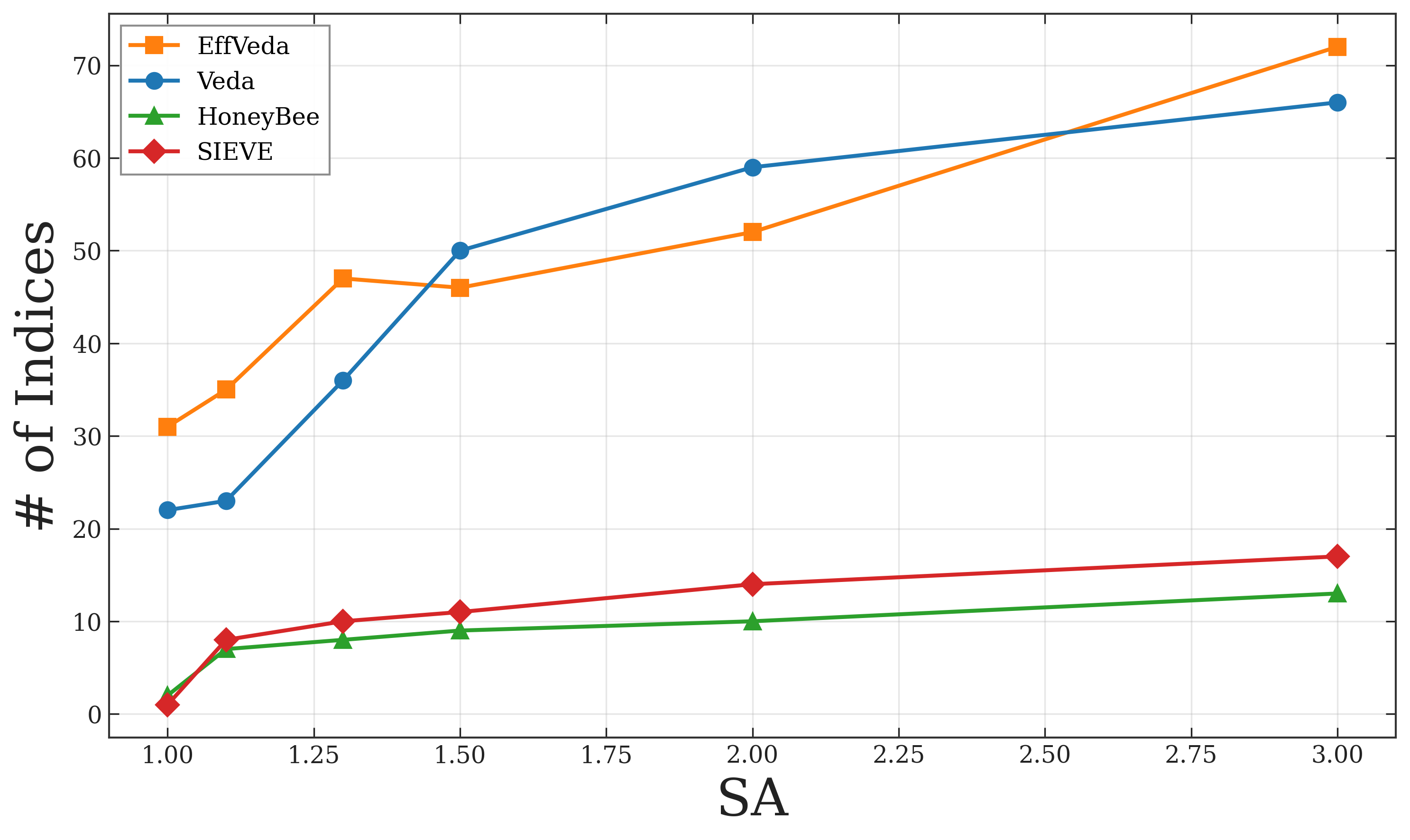}
      \caption{\# of indices vs.\ SA.}
      \label{fig:num_indices_built_vs_sa}
  \end{subfigure}
  \hfill
  \begin{subfigure}[t]{0.24\textwidth}
      \centering
      \includegraphics[width=\linewidth]{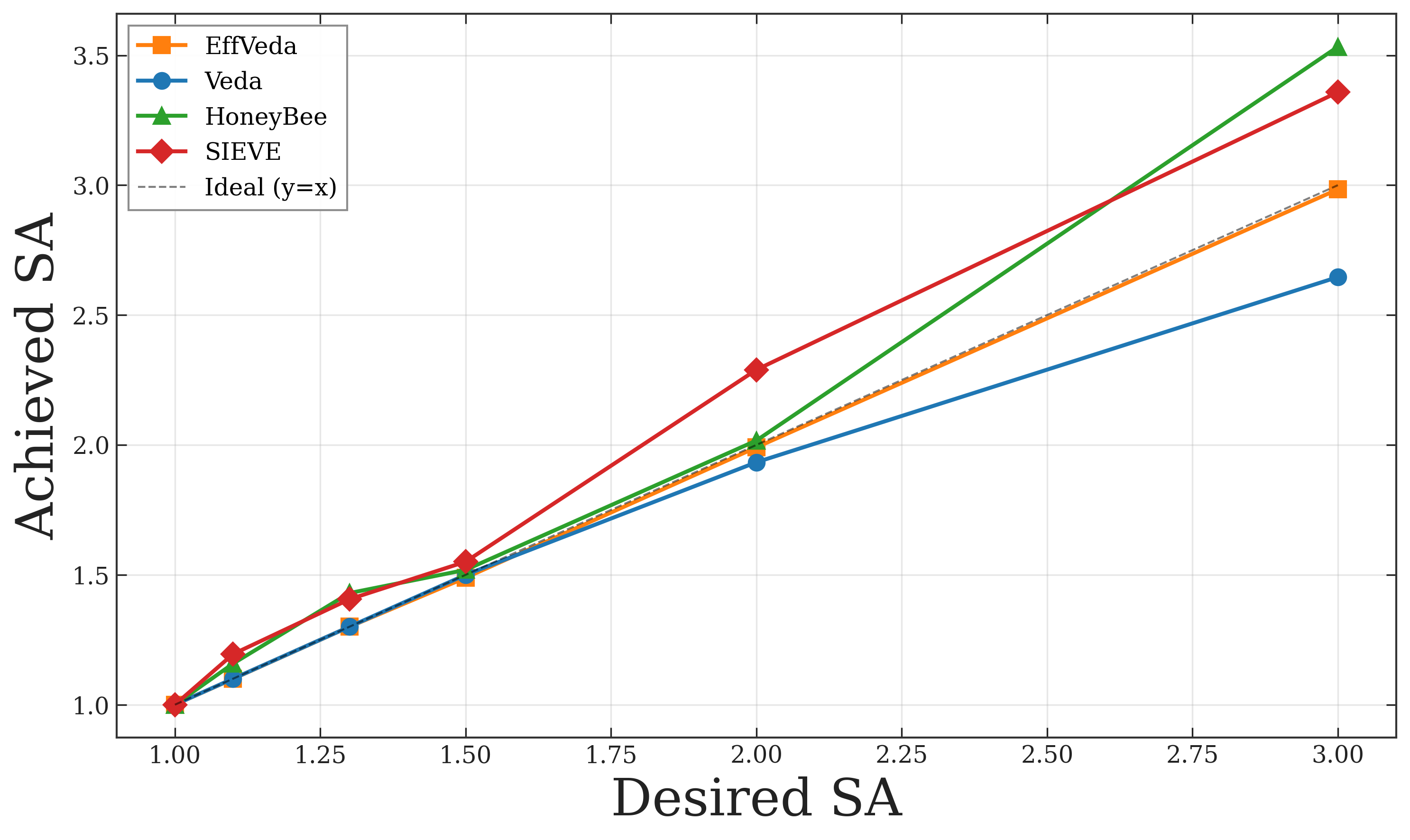}
      \caption{Achieved SA.}
      \label{fig:input_vs_achieved_sa}
  \end{subfigure}
  \caption{Index creation evaluation.}
  \label{fig:index_creation_vs_sa}
\end{figure*}

This subsection provides the detailed results for
Exp~\ref{exp:index-creation-time}--\ref{exp:desired-vs-achieved-sa}. All
measurements are reported on SIFT-1M while varying the target
\textsf{SA} over $\{1.0,1.1,1.3,1.5,2.0,3.0\}$. Figure~\ref{fig:index_creation_vs_sa}
summarizes construction performance across all data partitioning approaches, including end-to-end build time and the data partitioning time. Figure~\ref{fig:number-of-indexed-leftover-data} shows the amount
of indexed and leftover data; Figure~\ref{fig:num_indices_built_vs_sa} reports the number of HNSW indices built as \textsf{SA} varies; and Figure~\ref{fig:input_vs_achieved_sa} shows the achieved \textsf{SA} given the target \textsf{SA}. Table~\ref{tab:construction-time-ac-sift}
complements the figure with the concrete data-partition and total-build
times for HoneyBee, SIEVE, \name, and \effname.

\begin{table*}[t]
  \centering
  \caption{Construction time on SIFT-1M (s).}
  \label{tab:construction-time-ac-sift}
  \begin{tabular}{llrrrrrr}
  \toprule
  \textbf{Method} & \textbf{Metric} & SA=1.0 & SA=1.1 & SA=1.3 & SA=1.5 & SA=2.0 & SA=3.0 \\
  \midrule

  \multirow{2}{*}{HoneyBee}
    & Data Partition Time (s) & 164.849 & 4664.006 & 5671.606 & 7280.113 & 8453.377 & 9017.423 \\
    & Total Build Time (s)    & 265.848 & 4747.681 & 5766.158 & 7378.101 & 8576.683 & 9209.638 \\
  \midrule
  \multirow{2}{*}{SIEVE}
    & Data Partition Time (s) & \multicolumn{6}{c}{negligible (from past queries)} \\
    & Total Build Time (s)   & 54.071 & 56.177 & 60.396 & 63.582 & 85.768 & 120.792 \\
  \midrule
  \multirow{2}{*}{\name}
    & Data Partition Time (s) & 77.412 & 667.110 & 1042.916 & 1118.612 & 1229.079 & 1338.673 \\
    & Total Build Time (s)    & 111.214 & 708.913 & 1089.213 & 1168.700 & 1296.613 & 1428.290 \\
    \midrule
    \multirow{2}{*}{\effname}
      & Data Partition Time (s) & 1.983 & 2.161 & 2.315 & 2.539 & 3.053 & 3.794 \\
      & Total Build Time (s)    & 36.131 & 38.832 & 43.264 & 49.165 & 64.531 & 94.709 \\
  \bottomrule
  \end{tabular}
\end{table*}

\subsection{Ablation Tables for Exps~\ref{exp:n_indcies_per_query}
            and~\ref{exp:coordinated-search}}
\label{app:ablation-tables}

Tables~\ref{tab:avg-hnsw-indices-per-query}--\ref{tab:ef-budget-savings}
report the per-\textsf{SA} numbers behind the ablations in
\S\ref{sec:evaluations}. Table~\ref{tab:avg-hnsw-indices-per-query} gives the
average number of HNSW indices touched per query
(Exp~\ref{exp:n_indcies_per_query}). Table~\ref{tab:phase2-skip-rate} gives the
fraction of impure-index visits for which coordinated search skips the
inflated phase~2 probe, and Table~\ref{tab:ef-budget-savings} gives the
resulting average reduction in \textsf{efs} relative to always inflating
(Exp~\ref{exp:coordinated-search}). All numbers are on SIFT-1M.

\begin{table}[h]
  \centering
  \caption{Avg.\ \# of HNSW indices per query (Exp~\ref{exp:n_indcies_per_query}).}
  \label{tab:avg-hnsw-indices-per-query}
  \begin{tabular}{@{}lcccccc@{}}
  \toprule
  \textbf{SA} & \textbf{1.0} & \textbf{1.1} & \textbf{1.3} & \textbf{1.5} & \textbf{2.0} & \textbf{3.0} \\
  \midrule
  \name    & 3.53 & 3.69 & 3.01 & 5.01 & 5.03 & 5.14 \\
  \effname & 6.55 & 5.14 & 5.11 & 4.41 & 3.91 & 3.79 \\
  \bottomrule
  \end{tabular}
\end{table}

\begin{table}[h]
  \centering
  \caption{Phase-2 skip rate on impure nodes (Exp~\ref{exp:coordinated-search}).}
  \label{tab:phase2-skip-rate}
  \begin{tabular}{@{}lcccccc@{}}
  \toprule
  \textbf{SA} & \textbf{1.0} & \textbf{1.1} & \textbf{1.3} & \textbf{1.5} & \textbf{2.0} & \textbf{3.0} \\
  \midrule
  \name    & 89.81\% & 67.95\% & 74.89\% & 90.00\% & 74.07\% & 100.00\% \\
  \effname & 
              91.56 \% & 90.07 \% & 90.45 \% & 91.67 \% & 91.80 \% & 95.61\% \\
  \bottomrule
  \end{tabular}
\end{table}

\begin{table}[h]
  \centering
  \caption{Avg.\ \textsf{efs} savings on impure nodes (Exp~\ref{exp:coordinated-search}).}
  \label{tab:ef-budget-savings}
  \begin{tabular}{@{}lcccccc@{}}
  \toprule
  \textbf{SA} & \textbf{1.0} & \textbf{1.1} & \textbf{1.3} & \textbf{1.5} & \textbf{2.0} & \textbf{3.0} \\
  \midrule
  \name    & 16.36\% & 23.67\% & 19.27\% & 13.63\% & 13.65\% &  8.57\% \\
  \effname & 22.79 \%& 13.55 \%& 9.55 \%& 7.95 \%& 8.57 \%& 8.14 \% \\ 
  \bottomrule
  \end{tabular}
\end{table}

\subsection{Supplementary Experiment: Effectiveness of Leftovers}
\label{app:leftover-effectiveness}

\experiment{Effectiveness of leftovers}{exp:leftover-effectiveness}
Figure~\ref{fig:leftover-effectiveness} evaluates leftovers for our
strategies, \name and \effname, on SIFT-1M. Across the \textsf{SA} sweep,
\effname leaves about $39{,}000$--$70{,}000$ vectors as leftovers, while
\name leaves about $53{,}000$--$82{,}000$.
For each query, only the leftover blocks authorized for the searching
role are needed. If all leftover vectors were stored in one shared index,
the impurity would be high: the average leftover
impurity is $46.7$--$83.2$ for \effname and $74.4$--$112.0$ for \name. Such a
single leftover index would force large result and beam-width inflation before
filtering. Instead, \name and \effname store leftovers separately as exclusive
blocks. The selected leftover blocks for a query are pure for the searching
role and small enough for efficient linear scan. Moreover, coordinated search
processes them before impure nodes, and their exact candidates further tighten the
global top-$k$ bound to prune later searches on impure indices.

\begin{figure}[t]
    \centering
    \begin{subfigure}[t]{0.23\textwidth}
        \centering
        \includegraphics[width=\linewidth]{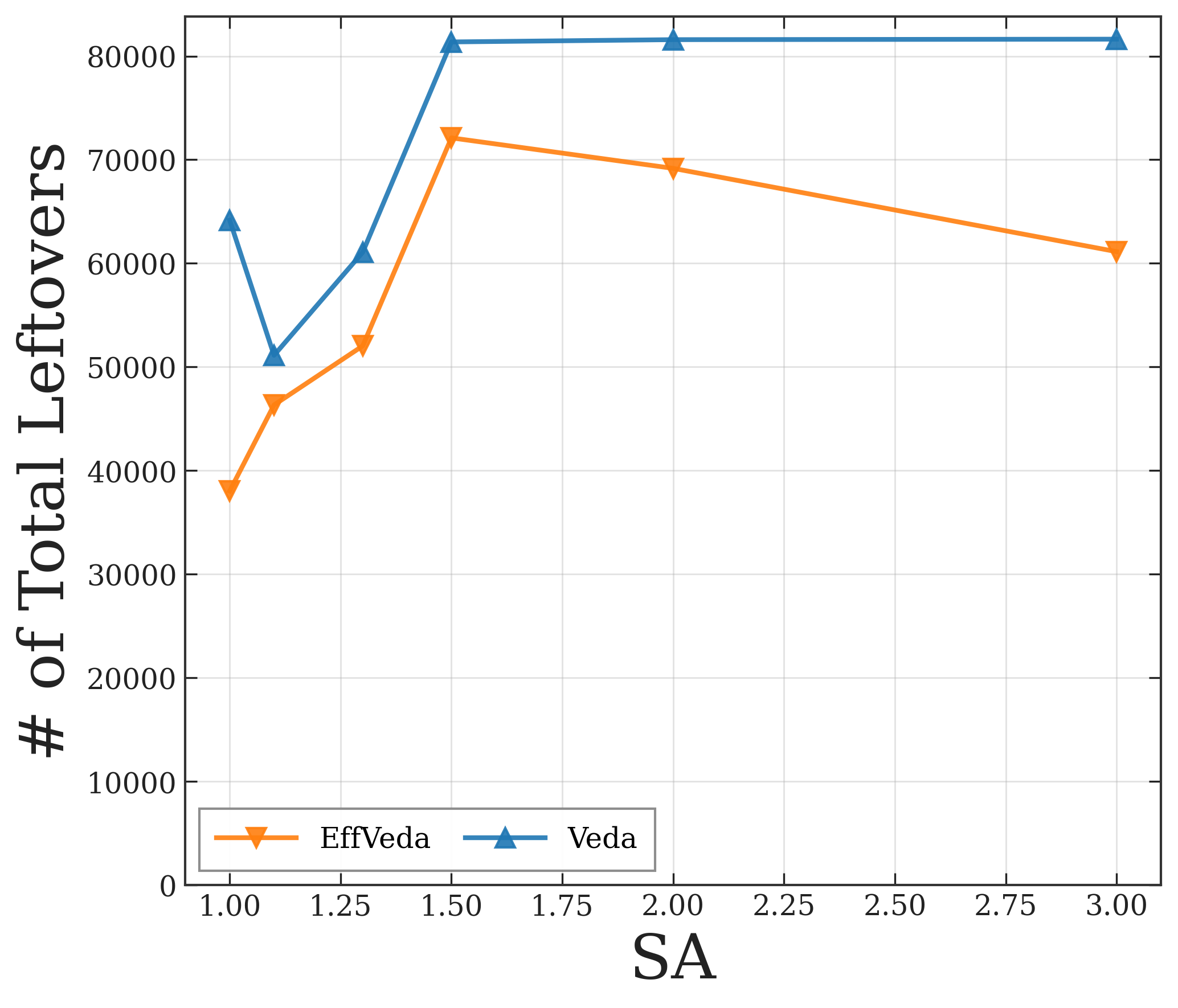}
        \caption{Total \# of leftover vectors.}
        \label{fig:total-leftover-data-varying-sa}
    \end{subfigure}
    \hfill
    \begin{subfigure}[t]{0.23\textwidth}
        \centering
        \includegraphics[width=\linewidth]{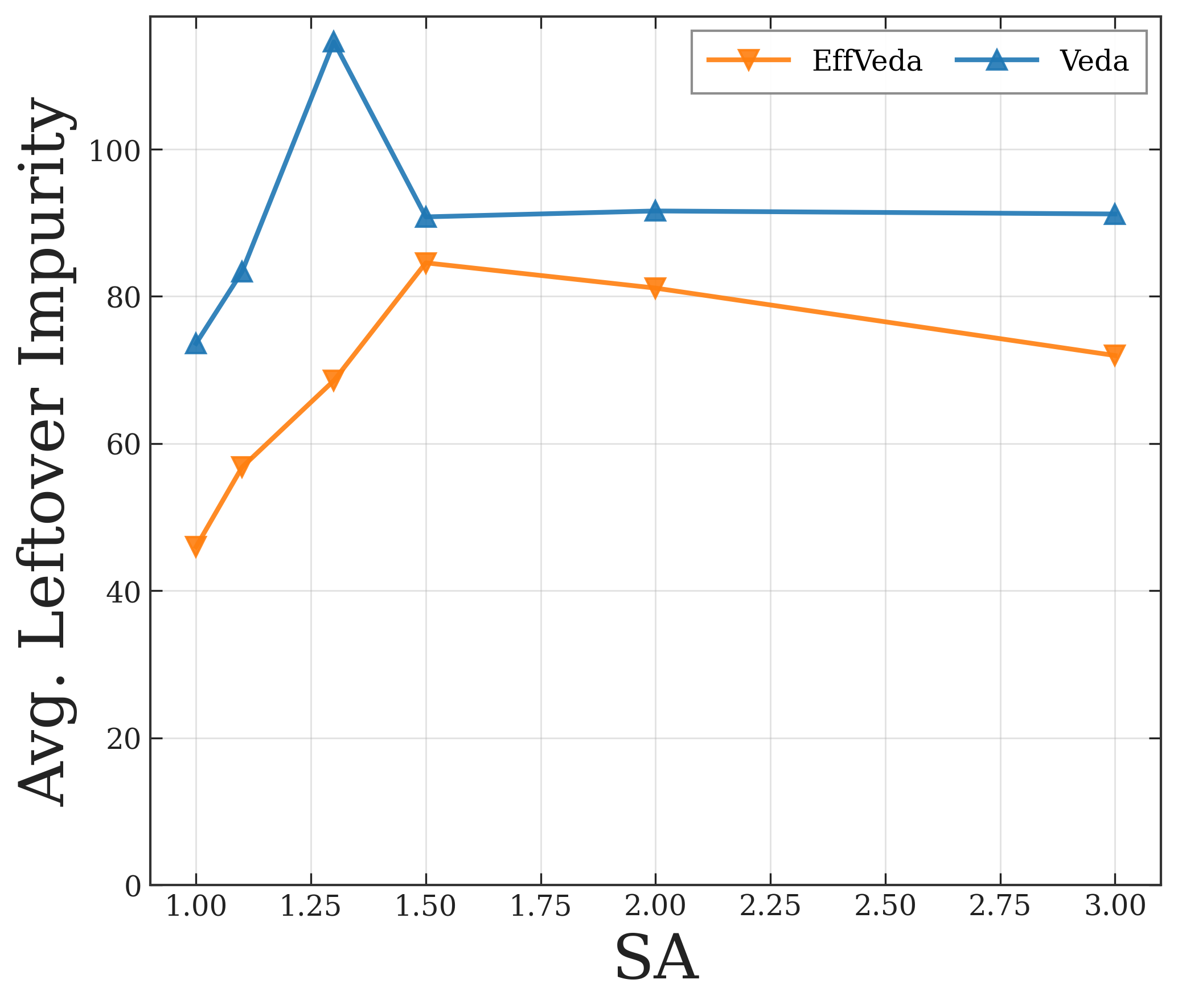}
        \caption{Avg. leftover impurity.}
        \label{fig:leftover-impurity}
    \end{subfigure}
    \caption{Effectiveness of leftovers on SIFT-1M.}
    \label{fig:leftover-effectiveness}
\end{figure}

\subsection{Supplementary Results for Exp~\ref{exp:indexing-threshold}: Varying Indexing Thresholds}
\label{app:vary-lambda-more-sa}

Tables~\ref{tab:vary-lambda-sift1m-qps-sa11}--\ref{tab:vary-lambda-sift1m-qps-sa15}
and Figure~\ref{fig:big-idx-num-vs-indexing-threshold-sift}
report QPS and the number of materialized indices as $\Lambda$ varies at
$\textsf{SA}\in\{1.1,1.3,1.5\}$.
Most strong results occur around $\Lambda{=}2{,}900$, but performance is not
sensitive to the exact threshold. 
Across \name and \effname, the worst QPS remains at least $77.96\%$ of the
best, and in most \textsf{SA} settings it is close to $90\%$. This is
because finalization absorbs many threshold-induced differences in the lattice.
The number of indices reduces as the indexing threshold increases, since fewer
lattice nodes are large enough to be materialized as HNSW indices as $\Lambda$ increases.

\begin{table}[h]
  \centering
  \caption{Indexing threshold $\Lambda$ vs.\ QPS on SIFT-1M, SA = 1.1.}
  \label{tab:vary-lambda-sift1m-qps-sa11}
  \begin{tabular}{@{}lccccc@{}}
  \toprule
  \textbf{$\Lambda$} & \textbf{1900} & \textbf{2400} & \textbf{2900} & \textbf{3400} & \textbf{3900} \\
  \midrule
  \name & 2025.08 & 1893.38 & 2080.03 & 1971.71 & 1973.65 \\
  \effname & 1295.56 & 1530.55 & 1617.66 & 1567.40 & 1661.37 \\
  \bottomrule
  \end{tabular}
\end{table}

\begin{table}[t]
  \centering
  \caption{Indexing threshold $\Lambda$ vs. QPS on SIFT-1M, SA = 1.3.}
  \label{tab:vary-lambda-sift1m-qps-sa13}
  \begin{tabular}{@{}lcccccc@{}}
  \toprule
  \textbf{$\Lambda$} & \textbf{1900} & \textbf{2400} & \textbf{2900} & \textbf{3400} & \textbf{3900} \\
  \midrule
  \name & 2536.74 & 2455.95 & 2583.94 & 2522.00 & 2814.27 \\
  \effname & 1773.61 & 1821.52 & 1904.54 & 2033.74 & 2001.90 \\
  \bottomrule
  \end{tabular}
  \end{table}

\begin{table}[t]
\centering
\caption{Indexing threshold $\Lambda$ vs. QPS on SIFT-1M, SA = 1.5.}
\label{tab:vary-lambda-sift1m-qps-sa15}
\begin{tabular}{@{}lcccccc@{}}
\toprule
\textbf{$\Lambda$} & \textbf{1900} & \textbf{2400} & \textbf{2900} & \textbf{3400} & \textbf{3900} \\
\midrule
\name & 1521.78 & 1658.08& 1658.91 & 1699.36 & 1694.52 \\
\effname & 1364.52 & 1534.84 & 1670.74 & 1735.53 & 1701.31 \\
\bottomrule
\end{tabular}
\end{table}

\begin{figure}[t]
    \centering
    \begin{subfigure}[t]{0.32\linewidth}
        \centering
        \includegraphics[width=\linewidth]{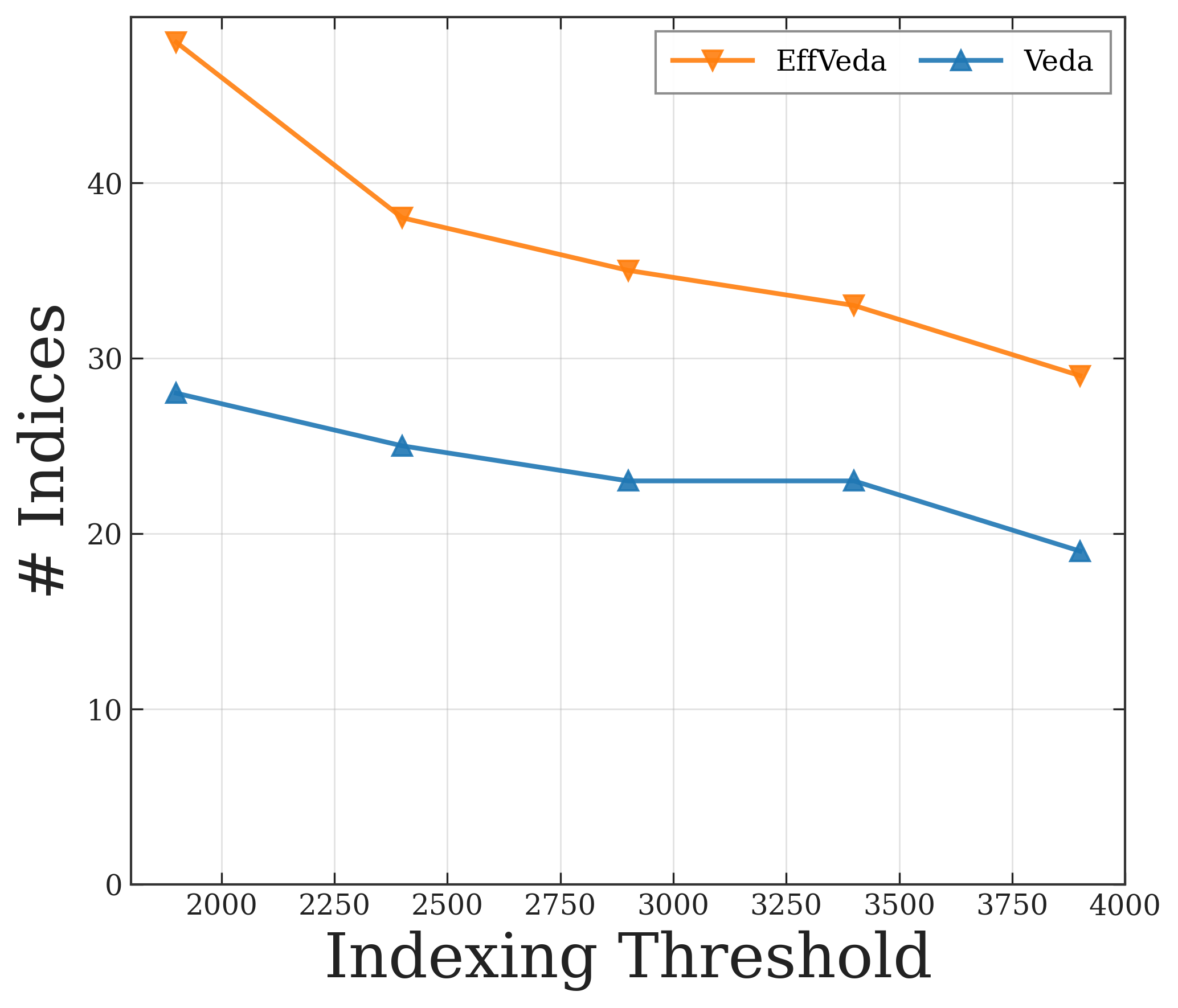}
        \caption{\textsf{SA} = 1.1.}
        \label{fig:big-idx-num-vs-indexing-threshold-sift-sa11}
    \end{subfigure}
    \hfill
    \begin{subfigure}[t]{0.32\linewidth}
        \centering
        \includegraphics[width=\linewidth]{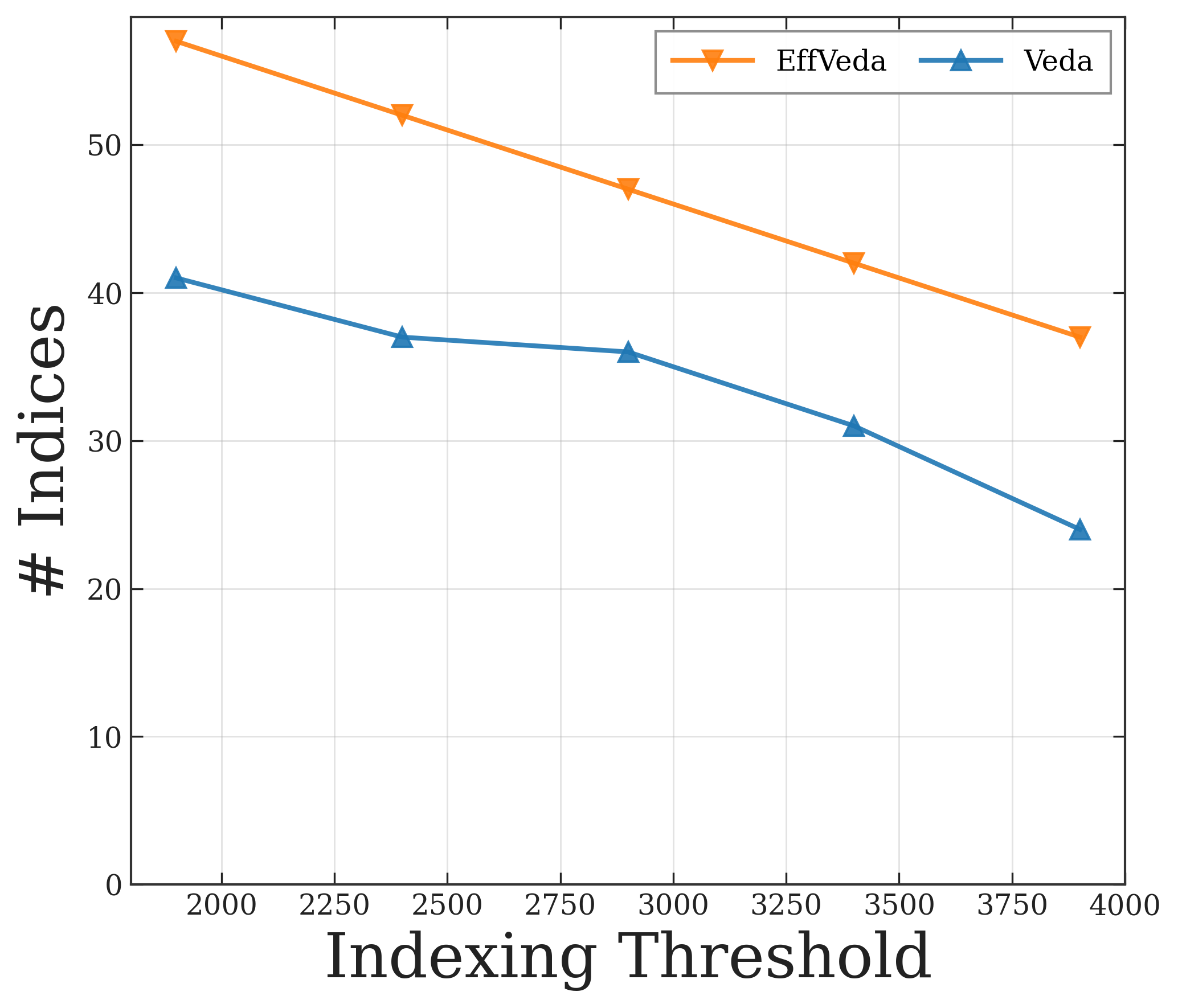}
        \caption{\textsf{SA} = 1.3.}
        \label{fig:big-idx-num-vs-indexing-threshold-sift-sa13}
    \end{subfigure}
    \hfill
    \begin{subfigure}[t]{0.32\linewidth}
        \centering
        \includegraphics[width=\linewidth]{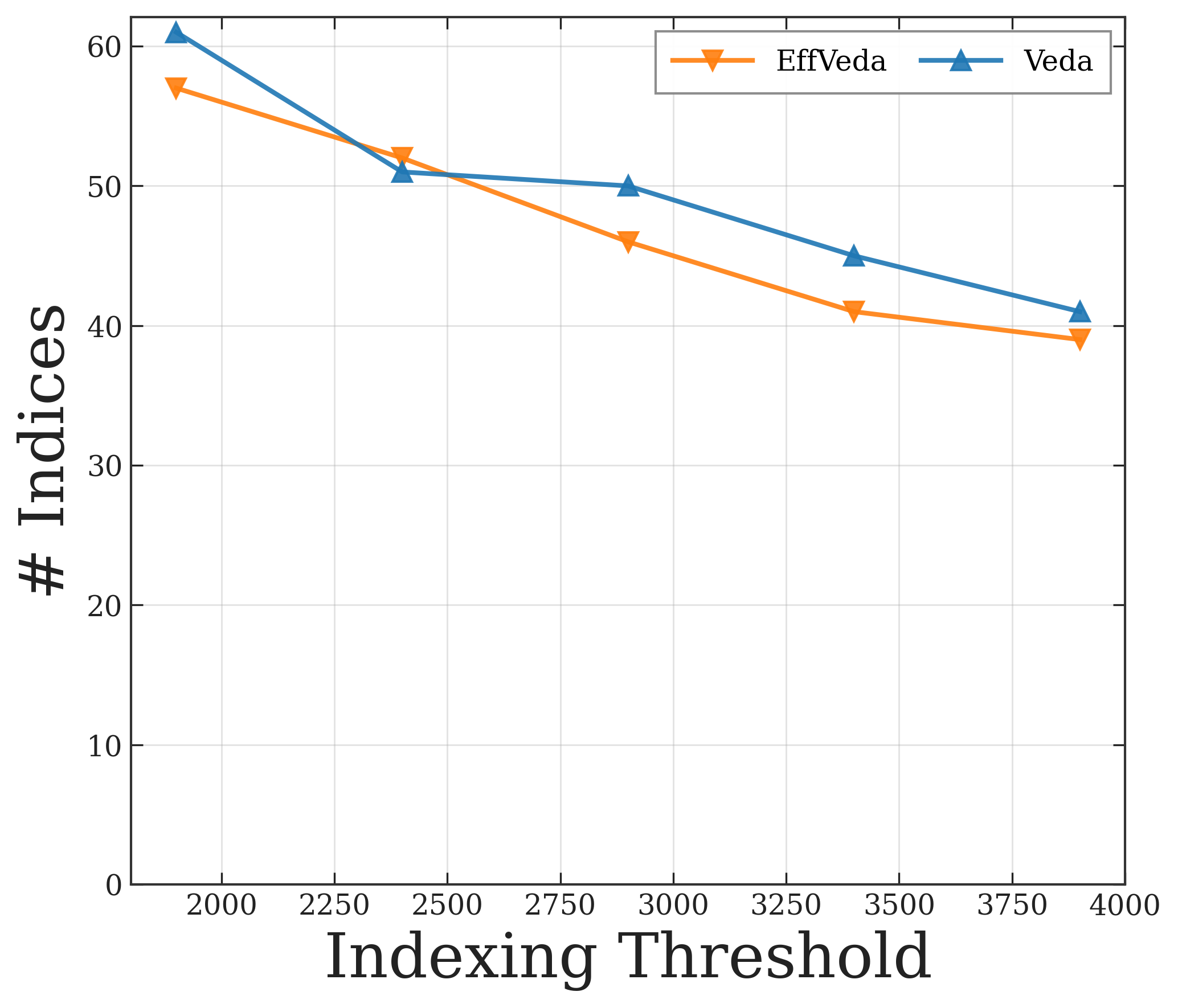}
        \caption{\textsf{SA} = 1.5.}
        \label{fig:big-idx-num-vs-indexing-threshold-sift-sa15}
    \end{subfigure}
    \caption{\# of indices vs.\ indexing threshold $\Lambda$ on SIFT-1M.}
    \label{fig:big-idx-num-vs-indexing-threshold-sift}
\end{figure}

\subsection{Supplementary Results for Exp~\ref{exp:efs}: Impacts of \textsf{efs}}
\label{app:efs-more-datasets}

We provide evaluation results in Figure~\ref{fig:qps_vs_ef_paper} and Figure~\ref{fig:qps_vs_ef_amzn} on the PAPER and AMZN datasets as a complement to Exp~\ref{exp:efs} in
\S\ref{sec:evaluations}. As on SIFT-1M, increasing \textsf{efs} lowers QPS for
all methods because each HNSW search expands a wider beam. The relative trends
are consistent across datasets. Oracle achieves the highest QPS on all datasets, and \name and \effname maintain high QPS over most of the evaluations.

\begin{figure*}[t]
    \centering
    \begin{minipage}{\linewidth}
        \centering
        \includegraphics[width=0.9\linewidth]{figures/legend.png}
    \end{minipage}
    \begin{subfigure}[t]{0.24\textwidth}
        \centering
        \includegraphics[width=\linewidth]{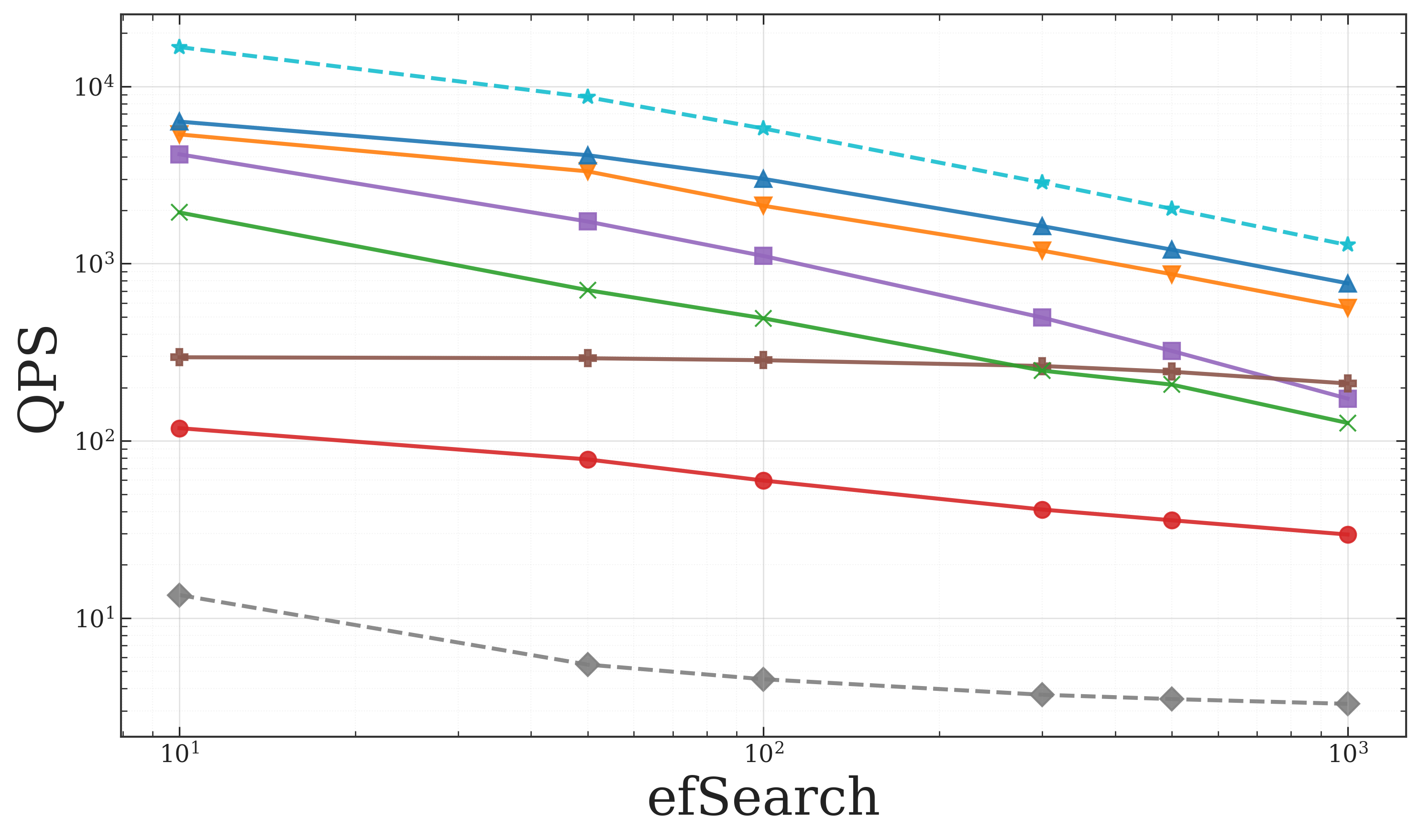}
        \caption{QPS vs.\ \textsf{efs}, PAPER.}
        \label{fig:qps_vs_ef_paper}
    \end{subfigure}
    \hfill
    \begin{subfigure}[t]{0.24\textwidth}
        \centering
        \includegraphics[width=\linewidth]{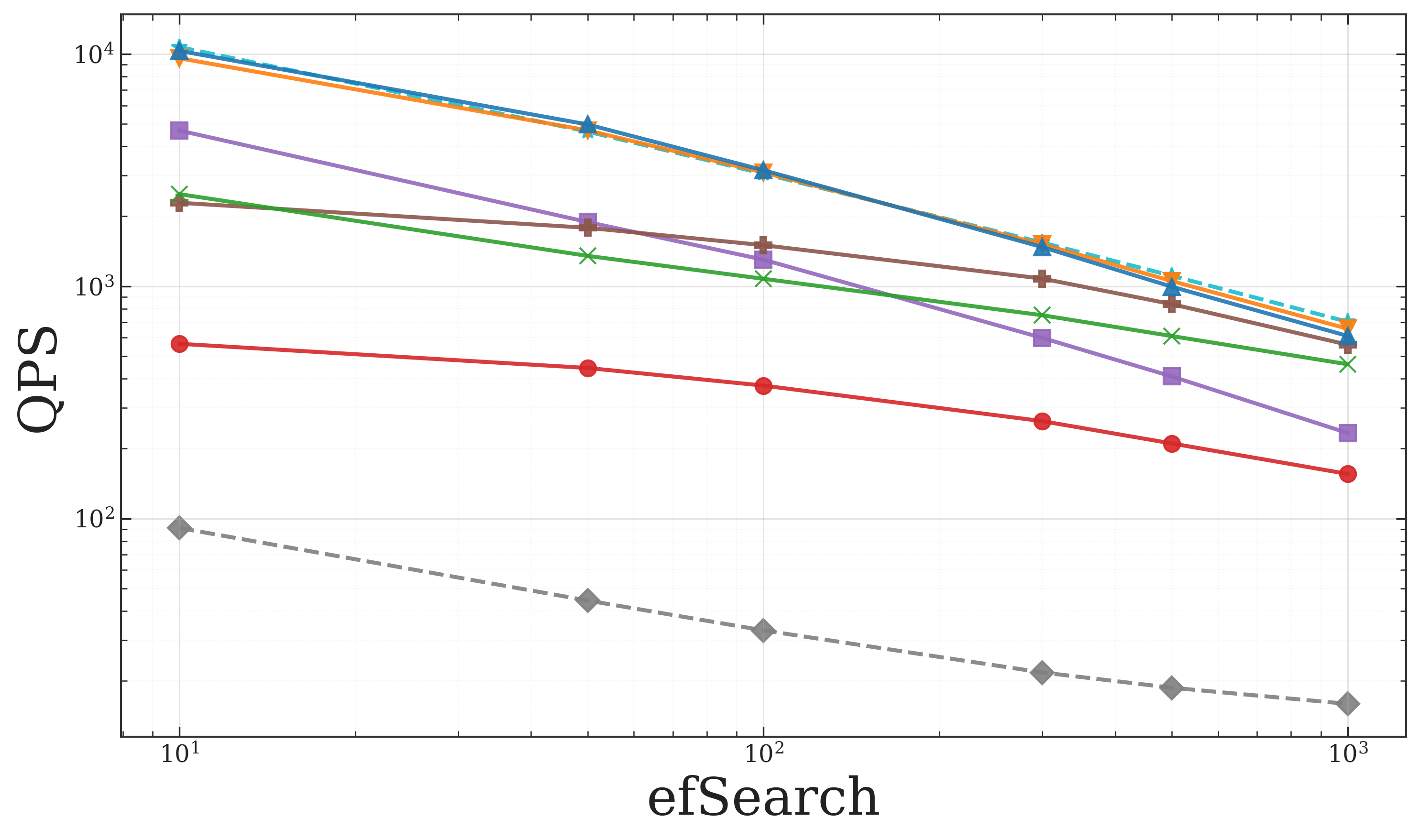}
        \caption{QPS vs.\ \textsf{efs}, AMZN.}
        \label{fig:qps_vs_ef_amzn}
    \end{subfigure}
    \hfill
    \begin{subfigure}[t]{0.24\textwidth}
        \centering
        \includegraphics[width=\linewidth]{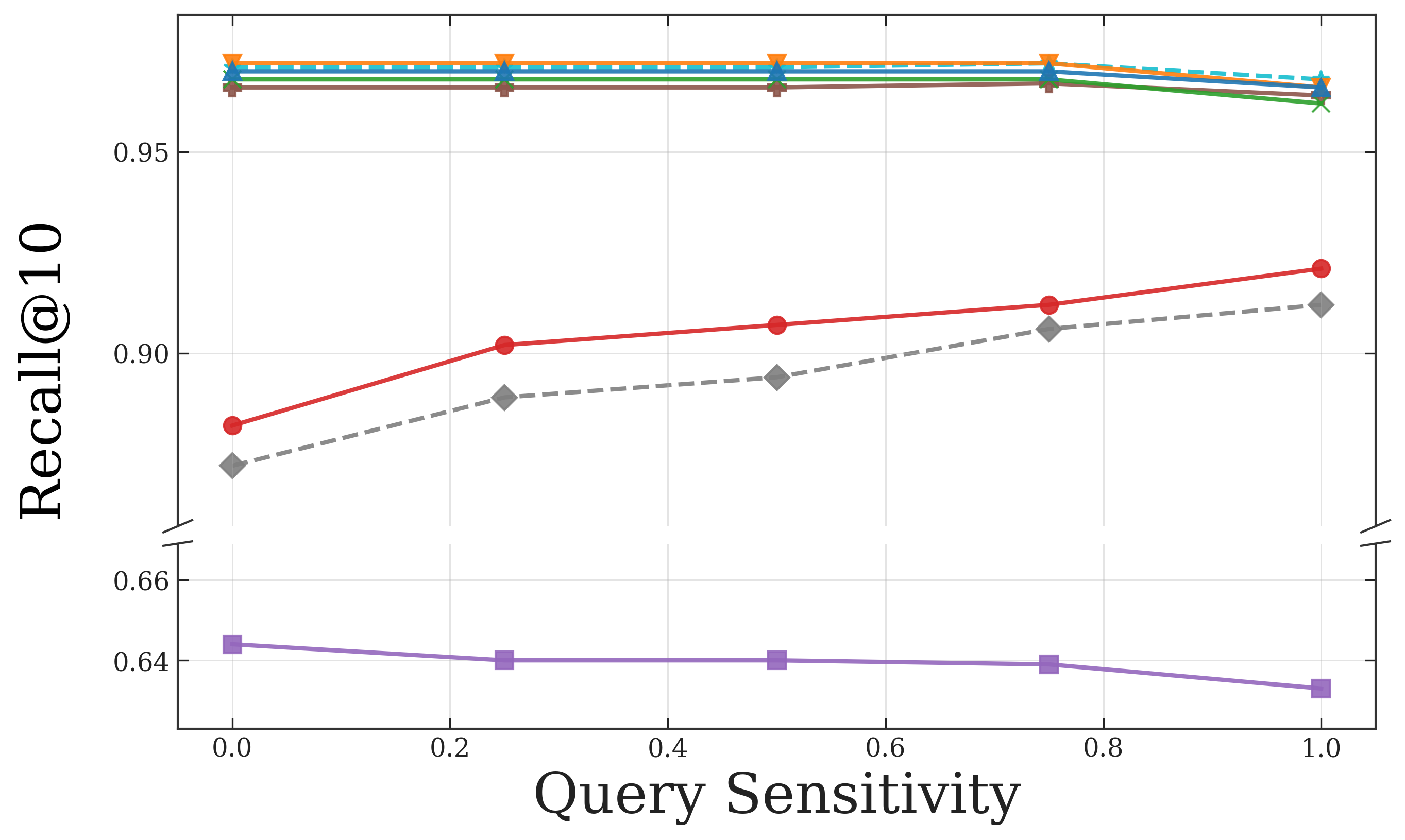}
        \caption{Sensitivity, PAPER.}
        \label{fig:recall_vs_sensitivity_paper}
    \end{subfigure}
    \hfill
    \begin{subfigure}[t]{0.24\textwidth}
        \centering
        \includegraphics[width=\linewidth]{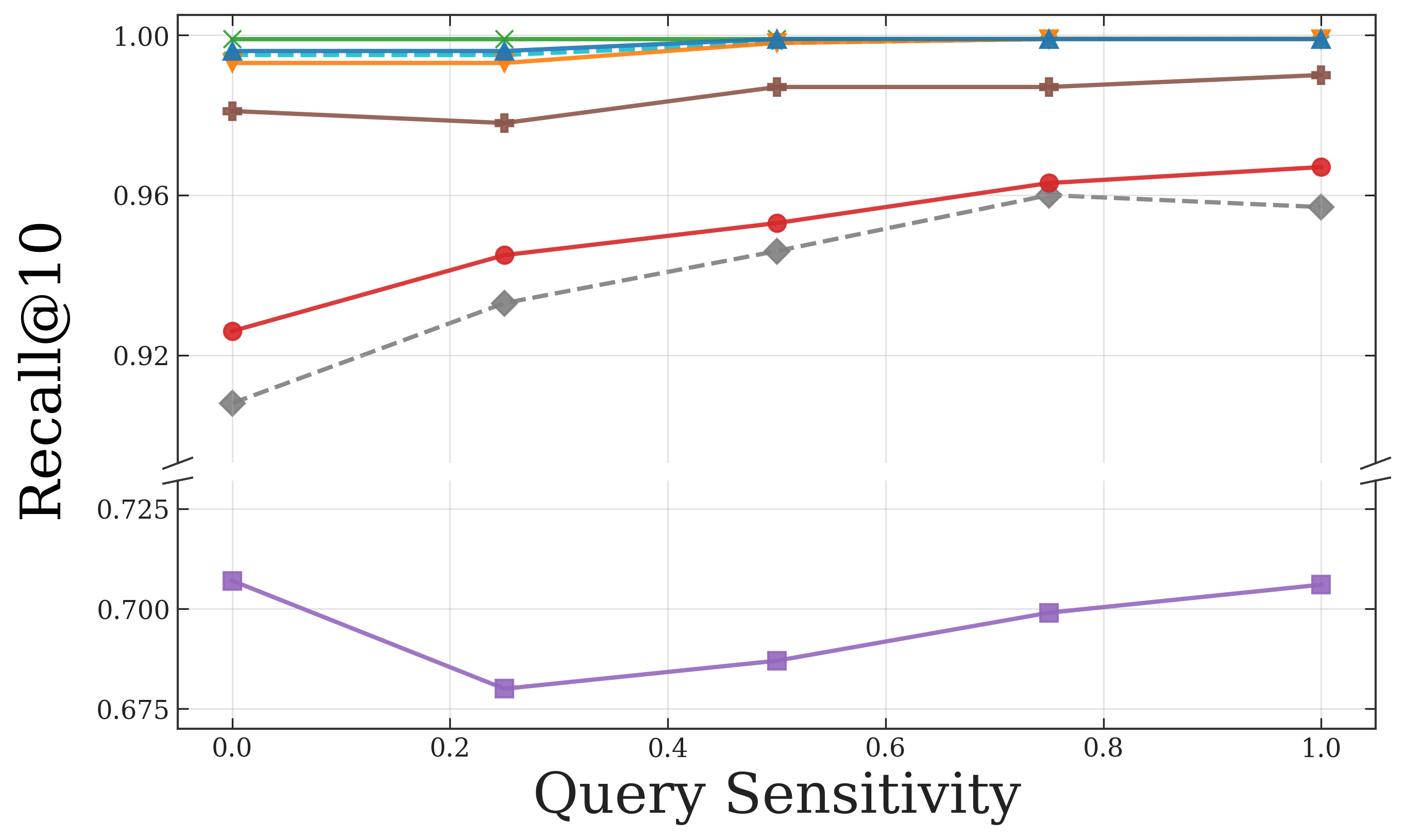}
        \caption{Sensitivity, AMZN.}
        \label{fig:recall_vs_sensitivity_amzn}
    \end{subfigure}
    \caption{Supplementary results on PAPER and AMZN: (a)--(b) show QPS vs.\ \textsf{efs} for Exp~\ref{exp:efs}, and (c)--(d) show Recall@10 vs.\ query sensitivity for Exp~\ref{exp:query-sensitivity}.}
    \label{fig:qps_vs_ef_additional}
    \label{fig:recall_vs_sensitivity_additional}
\end{figure*}

\subsection{Supplementary Results for Exp~\ref{exp:query-sensitivity}: Query Sensitivity}
\label{app:query-sensitivity-complementary}

Panels (c)--(d) of Figure~\ref{fig:recall_vs_sensitivity_additional}
report the query-sensitivity results on PAPER and AMZN. The results are
consistent with the SIFT-1M trend in
\S\ref{sec:evaluations}: \name, \effname, SIEVE, and Oracle maintain high recall
across all sensitivity levels, while global-graph filtering methods lose recall
because they do not reliably recover all authorized neighbors after filtering.
HoneyBee and the single global baseline also become less robust when queries are
drawn outside the issuing role's authorized set.

\subsection{Supplementary Results: Concrete Numbers for QPS}
\label{app:qps-vs-efs-numbers}

This subsection reports the QPS values under varying \textsf{efs} for
Exps~\ref{exp:qps-recall-sift}, \ref{exp:weighted-single-role},
and~\ref{exp:multi-role}. Tables~\ref{tab:qps-vs-efs-sift},
\ref{tab:qps-vs-efs-paper}, and~\ref{tab:qps-vs-efs-amzn} give the
uniform single-role results for Exp~\ref{exp:qps-recall-sift} on SIFT-1M,
PAPER, and AMZN, respectively. Table~\ref{tab:qps-vs-efs-single-role-weighted}
reports the weighted single-role results for
Exp~\ref{exp:weighted-single-role}. Tables~\ref{tab:qps-vs-efs-multi-role-avg}
and~\ref{tab:qps-vs-efs-multi-role-weighted} report the uniform and weighted
multi-role results for Exp~\ref{exp:multi-role} on SIFT-1M.

\begin{table}[t]
  \centering
  \footnotesize
  \caption{AVG. QPS by \textsf{efs} on SIFT-1M.}
  \begin{tabular}{l|c|c|c|c|c|c}
  \toprule
  \textbf{Method}\slash\textbf{\textsf{efs}} & \textbf{10} & \textbf{50} & \textbf{100} & \textbf{300} & \textbf{500} & \textbf{1000} \\
  \midrule
  HoneyBee & 187.492 & 176.164 & 162.605 & 130.234 & 112.222 & 85.636 \\
  \midrule
  SIEVE & 3119.077 & 1218.921 & 865.443 & 438.074 & 317.172 & 200.035 \\
  \midrule
  ACORN-1 & 5729.053 & 3062.951 & 2118.346 & 921.013 & 611.262 & 333.653 \\
  \midrule
  ACORN-$\gamma$ & 1560.614 & 1416.311 & 1289.799 & 966.778 & 755.451 & 479.835 \\
  \midrule
  Global & 82.484 & 33.225 & 24.282 & 13.799 & 11.043 & 8.525 \\
  \midrule
  Oracle & 19469.000 & 9054.963 & 5765.889 & 2761.838 & 1943.102 & 1210.747 \\
  \midrule
  \name & 5570.916 & 3067.795 & 2024.261 & 1010.634 & 716.655 & 451.970 \\
  \midrule
  \effname & 5086.119 & 2853.596 & 1913.186 & 1004.394 & 741.375 & 505.997 \\
  \bottomrule
  \end{tabular}\label{tab:qps-vs-efs-sift}
\end{table}


\begin{table}[t]
  \centering
  \footnotesize
  \caption{AVG. QPS by \textsf{efs} on Paper.}
  \begin{tabular}{l|c|c|c|c|c|c}
  \toprule
  \textbf{Method}\slash\textbf{\textsf{efs}} & \textbf{10} & \textbf{50} & \textbf{100} & \textbf{300} & \textbf{500} & \textbf{1000} \\
  \midrule
  HoneyBee & 117.963 & 78.603 & 59.744 & 40.985 & 35.648 & 29.621 \\
  \midrule
  SIEVE & 1944.825 & 708.676 & 491.717 & 249.205 & 207.699 & 126.251 \\
  \midrule
  ACORN-1 & 4141.419 & 1728.886 & 1104.888 & 496.706 & 321.846 & 173.284 \\
  \midrule
  ACORN-$\gamma$ & 296.255 & 292.565 & 285.247 & 264.597 & 245.668 & 210.790 \\
  \midrule
  Global & 13.497 & 5.466 & 4.520 & 3.694 & 3.493 & 3.282 \\
  \midrule
  Oracle & 16665.226 & 8699.385 & 5788.201 & 2868.663 & 2040.315 & 1274.185 \\
  \midrule
  \name & 6331.331 & 4093.082 & 3009.591 & 1628.945 & 1199.562 & 774.996 \\
  \midrule
  \effname & 5360.852 & 3317.271 & 2121.343 & 1184.398 & 871.909 & 563.141 \\
  \bottomrule
  \end{tabular}\label{tab:qps-vs-efs-paper}
\end{table}


\begin{table}[t]
  \centering
  \footnotesize
  \caption{AVG. QPS by \textsf{efs} on AMZN.}
  \begin{tabular}{l|c|c|c|c|c|c}
  \toprule
  \textbf{Method}\slash\textbf{\textsf{efs}} & \textbf{10} & \textbf{50} & \textbf{100} & \textbf{300} & \textbf{500} & \textbf{1000}  \\
  \midrule
  HoneyBee & 565.157 & 444.717 & 373.449 & 262.600 & 210.475 & 155.598 \\
  \midrule
  SIEVE & 2492.058 & 1350.972 & 1076.478 & 750.226 & 610.329 & 461.493 \\
  \midrule
  ACORN-1 & 4680.636 & 1889.237 & 1301.074 & 599.429 & 409.007 & 233.500 \\
  \midrule
  ACORN-$\gamma$ & 2287.439 & 1784.572 & 1502.050 & 1078.822 & 837.557 & 560.395 \\
  \midrule
  Global & 91.025 & 44.339 & 32.982 & 21.737 & 18.684 & 15.925 \\
  \midrule
  Oracle & 10702.294 & 4632.108 & 3049.857 & 1542.105 & 1109.317 & 703.179 \\
  \midrule
  \name & 10311.192 & 4970.202 & 3157.867 & 1474.527 & 995.525 & 611.028 \\
  \midrule
  \effname & 9583.520 & 4688.725 & 3089.447 & 1515.846 & 1052.995 & 656.514 \\
  \bottomrule
  \end{tabular}\label{tab:qps-vs-efs-amzn}
\end{table}


\begin{table}[t]
  \centering
  \footnotesize
  \caption{AVG. QPS by \textsf{efs} for single-role weighted queries on SIFT-1M.}
  \begin{tabular}{l|c|c|c|c|c|c}
  \toprule
  \textbf{Method}\slash\textbf{\textsf{efs}} & \textbf{10} & \textbf{50} & \textbf{100} & \textbf{300} & \textbf{500} & \textbf{1000} \\
  \midrule
  HoneyBee & 97.382 & 93.386 & 89.888 & 81.455 & 75.069 & 64.803 \\
  \midrule
  SIEVE & 3501.189 & 1233.381 & 825.063 & 430.808 & 294.577 & 158.195 \\
  \midrule
  ACORN-1 & 3709.269 & 2343.051 & 1710.205 & 825.065 & 556.623 & 292.556 \\
  \midrule
  ACORN-$\gamma$ & 2452.853 & 1937.127 & 1635.744 & 974.125 & 716.507 & 400.558 \\
  \midrule
  Global & 98.307 & 95.057 & 89.066 & 81.430 & 74.872 & 63.743 \\
  \midrule
  Oracle & 14325.983 & 6889.699 & 4199.032 & 1814.939 & 1279.416 & 760.606 \\
  \midrule
  \name & 4402.896 & 2235.286 & 1462.140 & 700.332 & 494.311 & 314.068 \\
  \midrule
  \effname & 3407.543 & 1872.680 & 1249.333 & 647.557 & 473.391 & 311.541 \\
  \bottomrule
  \end{tabular}\label{tab:qps-vs-efs-single-role-weighted}
\end{table}

\begin{table}[t]
  \centering
  \footnotesize
  \caption{AVG. QPS by \textsf{efs} for multi-role average queries on SIFT-1M.}
  \begin{tabular}{l|c|c|c|c|c|c}
  \toprule
  \textbf{Method}\slash\textsf{efs} & 10 & 50 & 100 & 300 & 500 & 1000 \\
  \midrule
  HoneyBee & 45.221 & 43.903 & 42.590 & 39.228 & 36.986 & 32.613 \\
  \midrule
  SIEVE & 6319.986 & 2896.698 & 1983.929 & 899.122 & 649.306 & 399.583 \\
  \midrule
  ACORN-1 & 5474.367 & 3695.168 & 2461.091 & 1175.011 & 764.182 & 406.920 \\
  \midrule
  ACORN-$\gamma$ & 6812.420 & 4524.342 & 3056.299 & 1508.263 & 979.229 & 512.690 \\
  \midrule
  Global & 45.220 & 44.681 & 44.260 & 42.995 & 42.418 & 41.746 \\
  \midrule
  Oracle & 11122.613 & 4797.081 & 2903.125 & 1018.783 & 416.122 & 38.906 \\
  \midrule
  \name & 6189.466 & 3175.273 & 2064.598 & 994.156 & 702.501 & 434.872 \\
  \midrule
  \effname & 5175.412 & 2818.353 & 1856.056 & 934.938 & 663.272 & 415.225 \\
  \bottomrule
  \end{tabular}\label{tab:qps-vs-efs-multi-role-avg}
\end{table}

\begin{table}[t]
  \centering
  \footnotesize
  \caption{AVG. QPS by \textsf{efs} for multi-role weighted queries on SIFT-1M.}
  \begin{tabular}{l|c|c|c|c|c|c}
  \toprule
  \textbf{Method}\slash\textsf{efs} & 10 & 50 & 100 & 300 & 500 & 1000 \\
  \midrule
  HoneyBee & 41.873 & 41.194 & 40.629 & 38.697 & 37.503 & 34.968 \\
  \midrule
  SIEVE & 5645.932 & 3113.955 & 2096.044 & 1021.345 & 736.921 & 469.705 \\
  \midrule
  ACORN-1 & 3776.462 & 2499.611 & 1877.226 & 1022.629 & 706.584 & 383.792 \\
  \midrule
  ACORN-$\gamma$ & 3640.366 & 2837.153 & 2601.271 & 1321.165 & 918.777 & 491.827 \\
  \midrule
  Global & 43.846 & 43.315 & 43.051 & 41.832 & 41.368 & 39.981 \\
  \midrule
  Oracle & 10939.646 & 4765.773 & 3006.480 & 1396.386 & 961.537 & 582.810 \\
  \midrule
  \name & 10157.987 & 4947.543 & 3117.508 & 1455.797 & 1022.009 & 616.855 \\
  \midrule
  \effname & 9346.597 & 5023.792 & 3239.555 & 1504.310 & 1056.843 & 636.080 \\
  \bottomrule
  \end{tabular}\label{tab:qps-vs-efs-multi-role-weighted}
\end{table}

\section{Supporting Dynamic Workloads}
\label{app:dynamic-workloads}

The construction in \S\ref{sec:veda}--\S\ref{sec:eff_veda} assumes a
static dataset and policy. We sketch here how the lattice absorbs both
data updates and policy changes, including permission grants/revocations
and role merges, without a full rebuild.

\paragraph{Data updates.}
Every vector $v$ belongs to exactly one exclusive block
$N^{\mathit{ex}}(\tau)$, and the container map $\Phi$
(\S\ref{sec:qplan}) records the set of lattice nodes that physically hold
that block. $\textsf{Insert}(v,\tau)$ adds $v$ to each node in
$\Phi(N^{\mathit{ex}}(\tau))$; $\textsf{Delete}(v)$ removes it from the
same set; $\textsf{Update}(v,v')$ is a delete followed by an insert.
Indexed nodes use HNSW's native incremental insertion~\cite{hnsw} and
tombstone deletion; leftover nodes are arrays with $\mathcal{O}(1)$
update.

\paragraph{Policy updates.}
Granting or revoking role $r$ on $v$ moves $v$ from block $\tau$ to
$\tau\cup\{r\}$ or $\tau\setminus\{r\}$; only nodes in the symmetric
difference $\Phi(N^{\mathit{ex}}(\tau))\,\triangle\,\Phi(N^{\mathit{ex}}(\tau'))$ are touched, and a
previously unseen destination block is created as a fresh leftover node.
Adding a role is metadata only. Dropping role $r$ relabels every block
$\tau\ni r$ to $\tau\setminus\{r\}$ and merges blocks that collide; node
contents are unchanged, so no HNSW index is rebuilt.
Administrative changes can be expressed by the same relabeling operation.
For example, if two departments are merged, the corresponding roles are
replaced by the merged role, and any blocks with identical resulting role
sets are coalesced.

\paragraph{Optimization.}
These operations keep the layout \emph{correct}. Every authorized vector
stays reachable through some node in $\textsf{QP}(r)$, and coordinated
search post-filters any newly impure node---but let it drift from the
\textsf{QA} optimum as block sizes and impurities change. We restore
optimality lazily: when a node crosses the threshold $\Lambda$ or its
size/impurity drifts beyond a slack, we re-run copy/merge on that node and
its lattice neighbors only. Only after large policy changes (e.g.,
$\textsf{DropRole}$), we re-run \effname in full.

\section{Detailed Related Work}
\label{app:detailed_related_work}

This appendix expands the qualitative comparison of \S\ref{sec:related_work}
into a method-level survey. The unifying question is \emph{where} a system pays
for the access predicate: before ANN search (pre-filtering), during graph
traversal (in-index filtering), through additional materialized sub-indexes
(workload-aware selection), or through an access-control-aware partitioning of
the data itself. We organize the discussion along that axis and, for each
method, state how \name and \effname differ.

\subsection{Access-Control Foundations}

Classical database access control defines policy models such as DAC, MAC, and
RBAC and studies their enforcement, administration, and auditing in relational
systems~\cite{bertino2011access}. That literature fixes the security semantics
a vector store must satisfy but says nothing about ANN-specific concerns:
approximate traversal, candidate-set inflation, or coordination across several
indices. Recent benchmarks for organization-scale retrieval-augmented
generation~\cite{lewis2020retrieval}, e.g., OrgAccess~\cite{sanyal2025orgaccess},
confirm that policy-compliant retrieval over embedding stores is now a
first-class requirement and motivate the problem we study.

\subsection{General-Purpose ANN Indexing}

Unconstrained ANN indices are the primitives every system below builds on.
FAISS~\cite{faiss} popularized partition- and compression-based indices (IVF,
product quantization) that assign vectors to coarse clusters and probe a subset
of inverted lists. HNSW~\cite{hnsw} builds a multi-layer navigable small-world
graph and answers queries by greedy descent from sparse upper layers to a dense
base layer. DiskANN~\cite{diskann} introduces the Vamana graph for high-recall
billion-scale search with most of the index resident on SSD.
ScaNN~\cite{scann} combines partition pruning with anisotropic vector
quantization for large-scale maximum-inner-product search. All four assume a
single global search space; they are orthogonal to access control and appear in
\name only as the per-group index implementation.

\subsection{Filtered and Category-Aware Vector Search}

Filtered ANN attaches attributes to vectors and returns nearest neighbors that
satisfy a predicate. The generic strategies---\emph{pre-filtering} (materialize
the matching set, then search it), \emph{post-filtering} (search globally, then
discard non-matching results), and \emph{result-set filtering} (test the
predicate before admitting a candidate to the top-$k$)---are easy to layer onto
any index but fail in opposite regimes: pre-filtering degenerates to a scan
when the predicate is broad, and post-filtering must over-search heavily when
the predicate is selective. The systems below each move the predicate deeper
into the index to avoid that ``unhappy middle.''

\noindent\textbf{Filtered-DiskANN}~\cite{filtered_diskann} makes the Vamana
graph label-aware. \emph{FilteredVamana} inserts points incrementally and
prunes edges using both geometry and label overlap; \emph{StitchedVamana}
builds one Vamana graph per label, overlays their edge sets into a single
graph, and prunes high-degree nodes. The result remains a single SSD-friendly
graph, so storage stays near \textsf{SA}${\approx}1$, but the relevant label
universe must be fixed at construction time. RBAC roles satisfy that
assumption, yet a role's visible set is the \emph{union} of many label-specific
subgraphs; Filtered-DiskANN has no mechanism to stitch those subgraphs at query
time, so a role with broad access still falls back to over-search on the
overlaid graph. \name avoids this by materializing the union explicitly as a
group when the storage budget allows.

\noindent\textbf{UNG}~\cite{ung} augments a graph index with a separate
\emph{label-navigating graph} that encodes containment among label sets, so
equality, subset, and overlap predicates can prune vectors that cannot match
before any distance is computed. The predicate model is richer than
Filtered-DiskANN's but remains a category-search model: it accelerates ``find
vectors whose labels contain $L$,'' not ``find vectors visible to role $r$
across all of $r$'s permissions.'' \name's role-subset lattice plays the
analogous navigational role for RBAC, but its nodes are physical index groups
rather than logical label sets.

\noindent\textbf{ACORN}~\cite{acorn} targets predicate-\emph{agnostic} hybrid
search on HNSW. For a query predicate $p$, search is confined to the subgraph
induced by vectors satisfying $p$; because filtering can disconnect that
subgraph, ACORN-$\gamma$ widens neighbor lists by a factor $\gamma$ at
construction time (with predicate-agnostic pruning), and ACORN-$1$ keeps
construction close to vanilla HNSW but expands one- and two-hop neighbors at
query time. Both variants restore reachability by exposing a denser
neighborhood before the predicate is applied. ACORN handles arbitrary
predicates and therefore cannot exploit that RBAC predicates are few, fixed,
and heavily overlapping; in our setting its $\gamma$-expansion pays a uniform
construction tax for flexibility \name does not need.

\noindent\textbf{SIEVE}~\cite{sieve} is the filtered-ANN system closest to our
optimization. Rather than modify traversal, it \emph{selects} a workload-aware
collection of HNSW sub-indexes under a memory budget and routes each query to
one of them. For a sub-index $I_h$ whose filter $h$ subsumes a query filter
$f$, SIEVE models memory as
\[
S(I_h)=M\cdot \operatorname{card}(h),
\]
and indexed search cost as
\[
C(I_h,\mathit{sef},w,f)
= \log(\operatorname{card}(h))\cdot \mathit{sef}\cdot
\left(\frac{\operatorname{card}(h)}{\operatorname{card}(f)}\right)^{
\operatorname{cor}(w,f,h)} ,
\]
where $\operatorname{card}(\cdot)$ is filter cardinality, $\mathit{sef}$ is the
HNSW search parameter, and $\operatorname{cor}(w,f,h)$ captures query--filter
correlation inside $I_h$. Brute-force search costs
$C_{bf}(f)=\operatorname{card}(f)$ up to a calibration constant, and a built
collection $\mathcal{I}$ serves $f$ at
\[
C(\mathcal{I},f)
= \min\!\left(C_{bf}(f),\;\min_{I_h\in\mathcal{I}} C(I_h,f)\right).
\]
Given a historical workload $\mathcal{H}=\{(h_i,c_i)\}$ and budget $B$, SIEVE
solves
\[
\min_{\mathcal{I}}\;\sum_{(h_i,c_i)\in\mathcal{H}} c_i\,C(\mathcal{I},h_i)
\quad\text{s.t.}\quad
I_{\infty}\in\mathcal{I},\;
\sum_{I_h\in\mathcal{I}} S(I_h)\leq B
\]
greedily, adding at each step the sub-index with the largest marginal
cost-reduction per memory unit, and organizes the result into a Hasse diagram
over filter subsumption so that a BFS from the root $I_{\infty}$ finds the
smallest subsuming sub-index (any node whose filter does not subsume $f$ prunes
its entire subtree). Two design choices separate SIEVE from \name. First, SIEVE
serves each query from \emph{one} sub-index; an RBAC role whose visible set is
the union of several disjoint blocks must either fall back to $I_{\infty}$ or
accept an impure superset. \effname instead probes multiple authorized indices
and shares a global distance bound across them. Second, SIEVE's search space is
``which sub-indexes to add''; \name's is ``which blocks to merge and which to
copy,'' which can also \emph{reduce} the number of indices a role touches
without spending budget.

\subsection{Access-Control-Aware Vector Search}

Production deployments typically sit at one of two extremes. A \emph{global
index with row-level filtering} stores every vector once
(\textsf{SA}${=}1$) and enforces access control by post-filtering, so queries
spend most of their beam on unauthorized vectors whenever the role is
selective. An \emph{oracle index} builds the ideal pure HNSW graph for each
access predicate known at construction time, e.g., a role or user visibility
set; this makes search trivial but duplicates every shared vector and becomes
unmanageable when permissions overlap heavily. \name treats these as the two
endpoints of the \textsf{SA}/\textsf{QA} curve
(\S\ref{sec:ac-indexing}) and searches the interior.

\noindent\textbf{HoneyBee}~\cite{honeybee} is the closest prior work. It casts
RBAC vector search as a dynamic partitioning problem: overlapping partitions
replicate vectors selectively, using the ``thin waist'' of RBAC roles to
balance storage, latency, and recall. Its per-user latency model is
\[
C_u(\Pi,u_i,\mathit{ef}_s)
= \sum_{j\in \mathrm{AP}_{\min}(u_i,\Pi)}
\log(|\pi_j|)\cdot(a\cdot \mathit{ef}_s+b),
\]
where $\Pi$ is the partitioning, $\mathrm{AP}_{\min}(u_i,\Pi)$ is the minimal
set of partitions that cover $u_i$'s visible data, and $\mathit{ef}_s$ is the
HNSW search parameter. HoneyBee pairs this with a selectivity-based recall
model, formulates partitioning as a constrained MINLP, and---because the exact
problem is NP-hard---solves it with a greedy \emph{split} heuristic that
repeatedly divides role groups when the predicted latency gain justifies the
storage cost.

\name shares HoneyBee's objective but differs on three axes.
(\emph{i})~\emph{Search space.} HoneyBee moves top-down by splitting; \name
moves over the role-subset lattice with both \emph{copy} and \emph{merge}, so
it can reach layouts---e.g., two blocks merged into one pure group for a
frequent role pair---that no sequence of splits produces.
(\emph{ii})~\emph{Cost model.} HoneyBee charges $\log|\pi_j|\cdot(a\,\mathit{ef}_s+b)$
per partition; \name uses the hardware-calibrated $C_{\theta}$ of
Definition~\ref{def:hnsw-cost}, whose linear-$\textsf{efs}$ term is the quantity
the optimizer actually manipulates through impurity $\lambda$.
(\emph{iii})~\emph{Execution.} HoneyBee searches each partition in
$\mathrm{AP}_{\min}$ independently. \effname runs them in a coordinated
schedule in which pure partitions are probed first and their top-$k$ distances
become an admission threshold for the beam on impure partitions and residual
linear scans, so the inflated $\lambda\,\textsf{efs}$ of
Definition~\ref{def:hnsw-cost} is paid only on the fraction of impure
candidates that survive the bound.

\end{document}